\documentclass[conference]{IEEEtran}
\ifCLASSINFOpdf
\else
\fi
\usepackage{url}
\usepackage{etoolbox}
\usepackage[T1]{fontenc}
\usepackage{mathrsfs}
\usepackage{etoolbox}
\usepackage{upgreek}
\usepackage{enumitem}
\usepackage{stmaryrd}
\usepackage{graphicx}
\usepackage{xspace,xcolor}
\usepackage{amsmath,amsfonts,amssymb,mathtools}
\usepackage{bbold}

\usepackage{rsfso}
\usepackage{scalerel}
\usepackage{hyperref}
\usepackage{mathpartir}
\usepackage{bussproofs}
\usepackage{lineno}

\usepackage{ellipsis}
\usepackage{placeins} % FloatBarrier

\usepackage{tikz}
\usepackage{tikz-cd}
\usetikzlibrary{matrix}
\usetikzlibrary{arrows}
\usetikzlibrary{arrows.meta}
\usetikzlibrary{decorations.pathmorphing,shapes}
\usetikzlibrary{decorations.markings}

\tikzset{spanmap/.style={
            decoration={markings,
            mark= at position 0.5 with {
                  \node[transform shape] (tempnode) {$|$};
                  }
              },
              postaction={decorate}
}
}

\def\proofskip{\vskip 4pt plus 1pt minus 1pt}
\def\proofbox{\hfill\rule{6pt}{6pt}}
\newenvironment{proof}{{\noindent \sc Proof\ }~}{\proofbox\proofskip}

\definecolor{red}{rgb}{.8,0.2,0.2}
\definecolor{rouge}{rgb}{.8,0,0}
\definecolor{green}{rgb}{0,.5,0}
\definecolor{antigreen}{rgb}{.8,0,.8}
\definecolor{blue}{rgb}{0,0,1}
\definecolor{littleblue}{rgb}{.2,.2,.6}
\definecolor{graille}{rgb}{.01,.01,.00}
\definecolor{grispale}{rgb}{.6,.6,.6}
\definecolor{yellow}{rgb}{.5,.5,.2}
\definecolor{cyan}{rgb}{0.75,1,1}

\definecolor{strongred}{rgb}{.7,0,0}
\definecolor{strongblue}{rgb}{0,0,.7}
\newcommand{\blackhole}[1]{}

\newcommand{\stratsource}{s}
\newcommand{\strattarget}{t}

\newcommand{\source}[1]{\sourcefn(#1)}

\newcommand{\target}[1]{\targetfn(#1)}

\newcommand{\transitionpath}{\twoheadrightarrow}

\newcommand{\anchor}{\moo}

\newcommand{\Games}{\mathbf{Games}}
\newcommand{\Span}[1]{\mathbf{Span}(#1)}

\newcommand{\Agame}{A}
\newcommand{\Bgame}{B}
\newcommand{\Cgame}{C}

\newcommand{\Sgame}{S}
\newcommand{\Tgame}{T}

\newcommand{\lrangle}[1]{\begin{footnotesize}\langle{#1}\rangle\end{footnotesize}}

\newcommand{\anchorofgames}{\anchor_{\mathrm{game}}}
\newcommand{\anchorofstrat}{\anchor_{\mathrm{strat}}}

\newcommand{\anchorofconcgames}{\anchor_{\mathrm{conc}}}

\newcommand{\anchorofasynch}{\anchor_{\mathrm{asynch}}}

\newcommand{\SpanGames}[1]{\mathbf{Games}(#1)}

\newcommand{\identityspan}[1]{\mathbf{id}_{#1}}
\newcommand{\interaction}[1]{\lambda_{#1}}

\newcommand{\polarityplus}{\textcolor{strongblue}{\emph{\textbf{P}}}}
\newcommand{\polarityminus}{\textcolor{strongred}{\emph{\textbf{O}}}}
\newcommand{\polarityplussource}{\textcolor{strongblue}{\emph{\textbf{P}}_{s}}}
\newcommand{\polarityminussource}{\textcolor{strongred}{\emph{\textbf{O}}_{s}}}
\newcommand{\polarityplustarget}{\textcolor{strongblue}{\emph{\textbf{P}}_{t}}}
\newcommand{\polarityminustarget}{\textcolor{strongred}{\emph{\textbf{O}}_{t}}}
\newcommand{\polarityplusmid}{\emph{\textbf{OP}}}
\newcommand{\polarityminusmid}{\emph{\textbf{PO}}}
\newcommand{\polarityplusmidone}{\emph{\textbf{OP}}_{s}}
\newcommand{\polarityminusmidone}{\emph{\textbf{PO}}_{s}}
\newcommand{\polarityplusmidtwo}{\emph{\textbf{OP}}_{t}}
\newcommand{\polarityminusmidtwo}{\emph{\textbf{PO}}_{t}}
\newcommand{\polarityminusone}{\textcolor{strongred}{\emph{\textbf{O}}_{1}}}
\newcommand{\polarityplusone}{\textcolor{strongblue}{\emph{\textbf{P}}_{1}}}
\newcommand{\polarityminustwo}{\textcolor{strongred}{\emph{\textbf{O}}_{2}}}
\newcommand{\polarityplustwo}{\textcolor{strongblue}{\emph{\textbf{P}}_{2}}}
\newcommand{\polarityminusthree}{\textcolor{strongred}{\emph{\textbf{O}}_{3}}}
\newcommand{\polarityplusthree}{\textcolor{strongblue}{\emph{\textbf{P}}_{3}}}
\newcommand{\tensor}{\otimes}

\newcommand{\Acategory}{\mathscr{A}}
\newcommand{\Bcategory}{\mathscr{B}}
\newcommand{\Ccategory}{\mathscr{C}}
\newcommand{\Dcategory}{\mathscr{D}}

\newcommand{\Ecategory}{\mathscr{E}}

\newcommand{\Scategorysurround}{\mathbb{S}}

\newcommand{\Vbicategory}{\mathcal{V}}

\newcommand{\Xcategory}{\mathscr{X}}

\newcommand{\Set}{\textbf{Set}}
\newcommand{\Cat}{\textbf{Cat}}

\newcommand{\adjunction}[4]{{#3}:\vcenter{\xymatrix @-.1pc {{#1}\ar@/^.7pc/[rr]\ar@{}[rr]|-{\bot}&&{#2}\ar@/^.7pc/[ll]}}:{#4}}

\newcommand{\Asynch}{\textbf{Asynch}}

\newcommand{\id}[1]{\mathrm{id}_{#1}}
\usepackage{mathtools}
\usepackage{cuted}
\usepackage[utf8x]{inputenc}

\usepackage{shuffle}

\newcommand{\tic}{\mathsf{tic}}

\renewcommand{\source}[1]{\textit{source}_{#1}}
\renewcommand{\target}[1]{\textit{target}_{#1}}
\newcommand{\spanmap}{\longrightarrow\hspace{-1.45em}|\hspace{1.25em}}

\newcommand{\Label}{\mathcal{L}}

\newcommand{\asynchanchorof}[1]{\moo[\hspace{.1em}#1\hspace{.1em}]}
\newcommand{\onecatanchorof}[1]{\moo\{\hspace{-.19em}|\hspace{.1em}#1\hspace{.1em}|\hspace{-.19em}\}}
\newcommand{\twocatanchorof}[1]{\moo\{\hspace{-.19em}|\hspace{.1em}#1\hspace{.1em}|\hspace{-.19em}\}}

\newcommand{\TwoCat}{\mathbf{2}\textbf{-}\textbf{Cat}}
\newcommand{\graytensor}{\boxtimes}
\newcommand{\grayunit}{\textbf{1}}
\newcommand{\graytensoreq}[1]{\boxtimes_{#1}}

\newcommand{\funnytensor}{\hspace{.06em}\oblong\hspace{.06em}}
\newcommand{\coaction}[1]{\mathsf{coact}_{#1}}
\newcommand{\coactionleft}[1]{\mathsf{coact}^{\textrm{left}}_{#1}}
\newcommand{\coactionright}[1]{\mathsf{coact}^{\textrm{right}}_{#1}}
\newcommand{\comult}[1]{d_{#1}}
\newcommand{\counit}[1]{e_{#1}}
\newcommand{\fsubA}{a}
\newcommand{\fsubAprime}{a'}
\newcommand{\fsubAsecond}{a''}
\newcommand{\fsubB}{b}
\newcommand{\fsubBprime}{b'}
\newcommand{\fsubBsecond}{b''}
\newcommand{\fsubC}{c}

\newcommand{\horcomp}{\ast^h}

\newcommand{\vertcomp}{\ast^v}
\newcommand{\horid}[1]{id^{\hspace{.05em}h}_{#1}}
\newcommand{\vertid}[1]{id^{\hspace{.05em}v}_{#1}}
\newcommand{\graytile}[2]{\gamma_{#1,#2}}
\newcommand{\graytiletilde}[2]{\tilde{\gamma}_{#1,#2}}
\newcommand{\asynchtwocat}[1]{\langle\hspace{.1em}#1\hspace{.1em}\rangle}
\newcommand{\asynchsesquicat}[1]{\langle\!\langle\hspace{.1em}#1\hspace{.1em}\rangle\!\rangle}
\newcommand{\Comod}[1]{\textbf{Comod}(#1)}

\newcommand{\lengthindex}[1]{[#1]}
\newcommand{\permutationcategory}[3]{\textbf{Perm}(#2,#3)}

\newcommand{\Tran}{\mathsf{Tran}}

\newcommand{\shuffletensor}{\shuffle}

\newcommand{\Atwocategory}{\mathscr{A}}
\newcommand{\Btwocategory}{\mathscr{B}}
\newcommand{\Ctwocategory}{\mathscr{C}}
\newcommand{\Etwocategory}{\mathscr{E}}

\newcommand{\underlyingcat}[1]{|#1|}
\newcommand{\underlyingset}[1]{||#1||}

\tikzset{spanmap/.style={
            decoration={markings,
            mark= at position 0.5 with {
                  \node[transform shape] (tempnode) {$|$};
                  }
              },
              postaction={decorate}
}
}

\tikzset{doublespanmap/.style={
            decoration={markings,
            mark= at position 0.5 with {
                  \node[transform shape] (tempnode) {$||$};
                  }
              },
              postaction={decorate}
}
}

\newcommand{\Comonoid}[1]{\textbf{Comonoid}(#1)}

\newtheorem{definition}{Definition}
\newtheorem{proposition}{Proposition}

\newtheorem{theorem}{Theorem}

\begin{document}
%\title{Asynchronous Template Games\\ and the Gray tensor product of 2-categories}
\title{Asynchronous Template Games\\
and the Gray tensor product of 2-categories}
%\vspace{-1.2em}}

%\author{\IEEEauthorblockN{Anonymous Unlimited}}
\author{\IEEEauthorblockN{Paul-Andr\'e Melli\`es}
\IEEEauthorblockA{Institut de Recherche en Informatique Fondamentale (IRIF)\\
CNRS, Universit\'e de Paris, France\\
Email: mellies@irif.fr}}

%Telephone: (800) 555--1212, Fax: (888) 555--1212}


\maketitle

\begin{abstract}
In his recent and exploratory work on template games and linear logic,
Melli{\`e}s defines sequential and concurrent games as categories
with positions as objects and trajectories as morphisms,
labelled by a specific synchronization template.
%describing how games and strategies are scheduled and interact.
%
In the present paper, we bring the idea one dimension higher
and advocate that 
%in order to reflect properly the concurrent and asynchronous nature of linear logic, 
template games should not be just defined
as 1-dimensional categories but as 2-dimensional categories
of positions, trajectories and reshufflings (or reschedulings) as 2-cells.
%reorderings.
%
In order to achieve the purpose, we take seriously 
the parallel between asynchrony in concurrency
and the Gray tensor product of 2-categories.
One technical difficulty on the way is that the category $\Scategorysurround=\TwoCat$
of small 2-categories equipped with the Gray tensor product is monoidal, and not cartesian.
This prompts us to extend the framework of template games originally formulated
by Melli{\`e}s in a category~$\Scategorysurround$ with finite limits,
%based on internal categories 
%and functorial spans, 
%and adapt it 
%where every concurrent game is a Gray comonoid
%and every concurrent strategy is a Gray comodule.
and to upgrade it in the style of Aguiar's work on quantum groups
to the more general situation of a monoidal category~$\Scategorysurround$
with coreflexive equalizers, preserved by the tensor product componentwise.
%, preserved by the tensor product.
%
We construct in this way an asynchronous template game semantics
of multiplicative additive linear logic (MALL) where every formula 
%$A$ 
and every proof
%$\pi$ 
is interpreted as a labelled 2-category
% of positions, trajectories and permutations
equipped, respectively, 
with the structure of Gray comonoid for asynchronous template games,
%$[A]$
and of Gray bicomodule for asynchronous strategies.
%, reflecting the polarities of the moves in the game.
%$[\pi]$
%in the category $\Scategorysurround=\TwoCat$ of 2-categories.
% the weak double category of spans is replaced by a double category of comodules.
%between comonoids.
%
%Somewhat surprisingly, it appears that every concurrent game defines a Gray comonoid,
%and every strategy as a Gray comodule on that comonoid.
%and where every game is a comonoid and every strategy is a comodule.
%and based on comodules on comonoids.
%and to define games, strategies and simulations
%in a coalgebraic style inspired by quantum algebra.
%
%The 2-categorical framework of template games also enables us
%to establish that composition preserves bisimulation of strategies, 
%using the correspondence between open maps and fibrations.
\end{abstract}

% \fi
\IEEEpeerreviewmaketitle

%\begin{verse}
%Certes il est r\'econfortant de d\'ecouvrir que REIRCNE est  un ENCRIER renvers\'e, mais ce n'est qu'au prix de douloureuses 
%(et parfois complaisantes) contorsions que l'on arrivera \`a faire semblant de d\'efinir des choses comme ATLE ou CCS.
%Georges Perec, \emph{Consid\'erations de l'auteur sur l'art et la mani\`ere de croiser les mots.}
%\end{verse}

%\begin{quote}
%Il faut aussi que tu n'ailles point
%\\
%Choisir tes mots sans quelque m\'eprise:
%\\
%Rien de plus cher que la chanson grise
%\\
%O\`u l'Ind\'ecis au Pr\'ecis se joint.
%\end{quote}

\section{Introduction and overview}\label{section/introduction}
Game semantics is the offspring of a fruitful encounter 
between concurrency theory, proof theory 
and programming language semantics.
Arising at this crossover position between the three fields,
game semantics inherits the rich toolbox of ideas and techniques
of concurrency theory, and provides a vivid landscape
of logical and programming phenomena to dissect and to interpret.
For that reason, it was understood at an early stage of game semantics 
\cite{abramsky-jagadeesan-malacaria,hyland-ong-full-abstraction}
that this nice hybridation at the heart of game semantics
would offer a fresh start and a number of new perspectives
on the mathematical foundations of concurrency theory.
%to revisit and improve the existing 
%pen new perspectives on concurrency theory,
%a perfect basis to reunderstand and to sharpen
%the very foundations of concurrency theory.
%
%The notion of asynchronous transition system.
%\paragraph{Labelled transition systems.}
%\paragraph{Concurrency theory.}
%Since the very origins of the field and 
Since the very origins of the field and the seminal contributions 
by Petri~\cite{petri-phd-1962,petri-asynchronous-1963}, 
Beki{\v{c}}~\cite{bekic-1971} and Mazurkiewicz~\cite{mazurkiewicz-concurrent-1977},
one main ambition of concurrency theory has been to give an abstract
and mathematically elegant account of the symbolic choreography at work 
in a distributed (software or hardware) system of intercommunicating components.
%
%paragraph
\medbreak
\noindent
\emph{\textbf{Labelled transition systems:}}
The description is typically based on the notion of transition system
$$\Tran=(G,\lambda)$$
on a specific set $\Label$ of labels,
defined as a graph $G=(V,E,\source{},\target{})$ 
consisting of a set~$V$ of vertices, a set~$E$ of edges,
and two functions
$\source{},\target{}:E\longrightarrow V$
equipped with a labelling function 
$\lambda:E\longrightarrow\Label.$
In this formulation, each state of the distributed system
is represented as a vertex $x,y\in V$ and every transition $u:x\to y$ 
is represented as an edge $u\in E$ labelled by an element $\lambda(u)\in\Label$.
One sometimes requires in the traditional definition
that there exists \emph{at most} one transition $u:x\to y$
with a given label $\ell\in\Label$ between two given states $x$, $y$;
one thus writes $u=(x,\ell,y)$ for the unique edge $u:x\to y$
with label $\lambda(u)=\ell$.
We do not make this simplifying assumption here,
see~\cite{HildebrandtSassone96} for a discussion.

%paragraph
\medbreak
\noindent
\emph{\textbf{Asynchronous transition systems:}}
In order to capture the concurrent nature of computations,
one equips the transition system $\Tran$ with a notion of \emph{independence}
between transitions $u$ and $v$ performed in separate components
of the distributed system, with no shared memory.
%
%This notion of independence reflects the fact that two transitions $u$ and $v$ 
%performed in separate components of the distributed system, 
%with no shared memory, should be considered as entirely independent for that reason.
%
Because the two transitions $u$ and $v$ are independent,
their order of execution does not matter.
%of the two transitions $u$ and $v$ 
%does not matter in this case.
%
% is captured by permuting transitions.
%
This idea is nicely captured mathematically by an intuition coming 
from algebraic topology and the connection between higher automata and cubical sets~\cite{GoubaultHeindelMimram13,dihomotopy-book}.
A \emph{square} in a graph~$G$ is defined as a pair $(p,q)$ 
of paths $p,q:x\transitionpath y$ of length 2, with same source $x$ and same target $y$, 
and thus of the form $p=u_1\cdot u_2$ and $q=v_1\cdot v_2$ as depicted below:
\begin{equation}\label{equation/square}
\raisebox{-3.8em}{\includegraphics[height=7.6em]{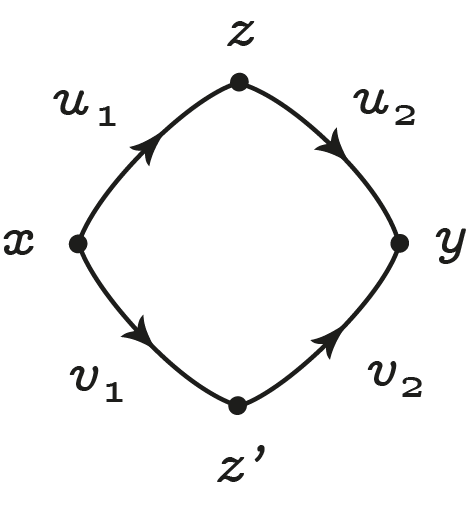}}
% \vspace{-0em}
\end{equation}
An \emph{asynchronous graph} $(G,\diamond)$ is defined as a graph~$G$ 
equipped with a set $\diamond$ of squares of $G$, satisfying three basic properties
detailed in \S\ref{section/asynchronous-graphs}.
We use the notation $p\diamond q$ when the square $(p,q)$
is an element of $\diamond$ and say in that case
that the square $(p,q)$ defines a \emph{permutation square.}
Following an intuition coming from algebraic topology,
a permutation tile $u_1\cdot u_2 \diamond v_1 \cdot v_2$
between the paths $p=u_1\cdot u_2$ and $q=v_1 \cdot v_2$
is depicted in the following way, as a 2-dimensional surface 
(or tile) between the paths~$p$ and~$q$:
 % \vspace{-1.2em}
\begin{equation}\label{equation/permutation-tile}
\raisebox{-3.8em}{\includegraphics[height=7.6em]{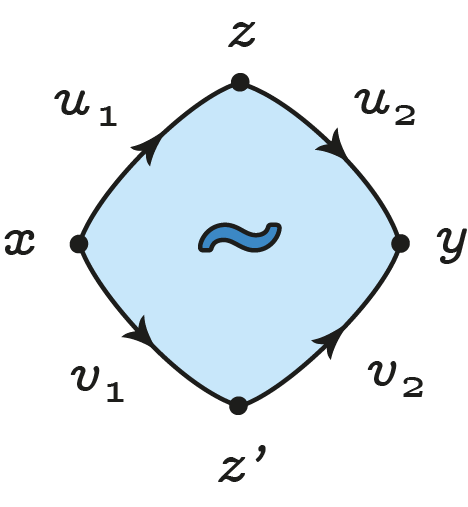}}
% \vspace{-0em}
\end{equation}
An asynchronous graph homomorphism
$$
f\quad : \quad (G,\diamond_G)\longrightarrow (H,\diamond_H)
$$
is defined as a graph homomorphism $f:G\to H$ with the additional property
that every permutation square in~$G$ is transported 
to a permutation square in $H$, in the sense that for every square $(p,q)$ in~$G$, 
$$
p\diamond_G q \quad \Rightarrow \quad f(p)\diamond_H f(q)
$$
where $f(p)=f(u_1)\cdot f(u_2)$ and $f(q)=f(v_1)\cdot f(v_2)$
are the paths of length 2 in the graph $H$
obtained as image of the paths $p=u_1\cdot u_2$ and $q=v_1\cdot v_2$.
This defines the category $\Asynch$ of asynchronous graphs
and asynchronous graph homomorphisms between them.

One nice and concise way to define a (labelled) asynchronous
transition system from there is to associate to every set $\Label$
the asynchronous graph noted 
\begin{equation}\label{equation/anchorof}
\asynchanchorof{\Label}=(\asynchanchorof{\Label},\diamond_{\Label})
\end{equation}
with a unique vertex noted $\ast$, one edge $\ell : \ast\to\ast$ for each label $\ell\in\Label$, 
and one permutation square 
\begin{equation}\label{equation/permutation-tile-core}
\raisebox{-4.2em}{\includegraphics[height=7.4em]{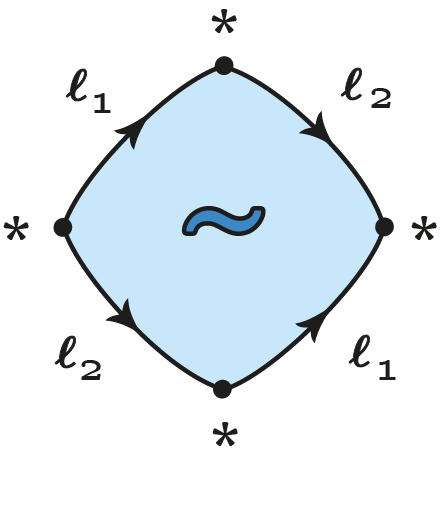}}
\vspace{-.4em}
\end{equation}
for every pair $(\ell_1,\ell_2)$ of (possibly equal) labels $\ell_1,\ell_2\in\Label$.
An asynchronous transition system is then defined as 
an asynchronous graph $(G,\diamond)$ equipped
with an asynchronous graph homomorphism
\begin{equation}\label{equation/transition-system-as-map}
%\lambda \quad : \quad (G,\diamond) \longrightarrow (\anchoroflabels{\Label},\diamond_{\Label})
\lambda \quad : \quad (G,\diamond) \longrightarrow \asynchanchorof{\Label}
\end{equation}
Note that the definition ensures that the labelling is asynchronous
in the sense that  $\lambda(u_1)=\lambda(v_2)$ and $\lambda(u_2)=\lambda(v_1)$
in any permutation square~(\ref{equation/permutation-tile-core})
of the asynchronous transition system~$(G,\diamond)$.
As we will see, this formulation~(\ref{equation/transition-system-as-map})
of asynchronous transition systems as a labelling map
in the category $\Asynch$ is particularly convenient.

%paragraph
\medbreak
\noindent
\emph{\textbf{The shuffle tensor product:}}
An important observation for this paper is that the category $\Asynch$ comes equipped 
with a symmetric monoidal structure, provided by the \emph{shuffle tensor product}
$$
(G,\diamond_G)\shuffletensor(H,\diamond_H)
\quad = \quad
(G\shuffletensor H,\diamond_{G\shuffletensor H})
$$
of two asynchronous graphs $(G,\diamond_G)$ and $(H,\diamond_H)$.
Here, ${G\shuffletensor H}$ denotes the graph
%of two asynchronous graphs $G=(G,\diamond_G)$
%and $H=(H,\diamond_H)$
%is the asynchronous graph 
%$G\shuffletensor H$
whose vertices $(x,y)$ are the pairs of a vertex $x$ in $G$
and of a vertex $y$ in $H$ and whose edges
$$
\begin{tikzcd}[column sep=1.5em, row sep=3em]
(x,y)\arrow[rr,"{(u,y)}","{}"{swap}] &&  (x',y)
\end{tikzcd}
\quad\quad
\begin{tikzcd}[column sep=1.5em, row sep=3em]
(x,y)\arrow[rr,"{(x,v)}","{}"{swap}] && (x,y')
\end{tikzcd}
$$
are either pairs $(u,y)$ consisting of an edge $u:x\to x'$
%$$
%\begin{tikzcd}[column sep=2em, row sep=3em]
%M\arrow[rr,"{u}","{}"{swap}] &&  M'
%\end{tikzcd}
%$$
in~${G}$ and of a vertex $y$ in~$H$ ;
or symmetrically, pairs $(x,v)$ consisting of a vertex $x$ in $G$ 
and of an edge $v:y\to y'$ in~$H$.
There is a permutation tile of $G\tensor H$ 
$$(u,y)\cdot (x',v)\diamond_{G\shuffletensor H} (x,v)\cdot (u,y')$$
and in the reverse direction 
$$(x,v)\cdot (u,y')\diamond_{G\shuffletensor H} (u,y)\cdot (x',v)$$
as depicted below in orange
\begin{equation}\label{equation/asynchronous-tensor-tiles-A}
\raisebox{-3.6em}{\includegraphics[height=7.4em]{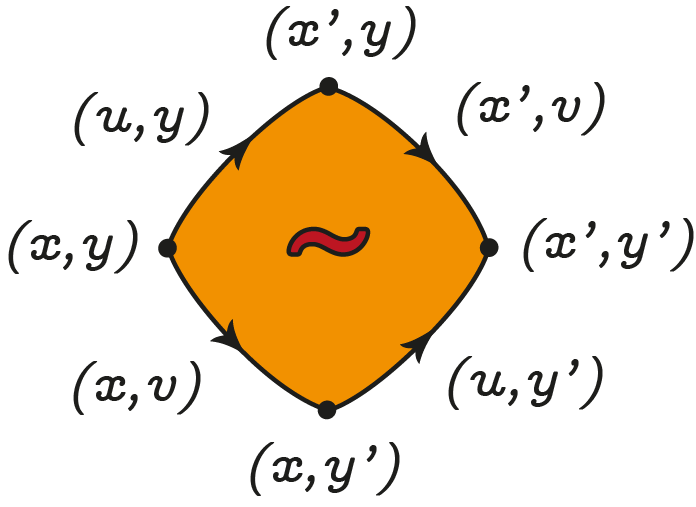}}
\quad\quad
\raisebox{-3.6em}{\includegraphics[height=7.4em]{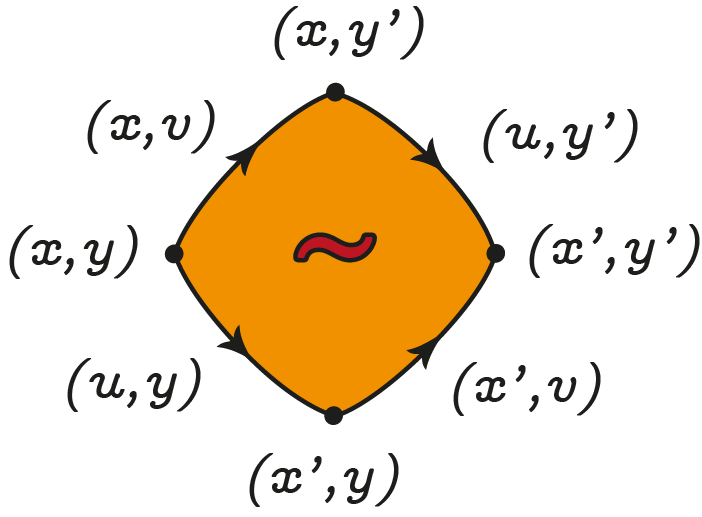}}
\vspace{-.2em}
\end{equation}
for every pair of edges $u:x\to x'$ in $G$ and $v:y\to y'$ in~$H$; 
a permutation tile (in orange on the right-hand side)
\begin{equation}\label{equation/asynchronous-tensor-tiles-B}
\raisebox{-3.6em}{\includegraphics[height=7.4em]{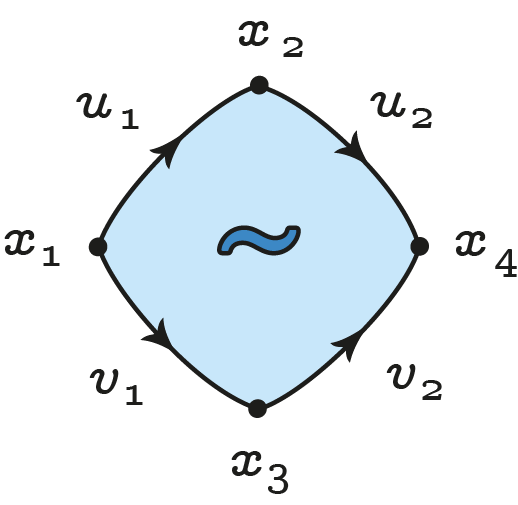}}
\hspace{.5em} \mapsto \hspace{.5em}
\raisebox{-3.6em}{\includegraphics[height=7.4em]{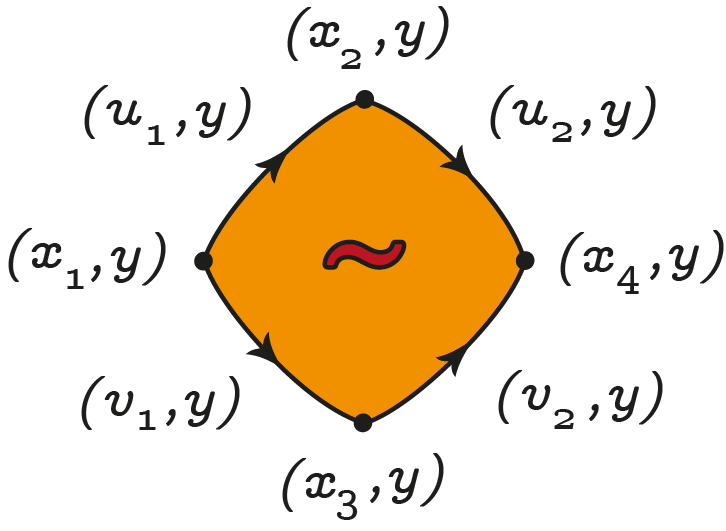}}
\vspace{-.2em}
\end{equation}
for every permutation tile $u_1\cdot u_2\diamond_G v_1\cdot v_2$ in~$G$ 
(in blue on the left-hand side) and every vertex~$y$ in~$H$;
symmetrically, a permutation tile (in orange on the right-hand side)
\begin{equation}\label{equation/asynchronous-tensor-tiles-C}
\raisebox{-3.6em}{\includegraphics[height=7.4em]{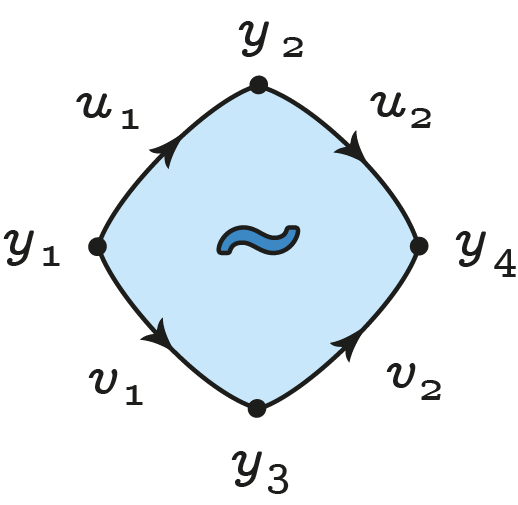}}
\hspace{.5em} \mapsto \hspace{.5em}
\raisebox{-3.6em}{\includegraphics[height=7.4em]{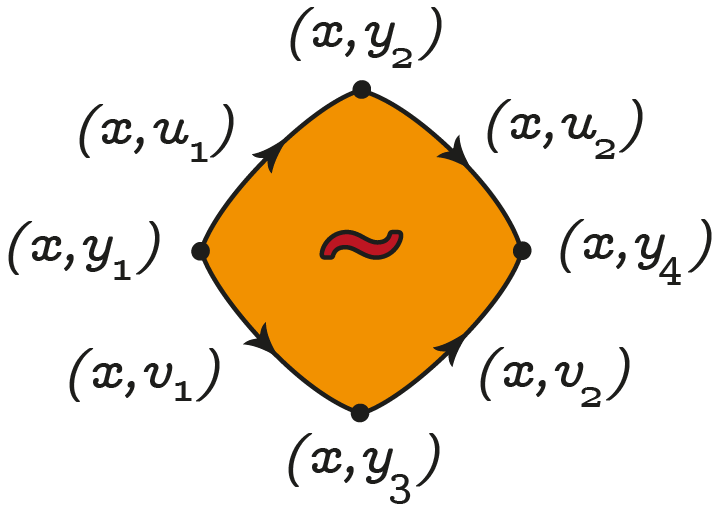}}
\vspace{-.2em}
\end{equation}
for every vertex~$x$ in~$G$ and permutation tile 
$u_1\cdot u_2\diamond_G v_1\cdot v_2$ in~$H$ 
(in blue on the left-hand side).
The unit of the shuffle tensor product
is the asynchronous graph $I$ with one vertex and no edge.
One establishes that
\medbreak
\begin{proposition}
The category $\Asynch$ of asynchronous graphs equipped
with the shuffle tensor product $\shuffletensor$ and the unit~$I$
defines a symmetric monoidal category.
\end{proposition}
\medbreak
An additional observation is that the construction $\asynchanchorof{-}$
described in~(\ref{equation/anchorof}) defines a symmetric monoidal functor
\begin{equation}\label{equation/anchorof-functor}
\begin{tikzcd}[column sep = .8em]
\asynchanchorof{-} \hspace{.5em} : \hspace{.5em} (\Set,+,\varnothing) \arrow[rr] && (\Asynch,\shuffletensor,I)
\end{tikzcd}
\end{equation}
with monoidal coercions given by the expected isomorphisms:
\begin{equation}\label{equation/monoidal-coercions}
\begin{small}
\begin{array}{cccc}
m_{\Label_1,\Label_2} \,\, :  &
\asynchanchorof{\Label_1}\shuffletensor\asynchanchorof{\Label_2}
& \longrightarrow &
\asynchanchorof{\Label_1+\Label_2}
\\
m_{\emptyset} \,\, : &
I
& 
\longrightarrow
&
\asynchanchorof{\varnothing}
\end{array}
\end{small}
\end{equation}
This basic observation has the remarkable consequence that 
for every set $\Label$ of labels:
\medbreak
\begin{proposition}
The asynchronous graph $\asynchanchorof{\Label}$
comes equipped with the structure of a commutative monoid
in the symmetric monoidal category $(\Asynch,\shuffletensor,I)$.
\end{proposition}
\medbreak
\noindent
The multiplication and unit of the commutative monoid
$$
\begin{array}{cccc}
\mathsf{mult}_{\Label} \,\, :  &
\asynchanchorof{\Label}\shuffletensor\asynchanchorof{\Label}
& \longrightarrow &
\asynchanchorof{\Label}
\\
\mathsf{unit}_{\Label} \,\, :  &
I
& 
\longrightarrow
&
\asynchanchorof{\Label}
\end{array}
$$
are defined by transporting along the symmetric monoidal functor $\asynchanchorof{-}$ 
the commutative monoid structure $(\Label,\nabla_{\Label},{\mathbf{0}}_{\Label})$ 
of the original set $\Label$ of labels,
provided by the canonical codiagonal and zero functions
\begin{equation}\label{equation/asynchronous-monoid-structure}
\begin{array}{cccc}
\nabla_{\Label} \,\, :  &
{\Label + \Label}
& \longrightarrow &
\Label
\\
{\mathbf{0}}_{\Label} \,\, :  &
\varnothing
& 
\longrightarrow
&
\Label
\end{array}
\end{equation}
in the cocartesian category $(\Set,+,\varnothing)$.

%paragraph
\medbreak
\noindent
\emph{\textbf{Asynchronous games:}}
Following the track and spirit of~\cite{mellies-popl-2019},
we define an asynchronous game $(G,\diamond_G,\lambda_G)$
as an asynchronous graph $(G,\diamond_G)$ where every edge $m:x\to y$ is seen as a move 
between positions $x$, $y$ ``polarized'' by a label~$\polarityplus$ 
when the move $m$ is played by Player on the side of the proof (or the program)
and by a label~$\polarityminus$ when the move $m$ is played
by Opponent on the side of the refutation (or the environment).
The polarization of moves is moreover required to be consistent with permutations.
For that reason, an asynchronous game $(G,\diamond_G,\lambda_G)$
is the same thing as an asynchronous graph $(G,\diamond_G)$
equipped with an asynchronous graph homomorphism
$$
\lambda_G \quad : \quad (G,\diamond_G) \longrightarrow \asynchanchorof{\polarityminus,\polarityplus}.
$$
%In this formulation, the dual $G^{\bot}$
%of an asynchronous game $(G,\diamond_G,\lambda_G)$ 
%has same support as $G$ and has labelling map defined as
%the composite
%$$
%\begin{tikzcd}[column sep=1em]
%\lambda_{G^{\bot}} \quad : \quad 
%{(G,\diamond)} \arrow[rr,"{\lambda_G}"]
%&&
%\anchorof{\polarityminus,\polarityplus}
%\arrow[rr,"neg"]
%&&
%\anchorof{\polarityminus,\polarityplus}.
%\end{tikzcd}
%$$
%where the asynchronous graph homomorphism $neg$ is defined by 
%the functor~(\ref{equation/anchorof-functor}) applied
%to the bijection $\polarityminus\mapsto \polarityplus$ and $\polarityplus\mapsto \polarityminus$
%which swaps the two polarities.
%$$
%\begin{tikzcd}[column sep = 1em]
%neg \,\, : \,\, \{\polarityminus,\polarityplus\} \arrow[rr] &&  \{\polarityminus,\polarityplus\}
%\end{tikzcd}
%$$
%transports $\polarityminus$ to $\polarityplus$ and conversely.
%
%Similarly, 
In this formulation, the tensor product of two asynchronous games 
$(G,\diamond_G,\lambda_G)$ and $(H,\diamond_H,\lambda_H)$
is defined as the asynchronous graph
$(G,\diamond_G)\shuffletensor(H,\diamond_H)$
with labelling map $\lambda_{G\shuffletensor H}$
defined as the composite of the tensor product
of the two labelling functions
$$
\begin{tikzcd}[column sep=2em]
{(G,\diamond_G)\shuffletensor(H,\diamond_H)}
\arrow[rr,"{\lambda_G\shuffletensor\lambda_H}"]
&&
\asynchanchorof{\polarityminus,\polarityplus}
\shuffletensor
\asynchanchorof{\polarityminus,\polarityplus}
\end{tikzcd}
$$
followed by the asynchronous graph homomorphism
\begin{equation}\label{equation/pince-on-asynchronous-graphs}
\begin{tikzcd}[column sep=1.6em]
\asynchanchorof{\polarityminus,\polarityplus}
\shuffletensor
\asynchanchorof{\polarityminus,\polarityplus}
\arrow[rr,"{\mathsf{mult}}"]
&&
\asynchanchorof{\polarityminus,\polarityplus}
\end{tikzcd}
\end{equation}
defined by the multiplication of the asynchronous graph $\asynchanchorof{\polarityminus,\polarityplus}$
described in (\ref{equation/asynchronous-monoid-structure})
%coercion (\ref{equation/monoidal-coercions}) 
%precomposed to the functor (\ref{equation/anchorof-functor})
%applied to the codiagonal function 
%$$
%\begin{tikzcd}[column sep=1em]
%\nabla_{\{\polarityminus,\polarityplus\}}
%\,\, : \,\,
%\{\polarityminus,\polarityplus\}+\{\polarityminus,\polarityplus\}
%\arrow[rr] && \{\polarityminus,\polarityplus\}
%\end{tikzcd}
%$$
which associates to each polarity $\polarityminus$ or $\polarityplus$
of the domain the same polarity $\polarityminus$ or $\polarityplus$ of the codomain.
This definition
%s of linear negation $A\mapsto A^{\bot}$
of the tensor product $A,B\mapsto A\tensor B$ of two asynchronous games
is reminiscent of a very similar construction in the template game model
of differential linear logic~\cite{mellies-popl-2019,mellies-lics-2019}
%recently designed by Melli{\`e}s, 
which we briefly recall now.

%After his work on concurrent games with Abramsky~\cite{AbramskyMellies},
%Melli\`es developed in \cite{Mellies} the idea that an asynchronous game
%should be defined as an asynchronous transition system where each transition
%is ``polarized'' and labelled by $+$ when it is played by Proponent (or Player)
%and by $-$ when it is played by Opponent.

%paragraph
\medbreak
\noindent
\emph{\textbf{Template games:}} 
The basic idea of template game semantics is to define a game $(A,\lambda_A)$ 
as a pair consisting $(a)$ of a category~$A$ whose objects are the positions $x,y\in A$
of the game, and whose morphisms $f:x\to y$ are the trajectories or plays of the game, 
and $(b)$ of a labelling functor
\begin{equation}\label{equation/template-game}
\begin{tikzcd}[column sep = 1.4em]
\lambda_A \quad : \quad A \arrow[rr] && \anchorofgames
\end{tikzcd}
\end{equation}
to a specific category $\anchorofgames$ of polarities called the \emph{template of games}.
One main advantage of the definition is that the template $\anchorofgames$
may be chosen to describe a particular class of games (sequential, concurrent)
and scheduling policy between Player and Opponent (alternating, non-alternating).
%
%Typically, in the case of alternating sequential games, the template of games 
%$\anchorofgames=\anchorofseqgames$
%is defined as the category freely generated by the graph
%\begin{equation}\label{equation/template-of-sequential-games}
%\begin{tikzcd}
%\lrangle{\polminus}
%\ar[rr,"{\polarityplus}"{swap}, yshift = -.3em]
%&&
%\lrangle{\polplus}
%\ar[ll,"{\polarityminus}"{swap}, yshift = .3em]
%\end{tikzcd}
%\end{equation}
%where $\lrangle{\polminus}$ and $\lrangle{\polplus}$ are positive and negative polarities on the positions,
%and the labels $\polarityminus$ and $\polarityplus$ indicate which trajectories $f:x\to y$ should be considered
%as Opponent and Player moves of the game.
%
Typically, in the case of non-alternating concurrent games, the template of games $\anchorofgames=\anchorofconcgames$
is defined as the category with one object $\lrangle{\ast}$ generated by two maps
\begin{equation}\label{equation/template-of-concurrent-games}
\begin{tikzcd}
\lrangle{\ast}
\arrow[loop left]{l}{\polarityplus}
\arrow[loop right]{r}{\polarityminus}
\end{tikzcd}
%\begin{tikzcd}[column sep = 1em]
%\ast
%\ar[rr,"{O}"]
%&&
%\ast
%\end{tikzcd}
%\quad\quad
%\begin{tikzcd}[column sep = 1em]
%\ast
%\ar[rr,"{P}"]
%&&
%\ast
%\end{tikzcd}
\end{equation}
and the unique equation $\polarityplus\cdot\polarityminus=\polarityminus\cdot\polarityplus$.
%\begin{equation}\label{equation/template-of-concurrent-games-equation}
%\begin{tikzcd}[column sep = 1.8em]
%\ast
%\ar[r,"{\polarityplus}"]
%&
%\ast
%\ar[r,"{\polarityinus}"]
%&
%\ast
%\end{tikzcd}
%\quad = \quad
%\begin{tikzcd}[column sep = 1.8em]
%\ast
%\ar[r,"{\polarityminus}"]
%&
%\ast
%\ar[r,"{\polarityplus}"]
%&
%\ast
%\end{tikzcd}
%\end{equation}
It is worth observing that since a category with one object is the same thing as a monoid,
the category $\anchorofgames=\anchorofconcgames$ may be equivalently defined as the commutative monoid
freely generated by the two elements $\{\polarityminus,\polarityplus\}$.
More generally, it makes sense to consider the symmetric monoidal functor 
\begin{equation}\label{equation/anchorof-categories}
\onecatanchorof{-} \quad : \quad (\Set,+,\varnothing) \longrightarrow (\Cat,\times,\textbf{1})
\end{equation}
which transports every set $\Label$ to the commutative monoid~$\onecatanchorof{\Label}$
freely generated by $\Label$, seen as a category with one object.
As a consequence, every category $\onecatanchorof{\Label}$ 
defines a commutative monoid in $\Cat$ with multiplication and unit noted
\begin{equation}\label{equation/monoid-on-anchorof-category}
\begin{array}{cccc}
\mathsf{mult}_{\Label} \,\, :  &
\onecatanchorof{\Label}\times\onecatanchorof{\Label}
& \longrightarrow &
\onecatanchorof{\Label}
\\
\mathsf{unit}_{\Label} \,\, :  &
{\mathbf{1}}
& 
\longrightarrow
&
\onecatanchorof{\Label}
\end{array}
\end{equation}
%Note that every category $\onecatanchorof{\Label}$ is a commutative monoid in $\Cat$
% strict symmetric monoidal category,
%and thus comes equipped with 
%
There is an obvious parallel with asynchronous games here,
since every concurrent template game~$A$ is a category 
labelled by the category of polarities 
\begin{equation}\label{equation/anchorofconcurrentgames}
\anchorofgames=\onecatanchorof{\polarityminus,\polarityplus}
\end{equation}
in the same way that every asynchronous game $(G,\diamond_G,\lambda_G)$
is an asynchronous graph $(G,\diamond_G)$ labelled 
by the asynchronous graph of polarities $\asynchanchorof{\polarityminus,\polarityplus}$.

%paragraph
\medbreak
\noindent
\emph{\textbf{Template games continued:}} 
One key observation of \cite{mellies-popl-2019} is that in good situations,
%the sequential as well as in the concurrent case,
%many important situations,
% in both the sequential and concurrent cases, 
the template of games $\anchorofgames$ comes equipped 
with a companion category $\anchorofstrat$ called the \emph{template of strategies},
whose purpose is to describe the possible schedulings of a strategy $\sigma$ 
between two template games $(A,\lambda_A)$ and $(B,\lambda_B)$.
The category $\anchorofstrat$ comes moreover equipped with two functors
\begin{equation}\label{equation/the-two-functors}
\begin{tikzcd}[column sep = 2em]
\anchorofgames 
&&
\anchorofstrat
\arrow[ll,"{s}"{swap}]\arrow[rr,"{t}"] 
&& \anchorofgames
\end{tikzcd}
\end{equation}
which describe how the scheduling of the strategy~$\sigma$
is related to the schedulings of the games $(A,\lambda_A)$ and $(B,\lambda_B)$.
%
%One main advantage of the approach is that one can define different 
%
A strategy $\sigma$ between two games $(A,\lambda_A)$ and $(B,\lambda_B)$
%of the bicategory
\begin{equation}\label{equation/sigma-maps}
\begin{tikzcd}[column sep=.8em, row sep=1.2em]
\sigma\,=\,(S,s,t,\lambda_{\sigma}) \,\, : \,\, (A,\lambda_A) \arrow[spanmap]{rrrr} &&&& (B,\lambda_B)
\end{tikzcd}
\end{equation}
is defined as a span of functors
$$
\begin{tikzcd}[column sep=2em]
\Agame && S \arrow[ll,"s"{swap}]\arrow[rr,"t"] && \Bgame
\end{tikzcd}
$$
with support a category~$S$, together with
a functor $\lambda_{\sigma}:S\to\anchorofstrat$
making the diagram below commute:
\begin{equation}\label{equation/morphism-of-span}
\begin{tikzcd}[column sep=1.8em, row sep=1.2em]
\Agame \arrow[dd,"{\lambda_A}"{swap}]
&& 
\Sgame
\arrow[ll,"{\stratsource}"{swap}]
\arrow[dd,"{\interaction{\sigma}}"]
\arrow[rr,"{\strattarget}"]
&&
\Bgame \arrow[dd,"{\lambda_\Bgame}"]
\\
\\
{\anchorofgames}
&& 
{\anchorofstrat}
\arrow[ll,"{\stratsource}"{swap}]
\arrow[rr,"{\strattarget}"]
&&
{\anchorofgames}
\end{tikzcd}
\end{equation}
In the case of concurrent template games, where $\anchorofgames=\onecatanchorof{\polarityminus,\polarityplus}$,
the template of strategies is defined as the category
\begin{equation}\label{equation/anchorofconcurrentstrategies}
\anchorofstrat \quad = \quad
\onecatanchorof{\polarityminussource,\polarityplussource,\polarityminustarget,\polarityplustarget}
\end{equation}

The intuition is that every move~$m$ played by a strategy $\sigma$
between two concurrent games $A$ and $B$ may be labelled 
in four different ways depending on the polarity of $m$ 
and the component $A$ or $B$ in which it is played:
\begin{center}
\begin{tabular}{ccl}
$\polarityminussource$ &:& an Opponent move played by $\sigma$ in the source $A$,\\
$\polarityplussource$ &:& a Player move played by $\sigma$ in the source $A$,\\
$\polarityminustarget$ &:& an Opponent move played by $\sigma$ in the target $B$,\\
$\polarityplustarget$ &:& a Player move played by $\sigma$ in the target $B$.\\
\end{tabular}
\end{center}
The functors $s$ and $t$ are then characterized 
by the image of the four generators of 
$\anchorofstrat=\onecatanchorof{\polarityminussource,\polarityplussource,\polarityminustarget,\polarityplustarget}$, as follows:
\begin{equation}\label{equation/st}
%\begin{small}
\begin{array}{llll}
\hspace{-.8em}
s: \,\, \polarityminussource\mapsto\polarityplus
&
\hspace{-.4em}
\polarityplussource\mapsto\polarityminus
&
\hspace{-.4em}
\polarityminustarget\mapsto \id{\lrangle{\ast}}
&
\hspace{-.4em}
\polarityplustarget\mapsto \id{\lrangle{\ast}}
\\
\hspace{-.8em}
t:\,\,
\polarityminussource\mapsto\id{\lrangle{\ast}}
&
\hspace{-.4em}
\polarityplussource\mapsto\id{\lrangle{\ast}}
&
\hspace{-.4em}
\polarityminustarget\mapsto\polarityminus
&
\hspace{-.4em}
\polarityplustarget\mapsto \polarityplus
\end{array}
%\end{small}
\end{equation}
Note that the definition of $s$ and $t$ formalizes the intuition
that every move~$m$ of polarity $\polarityminussource$ or $\polarityplussource$
played by the strategy~$\sigma$ should have a reverse polarity~$\polarityplus$
or~$\polarityminus$ in the source game~$A$, 
while a move~$m$ of polarity $\polarityminustarget$ or $\polarityplustarget$
should retain its polarity~$\polarityminus$ or~$\polarityplus$ in the target game~$B$.
%
%This is precisely the task performed by the two functors~$s$ and~$t$ mentioned in (\ref{equation/the-two-functors})
%and explicitly described above in the case $\anchorofgames=\onecatanchorof{\polarityminus,\polarityplus}$
%of concurrent template games.

%paragraph
\medbreak
\noindent
\emph{\textbf{A tale of deadlocks and diagonals:}}
As we saw, there is a strong affinity between asynchronous games
and concurrent template games.
This is particularly striking in the definition of the tensor product
of two concurrent template games $(A,\lambda_A)$ and $(B,\lambda_B)$
as the template game 
\begin{equation}\label{equation/tensor-on-cat}
(A,\lambda_A)\tensor (B,\lambda_B) = (A\times B, \lambda_{A\tensor B})
\end{equation}
with support the cartesian product~$A\times B$ of the categories~$A$ and~$B$, 
and with labelling functor~$\lambda_{A\tensor B}$ defined as the composite
of the cartesian product of the two labelling functors
$$
\begin{tikzcd}[column sep=2.2em]
{A\times B}
\arrow[rr,"{\lambda_A\times\lambda_B}"]
&&
\onecatanchorof{\polarityminus,\polarityplus}
\times
\onecatanchorof{\polarityminus,\polarityplus}
\end{tikzcd}
$$
followed by the multiplication functor defined in~(\ref{equation/monoid-on-anchorof-category}) 
\begin{equation}\label{equation/pince-on-template-categories}
\begin{tikzcd}[column sep=1.9em]
\onecatanchorof{\polarityminus,\polarityplus}
\times
\onecatanchorof{\polarityminus,\polarityplus}
\arrow[rr,"{\mathsf{mult}_{\Label}}"]
&&
\onecatanchorof{\polarityminus,\polarityplus}
\end{tikzcd}
\end{equation}
%by the multiplication $\mathsf{mult}_{\Label}$ 
for the category $\onecatanchorof{\Label}$ where $\Label=\{\polarityminus,\polarityplus\}$.
However, one main difference between~(\ref{equation/pince-on-template-categories})
and~(\ref{equation/pince-on-asynchronous-graphs}) is that the cartesian product
of categories replaces the shuffle tensor product of asynchronous graphs.
Let us illustrate the difference with an example, and explain
why the interpretation provided by the asynchronous game model
is more satisfactory than the concurrent template model.
%see on one example
%This has the quite unfortunate consequence of identifying too many trajectories
%of the interpretation of formulas.
%
Consider the asynchronous game $A$ consisting of
one Opponent move $m:x\to x'$ between two positions $x,x'$ ;
and the asynchronous game $B$ consisting of one Player move $n:y\to y'$
between two positions $y,y'$.
The asynchronous game $A\tensor B$ has the asynchronous graph
with four vertices (positions), four edges (moves) and a permutation tile
as support:
\begin{equation}\label{equation/permutation-tile-example}
\raisebox{-3.8em}{\includegraphics[height=7.6em]{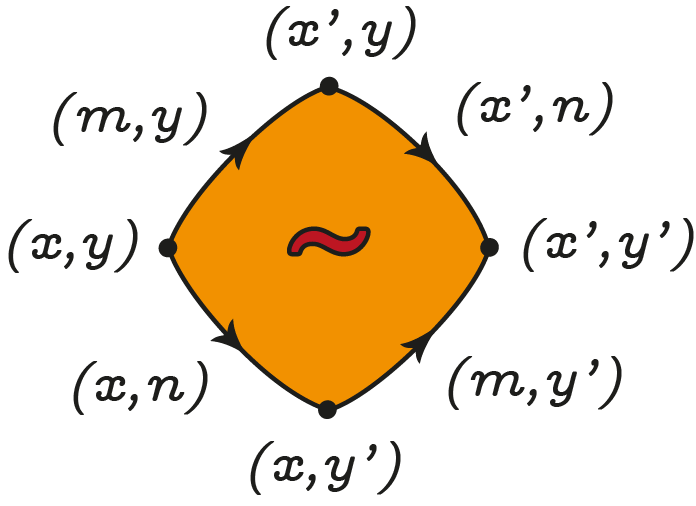}}
% \vspace{-0em}
\end{equation}
Now, consider the strategy $\sigma$ of the game $A\shuffletensor B$ 
which plays the trajectory (and its prefixes)
\begin{equation}\label{equation/first-trajectory}
\begin{tikzcd}[column sep = 1.5em]
{(x,y)}\arrow[rr,"{(m,y)}"] && {(x',y)}\arrow[rr,"{(x',n)}"] && {(x',y')}
\end{tikzcd}
\end{equation}
from the starting position $(x,y)$ and the counter-strategy $\tau$ 
which plays the trajectory (and its prefixes)
\begin{equation}\label{equation/second-trajectory}
\begin{tikzcd}[column sep = 1.5em]
{(x,y)}\arrow[rr,"{(x,n)}"] && {(x,y')}\arrow[rr,"{(m,y')}"] && {(x',y')}
\end{tikzcd}
\end{equation}
from the very same position $(x,y)$.
In this situation, the interaction between $\sigma$ and $\tau$ produces a deadlock,
because each strategy $\sigma$ and $\tau$ is waiting for a signal
which never comes from the other side in order to proceed.
The deadlock between $\sigma$ and $\tau$ is nicely reflected in the asynchronous game $A\shuffletensor B$ 
by the fact that the intersection of~$\sigma$ and of~$\tau$ only contains the empty path from $(x,y)$ to itself.

Unfortunately, the situation is not as satisfactory
%\footnote{It seems that this point is entirely overlooked in~\cite{mellies-popl-2019,mellies-lics-2019} and we are not unhappy to teach the author a lesson of mathematics.}
when one shifts to the corresponding concurrent template games~$A$ and~$B$.
Indeed, when one considers~$\sigma$ and~$\tau$ as strategies playing on the concurrent template game~$A\tensor B$
obtained by tensoring~$A$ and~$B$ using~(\ref{equation/tensor-on-cat}),
it appears that the two trajectories~(\ref{equation/first-trajectory}) and~(\ref{equation/second-trajectory})
are identified as the ``synchronized move'' given by the diagonal map $(m,n):(x,y)\to (x',y')$
%\begin{equation}\label{equation/synchronized-trajectory}
%\begin{tikzcd}[column sep = 1.5em]
%{(x,y)}\arrow[rrrr,"{(m,n)}"] &&&& {(x',y')}
%\end{tikzcd}
%\end{equation}
of the cartesian category $A\times B$ underlying the concurrent template game $A\tensor B$,
which we represent in full below:
\begin{equation}\label{equation/cartesian-product}
\begin{tikzcd}[column sep=.65em,row sep=.75em]
&& 
{(x',y)}
\arrow[rrdd,"{(x',n)}"]
\\
\\
{(x,y)}
\arrow[rrrr,"{(m,n)}"]
\arrow[rrdd,"{(x,n)}"{swap}]
\arrow[rruu,"{(m,y)}"]
&&&&
{(x',y')}
\\
\\
&&
{(x,y')}
\arrow[rruu,"{(m,y')}"{swap}]
\end{tikzcd}
\end{equation}
The fact that the two trajectories~(\ref{equation/first-trajectory}) and~(\ref{equation/second-trajectory})
are identified to the diagonal map $(m,n)$ means that the intersection of~$\sigma$ and~$\tau$ 
is not trivial anymore, since it contains the diagonal map.
In a sense, the categorical interpretation believes (wrongly!) that the two strategies~$\sigma$ and~$\tau$
could ``resolve'' the deadlock by playing the two moves~$m$ and~$n$ synchronously.
This excess of synchronization in the interpretation of $A\tensor B$
should be seen as a defect (not as a feature!) of the original concurrent template game semantics
formulated in~\cite{mellies-popl-2019,mellies-lics-2019}.
%

%paragraph
\medbreak
\noindent
\emph{\textbf{The Gray tensor product:}}
We correct the situation in the present paper by shifting one dimension higher
the original formalism of template game semantics:
we design an asynchronous template game model
where template games $(A,\lambda_A)$ are defined
as 2-categories instead of simply as 1-categories.
% of positions, trajectories and reshufflings (or reschedulings)
%as in the original formulation~\cite{mellies-popl-2019,mellies-lics-2019}.
%
To that purpose, we take seriously the pretty striking analogy
between the shuffle tensor product
$A,B\mapsto A\shuffletensor B$ of asynchronous graphs,
and the Gray tensor product $A,B\mapsto A\graytensor B$ of 2-categories.
The Gray tensor product was introduced by Gray in~\cite{Gray-book-1974}
and it plays a fundamental role in contemporary categorical algebra,
see for instance \cite{BataninCisinskiWeber13,bourke-gurski-acs-2017}.
One reason is that the Gray tensor product comes with 
a closed structure provided by the hom-2-category $[A,B]_{ps}$
between 2-categories, defined as the 2-category of 2-functors
from~$A$ to~$B$, pseudonatural transformations and modifications.
Another reason is a subtle coherence theorem for tricategories
\cite{GordonPowerStreet95,Gurski06} establishes
that every tricategory is equivalent to a Gray-category,
defined as category enriched over the category $\TwoCat$ 
equipped with the Gray tensor product.

One first contribution of the paper is to clarify
the relationship between the shuffle tensor product
of asynchronous graphs and the Gray tensor product 
of 2-categories by constructing a symmetric monoidal functor
\begin{equation}\label{equation/asynchtwocat}
\asynchtwocat{-} \,\,\, : \,\,\, (\Asynch,\shuffletensor,\textbf{I}) \,\,\, \longrightarrow \,\,\, (\TwoCat,\graytensor,\textbf{I})
\end{equation}
The functor $\asynchtwocat{-}$ associates to every 
asynchronous graph~$(G,\diamond)$ the 2-category~$\asynchtwocat{G,\diamond}$
whose objects are the vertices of $G$, whose morphisms are the paths of $G$
and whose 2-dimensional cells are the reshufflings (or reschedulings) 
between paths of the graph~$G$.
A detailed definition of reshuffling appears in \S\ref{section/reshuffling}
but let us already mention that the main purpose of reshufflings
is to identify different combinations of permutation tiles such as
\begin{center}
\begin{tabular}{ccc}
\raisebox{-3.4em}{\includegraphics[height=7.5em]{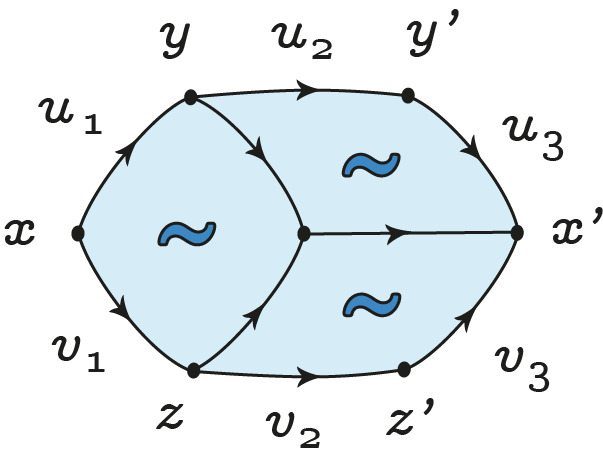}}
& $\cong$ & 
\raisebox{-3.4em}{\includegraphics[height=7.5em]{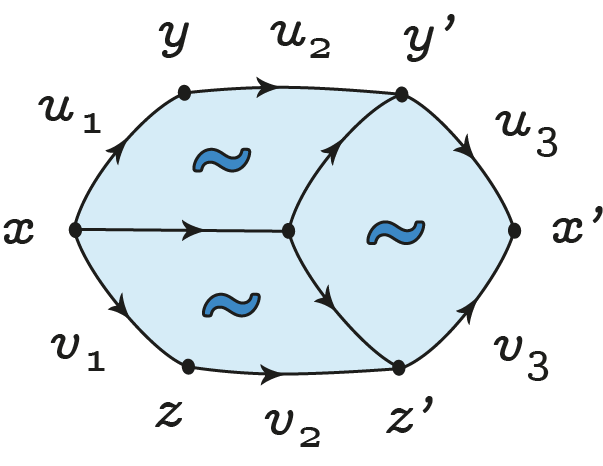}}
\end{tabular}
\end{center}
which deserve to be considered as equivalent in the asynchronous graph $G$.
The fact that the functor $\asynchtwocat{-}$
comes equipped with natural isomorphisms of 2-categories:
%\begin{equation}
%\begin{small}
$$
\begin{array}{cccc}
p_{G,H} \,\, :  &
\asynchtwocat{G,\diamond_G}\graytensor\asynchtwocat{H,\diamond_H}
& \longrightarrow &
\asynchtwocat{G\shuffletensor H,\diamond_{G\shuffletensor H}}
\\
p_{I} \,\, : &
{\textbf{1}}
& 
\longrightarrow
&
\asynchtwocat{\mathbf{I}}
\end{array}
$$
qualifies the Gray tensor product as a natural and expressive extension
of the usual shuffle product of asynchronous graphs.
In particular, the functor $\asynchtwocat{-}$ may be composed
to $\asynchanchorof{-}$ in order to obtain the symmetric monoidal functor 
\begin{equation}\label{equation/anchorof-two-categories}
\twocatanchorof{-} 
%=
%\asynchtwocat{-}\circ\asynchanchorof{-}
\,\, : \,\, (\Set,+,\varnothing) \,\, \longrightarrow \,\, (\TwoCat,\times,\textbf{1})
\end{equation}
which associates to every set $\Label$
the symmetric monoidal category freely generated by $\Label$,
seen as a 2-category with one object.
Note that $\twocatanchorof{-}$ plays the same role
as the functor~(\ref{equation/anchorof-categories})
which it upgrades from categories to 2-categories.

\medbreak
\noindent
\emph{\textbf{Asynchronous template games:}}
We have accumulated enough material
%after this long but necessary introduction
in order to define the notion of \emph{asynchronous template game}.
There is still one fundamental obstruction 
which we need to resolve however: the notion of template
formulated in \cite{mellies-popl-2019,mellies-lics-2019} is defined
as an internal category 
$$
\moo = (\moo[0],\moo[1],\mathsf{mult}:\moo[2]\to\moo[1],\mathsf{unit}:\moo[1]\to\moo[0])
$$
in the category~$\Scategorysurround$ with finite limits,
where $\moo[0]=\anchorofgames$ is the object of objects
and $\moo[1]=\anchorofstrat$ is the object of morphisms.
For that reason, the formalism of template games
originally developed in \cite{mellies-popl-2019,mellies-lics-2019}
does not work any more 
when the category $\Scategorysurround=\Cat$ with finite limits
is replaced by the monoidal category $\Scategorysurround=\TwoCat$
equipped with the Gray tensor product.
One main contribution of the paper is to resolve 
that foundational issue by establishing 
an unexpected connection with the seminal work
by Marcelo Aguiar on quantum groups~\cite{aguiar-phd-1997}.
The key observation is that the 2-category
${\twocatanchorof{\polarityminus,\polarityplus}}$
of asynchronous polarities is equipped
with a comonoid structure
\begin{equation}\label{equation/Gray-comonoid-of-objects}
\begin{array}{cc}
\begin{tikzcd}[column sep = 1em]
d \,\,\, : \,\,\, {\twocatanchorof{\polarityminus,\polarityplus}}
\arrow[rr] && {\twocatanchorof{\polarityminusone,\polarityplusone}\graytensor\twocatanchorof{\polarityminustwo,\polarityplustwo}}
\end{tikzcd}
\\
\begin{tikzcd}[column sep = 1em]
e \,\,\, : \,\,\, {\twocatanchorof{\polarityminus,\polarityplus}}
\arrow[rr] && {\grayunit}
\end{tikzcd}
\end{array}
\end{equation}
where we find useful to recall the isomorphism:
$$
\twocatanchorof{\polarityminusone,\polarityplusone}\graytensor\twocatanchorof{\polarityminustwo,\polarityplustwo}
\quad \cong \quad
\twocatanchorof{\polarityminusone,\polarityplusone,\polarityminustwo,\polarityplustwo}
$$
associated to the symmetric monoidal functor~(\ref{equation/anchorof-two-categories}).
The comultiplication $d$ is then characterized
(and defined) by the fact that~$d$ transports the edges $\polarityminus$
and $\polarityplus$ to the paths of length 2:
$$
d \quad : \quad
\polarityminus \mapsto \polarityminustwo \cdot \polarityminusone
\quad \quad \quad
\polarityplus \mapsto \polarityplusone\cdot \polarityplustwo
$$
Note the right-to-left orientation of the Opponent polarity $\polarityminus$
and the left-to-right orientation of the Player polarity $\polarityplus$.
A nice connection emerges here between the notion
of \emph{copycat strategy} in game semantics
and the notion of \emph{Gray comonoid}
defined as a comonoid for the Gray tensor product~$\graytensor$
of 2-categories: typically, the path or trajectory
$$
\polarityminus\cdot\polarityplus\cdot\polarityminus \quad \in \quad {\twocatanchorof{\polarityminus,\polarityplus}}
$$
describing the sequence of an Opponent move, 
a Player move and an Opponent move
played in the 2-category $\twocatanchorof{\polarityminus,\polarityplus}$
of polarities, is transported to the path or trajectory
$$
\polarityminustwo\cdot\polarityminusone\cdot\polarityplusone
\cdot\polarityplustwo\cdot\polarityminustwo\cdot\polarityminusone
\quad \in \quad
\twocatanchorof{\polarityminusone,\polarityplusone}\graytensor\twocatanchorof{\polarityminustwo,\polarityplustwo}
$$
of six moves shuffled between 
$\twocatanchorof{\polarityminusone,\polarityplusone}$
and 
$\twocatanchorof{\polarityminustwo,\polarityplustwo}$.
The six moves may be represented
%The behavior of the comultiplication $d$ may be represented 
as blue nodes and red nodes on a string diagram, 
where the ``zig-zag'' induced by the comultiplication~$d$ implements 
the scheduling of the usual copycat strategy
\cite{abramsky-jagadeesan-malacaria,hyland-ong-full-abstraction} of game semantics:
\begin{center}
\fbox{\begin{tabular}{ccc}
\raisebox{-5.4em}{\includegraphics[height=11.2em]{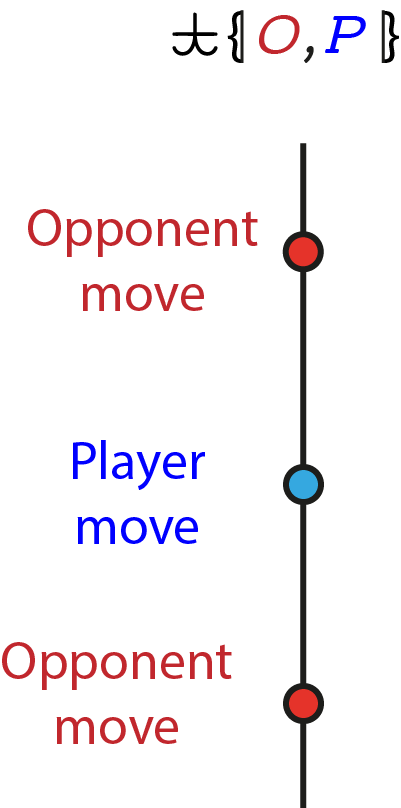}}
& $\mapsto$ & 
\raisebox{-5.4em}{\includegraphics[height=11.2em]{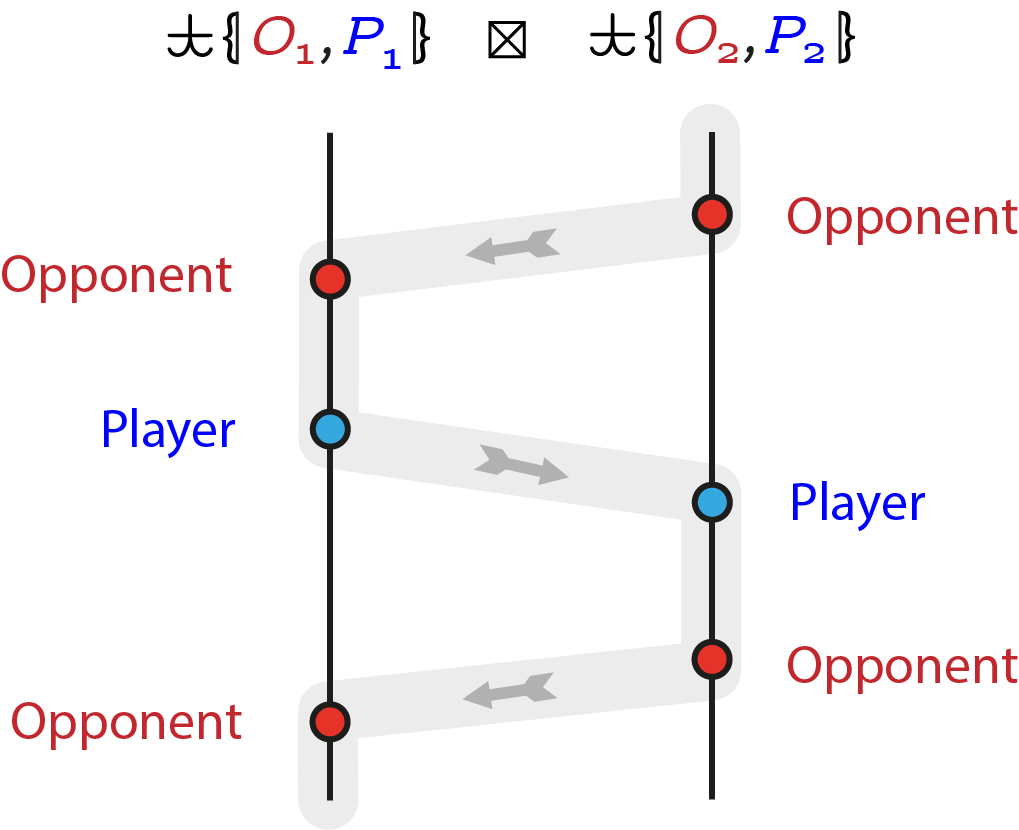}}
\end{tabular}}
\end{center}
Note that following a well-established convention of game semantics,
%adapted to our 2-categorical framework, 
the flow of time goes top-down, as indicated
by the grey arrows switching components in the diagram.
This connection between Gray comonoids and game semantics
can be explained by the existence of a canonical 2-functor
$$
\begin{tikzcd}[column sep = 1em]
q_{A,B} \quad : \quad 
A\graytensor B \arrow[rr] && A\times B
\end{tikzcd}
$$
and the fact that the composite
$$
\begin{tikzcd}[column sep = 1.5em]
A \arrow[rr,"{d_A}"] && 
{A\graytensor A} \arrow[rr,"{q_{A,B}}"] && 
{A\times A}
\end{tikzcd}
$$
coincides for every Gray comonoid $(A,d_A,e_A)$
with the diagonal of the 2-category~$A$.
For this reason, the comultiplication~$d_A$
of a Gray comonoid $A$ is required to provide
a specific recipe to transport every morphism $f:x\to y$ in~$A$
to a morphism $d_A(f):(x,x)\to(y,y)$ in $A\graytensor A$
which projects as $f:x\to y$ on both left and right component.
For that reason, one establishes easily that

\medbreak
\begin{proposition}\label{proposition/polarized-set}
Given a set $S$, a Gray comonoid structure 
$(\twocatanchorof{S},d_S,e_S)$
on the 2-category $\twocatanchorof{S}$ is the same thing 
as a polarity function $\lambda_S:S\to\{\polarityminus,\polarityplus\}$.
\end{proposition}
\medbreak

The polarity $\lambda_S(e)$ for $e\in S$ 
indicates whether the morphism $d(e):(\ast,\ast)\to(\ast,\ast)$
in the 2-category $\twocatanchorof{S}\graytensor\twocatanchorof{S}$
is defined left-to-right as $e_1\cdot e_2$
in the case $\lambda_S(e)=\polarityplus$
or right-to-left as $e_2\cdot e_1$
in the case $\lambda_S(e)=\polarityminus$.

\medbreak

This pleasant convergence between the notions
of polarity in a game and of Gray comonoid
%in the monoidal category $(\TwoCat,\graytensor,\grayunit)$
leads us to the main definition of the paper:
\medbreak
\begin{definition}[Asynchronous template game]\label{definition/asynchronous-template-games}
An asynchronous template game $(A,\lambda_A)$
is defined as a Gray comonoid $(A,d_A,e_A)$
equipped with a Gray comonoid homomorphism
%${\lambda_G}:\asynchtwocat{G,\diamond}\to{\twocatanchorof{\polarityminus,\polarityplus}}$
$$
\begin{tikzcd}[column sep = 1em]
{\lambda_A}\quad : \quad
A
\arrow[rr] && {\twocatanchorof{\polarityminus,\polarityplus}}.
\end{tikzcd}
$$
to the Gray comonoid $\twocatanchorof{\polarityminus,\polarityplus}$
in~(\ref{equation/Gray-comonoid-of-objects}).
The 2-category $A$ is called the \emph{support} of the 
asynchronous template game~$(A,\lambda_A)$.
Here, by Gray comonoid homomorphism 
$$\begin{tikzcd}[column sep = 1em]
h\quad : \quad
{(A,d_A,e_A)}
\arrow[rr] && 
{(B,d_B,e_B)}
\end{tikzcd}$$
we mean a 2-functor $h:A\to B$ between the underlying 2-categories~$A$ and~$B$,
making the diagram below commute:
$$
\begin{tikzcd}[column sep = 3em, row sep = 1em]
A\arrow[rr,"{h}"]\arrow[dd,"{d_A}"{swap}]
&&
B
\arrow[dd,"{d_B}"]
\\
\\
A\graytensor A\arrow[rr,"{h\graytensor h}"]
&&
{B\graytensor B}
\end{tikzcd}
$$
\end{definition}
%One establishes
% by adapting the proof of Prop.~\ref{proposition/polarized-set}
The definition 
(Def.~\ref{definition/asynchronous-template-games}) 
of asynchronous template game provides a conservative extension
of the traditional notion of asynchronous games,
in the sense that an asynchronous game in the usual sense
(described earlier in the paper) happens to be the same thing as
an asynchronous template game $(A,\lambda_A)$
%is the same thing as an asynchronous game in the specific case 
whose support~$A$ is the 2-category $A=\asynchtwocat{G,\diamond_G}$
of positions, trajectories and reshufflings
of an asynchronous graph $(G,\diamond_G)$.

\medbreak
\noindent
\emph{\textbf{The template of asynchronous games:}}
At this stage, we want to upgrade one dimension up 
the definitions~(\ref{equation/anchorofconcurrentgames})
and~(\ref{equation/anchorofconcurrentstrategies}) 
of the template of concurrent games in~\cite{mellies-popl-2019}.
We thus write
$$
\anchorofgames = \twocatanchorof{\polarityminus,\polarityplus}
\quad\quad\quad
\anchorofstrat = 
\twocatanchorof{\polarityminussource,\polarityplussource, \polarityminustarget,\polarityplustarget}
$$
where $\anchorofgames$ and $\anchorofstrat$ are now understood as 2-categories of polarities, instead of categories.
We have seen in Def.~\ref{definition/asynchronous-template-games} that a template game $(A,\lambda_A)$
is a Gray comonoid $(A,d_A,e_A)$ equipped with a comonoid homomorphism
$$
\begin{tikzcd}[column sep = 1em]
{\lambda_A}\quad : \quad
A
\arrow[rr] && {\anchorofgames}
\end{tikzcd}
$$
There remains to define the appropriate asynchronous notion of strategy $\sigma$
adapting the original definition~(\ref{equation/sigma-maps}) as a span of categories.
Following the guidance and inspiration of~\cite{aguiar-phd-1997},
we observe that~$\anchorofstrat$ comes equipped with a 2-functor
\begin{equation}\label{equation/bicomodule-of-anchor}
\begin{tikzcd}[column sep = .8em]
{\mathsf{coact}_{\moo}}\,\, : \,\,
\anchorofstrat \arrow[rr] && {\anchorofgames\graytensor\anchorofstrat\graytensor\anchorofgames}
\end{tikzcd}
\end{equation}
defined by the images of the four generating edges:
$$
%{\mathsf{coact}_{\moo}}\,\, : \,\, 
\polarityminussource \mapsto \polarityplusone\cdot\polarityminussource
\quad
\polarityplussource \mapsto \polarityplussource\cdot\polarityminusone
\quad
\polarityminustarget \mapsto \polarityminusthree\cdot \polarityminustarget
\quad
\polarityplustarget \mapsto \polarityplustarget\cdot\polarityplusthree
$$
where the indices~$1$ and~$3$ indicate 
in which component 
%of ${\anchorofgames\graytensor\anchorofstrat\graytensor\anchorofgames}$
the moves are played.
Here, the 2-functor ${\mathsf{coact}_{\moo}}$ should be understood 
as an asynchronous variant (and refinement) of the span $(s,t)$
of projection functors~(\ref{equation/the-two-functors})
as defined in~(\ref{equation/st}).
Moreover, one observes that the 2-functor ${\mathsf{coact}_{\moo}}$
just defined satisfies the equations required of 
a $\anchorofgames,\anchorofgames$-bicomodule structure
on the template 2-category~$\anchorofstrat$ of strategies,
see \S\ref{section/horizontal-maps} for a definition of bicomodule between comonoids.

\medbreak

This observation leads us to the following definition 
of asynchronous strategy between template games:
\medbreak
\begin{definition}[Asynchronous strategies]\label{definition/asynchronous-strategies}
An asynchronous strategy
\begin{equation}\label{equation/sigma-comodule}
\begin{tikzcd}[column sep=.8em, row sep=1.2em]
\sigma\,=\,(S,\mathsf{coact}_{\sigma},\lambda_{\sigma}) \,\, : \,\, (A,\lambda_A) \arrow[spanmap]{rrrr} &&&& (B,\lambda_B)
\end{tikzcd}
\end{equation}
%between template games $(A,\lambda_A)$ and $(B,\lambda_B)$
%is defined as a triple $\sigma=(S,\mathsf{coact}_{\sigma},\lambda_{\sigma})$
is a triple~$(S,\mathsf{coact}_{\sigma},\lambda_{\sigma})$
consisting of 2-category~$S$ used as support,
%called the support of the strategy $\sigma$
together with an $A,B$-bicomodule structure
$$
\begin{tikzcd}[column sep = 1em]
\mathsf{coact}_{\sigma}\quad : \quad
S
\arrow[rr] && 
{A} \graytensor S\graytensor {B}
\end{tikzcd}
$$
and a polarity 2-functor 
$$
\begin{tikzcd}[column sep = 1em]
\lambda_{\sigma}
\quad : \quad
S
\arrow[rr] && 
\anchorofstrat
%{\twocatanchorof{\polarityminussource,\polarityplussource,\polarityminustarget,\polarityplustarget}}
\end{tikzcd}
$$
making the diagram below commute:
\begin{equation}\label{equation/strategy-polarity}
\begin{tikzcd}[column sep = 3em, row sep = .8em]
S\arrow[rr,"{\lambda_{\sigma}}"]
\arrow[dd,"{\mathsf{coact}_\sigma}"{swap}]
&&
{\moo[1]}
\arrow[dd,"{\mathsf{coact}_{\moo}}"]
\\
\\
{A
\graytensor
S
\graytensor
B}
\arrow[rr,"{\lambda_{A}\graytensor\lambda_{\sigma}\graytensor\lambda_{B}}"]
&&
{\moo[0]}
\graytensor
{\moo[1]}
\graytensor
{\moo[0]}
\end{tikzcd}
\end{equation}
%$$
%\begin{tikzcd}[column sep = 3em, row sep = 1em]
%S\arrow[dd,"{\lambda_{\sigma}}"{swap}]
%\arrow[rr,"{\mathsf{coact}_\sigma}"]
%&&
%{A
%\graytensor
%S
%\graytensor
%B}
%\arrow[dd,"{\lambda_{A}\graytensor\lambda_{\sigma}\graytensor\lambda_{B}}"]
%\\
%\\
%{\moo[1]}
%\arrow[rr,"{\mathsf{coact}_{\moo}}"]
%&&
%{\moo[0]}
%\graytensor
%{\moo[1]}
%\graytensor
%{\moo[0]}
%\end{tikzcd}
%$$
\end{definition}
As explained later in the paper, writing
$$\anchorofasynch=(\anchorofgames,\anchorofstrat,\mathsf{coact}_{\moo})$$
for the template of asynchronous games just described, 
%see \S\ref{section/asynchronous-template} and \S\ref{section/star-autonomous},
we obtain a symmetric monoidal closed (and in fact $\ast$-autonomous)
bicategory $\Games{(\anchorofasynch)}$ of asynchronous template games
and asynchronous strategies between them.

\medbreak
\begin{theorem}
The bicategory $\Games{(\anchorofasynch)}$ is $\ast$-autonomous
%symmetric monoidal closed
and cartesian, and thus defines a model of the multiplicative and additive fragment of linear logic (MALL).
\end{theorem}

\medbreak
On the practical side, the tensor product $A,B\mapsto A\tensor B$
of asynchronous template games is defined in $\Games{(\anchorofasynch)}$
using the Gray tensor product and thus avoids the defect
discussed in~(\ref{equation/permutation-tile-example}) and~(\ref{equation/cartesian-product}).
On the conceptual side, the construction extends in a very natural way
the original framework of template games~\cite{mellies-popl-2019,mellies-lics-2019}
to the more general situation where the template~$\anchor=\anchorofasynch$
defines an \emph{internal category}
in a monoidal category~$\Scategorysurround$
such as $\Scategorysurround=(\TwoCat,\graytensor,\grayunit)$
--- see Aguiar~\cite{aguiar-phd-1997} and \S\ref{section/weak-double-category-of-bicomodules},
\S\ref{section/asynchronous-template} for a definition ---
%in the sense of~\cite{aguiar-phd-1997} 
instead of an internal category in a category~$\Scategorysurround$ with limits such as $\Scategorysurround=\Cat$.

\section{Related works and synopsis}\label{section/related-works}
Besides the numerous connections already mentioned to the work on template games
by Melli{\`e}s~\cite{mellies-popl-2019,mellies-lics-2019},
we should mention the formal analysis of template games by Eberhard, Hirschowitz and Laouar~\cite{eberhard-hirschowitz-laouar-fscd-2019}
where a template $\moo$ is identified in full generality as a formal monad
$$\moo=(\moo[0],\moo[1],\mathsf{mult},\mathsf{unit})$$
living in a weak double category.
Our construction follows that track by defining in~\S\ref{section/asynchronous-template}
a template as a monad $\moo$
%=(\moo[0],\moo[1],\mathsf{mult},\mathsf{unit})$
living in the weak double category of bicomodules $\Comod{\Scategorysurround}$
in a monoidal category~$(\Scategorysurround,\graytensor,\grayunit)$
with coreflexive equalizers preserved by the tensor product $\graytensor$ componentwise.
This is precisely the definition of internal category formulated by Aguiar in his seminal work on quantum algebras \cite{aguiar-phd-1997}.

The present work is part of a broader trend of research on game semantics and concurrency,
started with the definition of concurrent games by Abramsky and Melli{\`e}s~\cite{abramsky-mellies-lics-99}
and the subsequent series of works on asynchronous games~\cite{mellies:ag2-tcs,Mellies05ctcs,mellies-mimram-concur-2007}
defined as asynchronous graphs (or event structures) with transitions polarized in $\{\polarityminus,\polarityplus\}$.
The insight was adapted by Rideau and Winskel~\cite{rideau-winskel-lics-2011,CastellanClairambaultRideauWinskel17}
to the language of event structures, giving rise to very interesting developments
by Castellan, Clairambault and Winskel~\cite{castellan-clairambault-winskel-symmetry,CastellanClairambaultWinskel19}
on concurrent games based on event structures with symmetries.
%
%We have seen in the introduction how to translate every asynchronous graph (and thus every event structure)
%$(G,\diamond_G)$ into a 2-category~$A=\asynchtwocat{G,\diamond_G}$.
%
The notion of asynchronous template game based on 2-categories (see Def.~\ref{definition/asynchronous-template-games})
% developed in the present paper
offers an expressive and powerful generalization of traditional asynchronous games based on event structures,
and it would be interesting to understand whether the concurrent game model
with symmetries formulated in~\cite{castellan-clairambault-winskel-symmetry,CastellanClairambaultWinskel19}
can be accommodated in the language of asynchronous template games developed in the present paper.

Finally, we would like to mention the work by Eberhart and Hirschowitz~\cite{EberhartHirschowitz18}
on a general theory of game semantics based on polynomial functors between presheaf or sheaf categories.
We would be interested to see how their structural approach 
to game semantics could be adapted to the 2-categorical 
and monoidal framework of asynchronous template games developed in the present paper.

%From that point of view, a compelling research question is to understand how the bicategory $\Games{(\anchor)}$
%developed in the present paper is related to the categories of games and strategies with symmetries
%developed in\cite{castellan-clairambault-winskel-symmetry,CastellanClairambaultWinskel19}.

%\cite{mellies-stefanesco-lics-2018}
%\cite{}

\medbreak
\noindent
\emph{\textbf{Synopsis of the paper:}}
After a long and detailed overview of the paper
which we found clarifying and mandatory in~\S\ref{section/introduction}
and a comparison with related works in~\S\ref{section/related-works},
we start the technical part of the paper with a precise
and formal definition of asynchronous graphs
in~\S\ref{section/asynchronous-graphs}.
We then recall the notion of Gray tensor product in~\S\ref{section/gray-tensor-product}
and establish the important property that $\graytensor$ preserves 
the coreflexive equalizers of $\TwoCat$ componentwise.
We then give an explicit description in~\S\ref{section/the-translation}
of the translation $(G,\diamond)\mapsto\asynchtwocat{G,\diamond}$
from asynchronous graphs to 2-categories.
Using the preservation of coreflexive equalizers established in~\S\ref{section/gray-tensor-product},
we construct in~\S\ref{section/weak-double-category-of-bicomodules}
the weak double category of bicomodules $\Comod{\Scategorysurround}$
in the symmetric monoidal category $\Scategorysurround=(\TwoCat,\graytensor,\grayunit)$.
We finally define in \S\ref{section/asynchronous-template}
the template $\anchor=\anchorofasynch$ of asynchronous games,
and establish in~\S\ref{section/star-autonomous}
that the resulting bicategory $\Games{}(\anchorofasynch)$ 
is symmetric monoidal closed, and in fact $\ast$-autonomous.
As a cartesian and co-cartesian $\ast$-autonomous bicategory,
it defines a 2-categorical model of the multiplicative additive fragment (MALL) of linear logic,
where formula are interpreted as asynchronous template games
and proofs as strategies between them.
We conclude and indicate future directions in \S\ref{section/conclusion}.

\section{Asynchronous graphs}\label{section/asynchronous-graphs}
%After the necessary preliminaries on graphs in \S\ref{section/asynchronous-graph/graphs},
We start by describing in \S\ref{section/asynchronous-graph/definition}
the specific notion of \emph{asynchronous graph} which we shall use in the paper.
This leads us to define in \S\ref{section/asynchronous-graph/category}
the category $\Asynch$ of asynchronous graphs and of morphisms between them.
Every asynchronous graph $(G,\diamond)$ may be seen as a presentation 
by generators and relations of a specific 2-category noted $\asynchtwocat{G,\diamond}$.
%$\asynchtwocat{G,\diamond}$ or simply 
%
We construct a symmetric monoidal category of asynchronous graphs.

\subsection{Asynchronous graphs}\label{section/asynchronous-graph/definition}
A \emph{square} in a graph $G$ is defined as a pair $(p,q)$ of paths $p,q:x\transitionpath y$ 
of length 2, with the same source $x$ and the same target $y$.
An \emph{asynchronous graph} $(G,\diamond)$ is a graph~$G$ 
equipped with a set $\diamond$ of squares, satisfying a number
of additional properties described below.
We use the notation $p\diamond q$ when the square $(p,q)$ is an element of $\diamond$
and say in that case that the square $(p,q)$ defines a \emph{permutation tile.}
A permutation tile $u_1\cdot u_2 \diamond v_1 \cdot v_2$
between the paths $p=u_1\cdot u_2$ and $q=v_1 \cdot v_2$
is depicted in the following way, as a 2-dimensional surface or tile 
between the paths $p$ and $q$:
 % \vspace{-1.2em}
\begin{equation}\label{equation/permutation-tile-bicore}
\raisebox{-3.8em}{\includegraphics[height=7.6em]{figure2.png}}
\end{equation}
An asynchronous graph $(G,\diamond)$ is required to satisfy the following properties:
%\begin{itemize}
\medbreak

\noindent
1. \textbf{every permutation tile is symmetric:} 
for all paths $p,q:x\transitionpath y$ of length 2
with same source $x$ and same target $y$,
$p\diamond q$ implies that $q\diamond p$,

\medbreak

\noindent
2. \textbf{every permutation tile is deterministic:} 
for all paths $p,q,q':x\transitionpath y$
of length 2 with same source $x$ and same target $y$,
$p\diamond q$ and $p\diamond q'$ implies that $q=q'$,

\medbreak

\noindent
3. \textbf{the cube property:}
for all pairs $p,q:x\transitionpath y$ of paths of length 3
with same source $x$ and same target $y$, noted $p=u_1\cdot u_2\cdot u_3$
and $q=v_1\cdot v_2\cdot v_3$, there are edges $w_3$, $u'_2$, $v'_2$
and a sequence of permutation tiles
$$
u_2\cdot u_3 \diamond u'_2\cdot w_3
\quad\quad
u_1\cdot u'_2\diamond v_1\cdot v'_2
\quad\quad
v'_2\cdot w_3\diamond v_2\cdot v_3
$$
if and only if there are edges $w_1$, $u''_2$, $v''_2$
and a sequence of permutation tiles
$$
u_1\cdot u_2 \diamond w_1\cdot u''_2
\quad\quad
u''_2\cdot u_3 \diamond v''_2\cdot v_3
\quad\quad
w_1\cdot v''_2\diamond v_1\cdot v_2.
$$
%\end{itemize}
The cube property is nicely described by the following picture:
\begin{equation}\label{equation/cube}
\vspace{-1em}
\begin{tabular}{c}
\raisebox{-3em}{\includegraphics[height=6.5em]{figure9a.png}}
\, $\iff$ \,
\raisebox{-3em}{\includegraphics[height=6.5em]{figure10a.png}}
\end{tabular}
\end{equation}
\medbreak

\subsection{The category of asynchronous graphs}\label{section/asynchronous-graph/category}
In this section, we define the category $\Asynch$ of asynchronous graphs
and establish that the category has all finite limits.
%
%To that purpose, we start 
%We are now ready to construct the category $\Asynch$ of asynchronous graphs
An \emph{asynchronous homomorphism} between asynchronous graphs is defined as
\medbreak
\begin{definition}
An \emph{asynchronous graph homomorphism}, or \emph{asynchronous homomorphism},
\begin{equation}\label{equation/asynchronous-graph-homomorphism}
  f \quad : \quad (G,\diamond_G) \,\, \longrightarrow \,\, (H,\diamond_H)
\end{equation}
is a graph homomorphism $f: G\to H$ between the underlying graphs,
such that 
$$
u_1\cdot u_2 \, \diamond_G \, v_1\cdot v_2
\quad \Rightarrow \quad
f(u_1)\cdot f(u_2) \, \diamond_H \, f(v_1)\cdot f(v_2)
$$
for every pair of paths $p=u_1\cdot u_2$ and $q=v_1\cdot v_2$ 
of length~2 with same source and target.
\end{definition}

\medbreak

\noindent
The category $\Asynch$ is defined in the following way:
its objects are the asynchronous graphs and its morphisms
are the asynchronous homomorphisms between them.
%
%Before explaining how to translate asynchronous graphs
%into 2-categories in \S\ref{section/asynchronous-graph/from-ag-to-tc},
%we would like to make a number of basic observations on the category $\Asynch$.
%

\subsection{Finite limits of asynchronous graphs}\label{section/finite-limits-of-Asynch}
We establish now that the category $\Asynch$ has all finite limits.
We proceed in two steps.
A preliminary observation is that the category $\Asynch$ is cartesian,
with cartesian product defined as
$$(G,\diamond_G)\times(H,\diamond_H)
\quad = \quad (G\times H, \diamond_{G\times H})$$
where $G\times H$ denotes the cartesian product of the underlying graphs $G$ and $H$,
and where the set of permutation tiles $\diamond_{G\times H}=\diamond_{G}\times\diamond_{H}$ 
is defined as expected: every pair of paths 
$$p,q:(x,x')\transitionpath (y,y')$$
of length 2 in the graph $G\times H$ can be written as
$$
\begin{tikzcd}[column sep = 2em, row sep = -.2em]
p \,=\, (x,x')
\arrow[rr,"{(u_1,u'_1)}"]
&&
(z_1,z_1')
\arrow[rr,"{(u_2,u'_2)}"]
&&
(y,y')
\\
q \,=\, (x,x')
\arrow[rr,"{(v_1,v'_1)}"]
&&
(z_2,z_2')
\arrow[rr,"{(v_2,v'_2)}"]
&&
(y,y')
\end{tikzcd}
$$
One thus declares that the two paths $p$ and $q$ above 
define a permutation tile $p\diamond_{G\times H} q$
precisely when their projections define permutation tiles 
in $(G,\diamond_G)$ and in $(H,\diamond_H)$:
$$
u_1\cdot u_2 \diamond_{G} v_1\cdot v_2
\quad
\mbox{and}
\quad
u'_1\cdot u'_2 \diamond_{H} v'_1\cdot v'_2.
$$
In other words,
$$
p\diamond_{G\times H} q
\quad \iff \quad
\pi_1(p)\diamond_{G} \pi_1(q)
\,\,
\mbox{and}
\,\,
\pi_2(p)\diamond_{H} \pi_2(q)
$$
where $\pi_1:G\times H\to G$ and $\pi_2:G\times H\to H$ 
denote the two graph homomorphisms defined by projection.
The terminal object of $\Asynch$ 
%and unit $\textbf{1}=(\textbf{1},\diamond)$ of the cartesian product 
is the asynchronous graph $\moo[\tic]$
with one single vertex $\ast$, one single edge $\tic:\ast\to\ast$
and a permutation tile $\tic\cdot\tic\diamond\tic\cdot\tic$
which permutes the edge~$\tic$ with itself.

\medbreak

We then establish in the Appendix~\S\ref{appendix/finite-limits}
the less obvious property that the category $\Asynch$ has equalizers.
From this follows that:
\medbreak
\begin{proposition}\label{proposition/finite-limits-in-Asynch}
The category $\Asynch$ has finite products and equalizers,
and thus has all finite limits.
\end{proposition}

%\medbreak
%
%A second and less immediate observation
%%less immediate and made possible
%%by the careful way we carved our definition of asynchronous games,
%is that the category $\Asynch$ has all equalizers,
%and thus all finite limits.
%%
%Indeed, given a pair of homomorphisms between asynchronous graphs
%$$
%\begin{tikzcd}[column sep = 1.4em]
%(G,\diamond_G)
%\arrow[rr,"{f}",yshift=.2em]
%\arrow[rr,"{g}"{swap},yshift=-.2em]
%&&
%(H,\diamond_H)
%\end{tikzcd}
%$$
%the equalizer $(E,\diamond_E)$ may be defined as follows:
%its vertices are the vertices $x$ of the graph $G$ such that $f(x)=g(x)$
%and its edges are the edges $u:x\to y$ such that $f(u)=g(u)$.
%%
%Note that $E$ may be seen as a subgraph of $G$.
%%
%Then, a pair of paths $p,q:x\transitionpath y$ of length 2 in the graph $E$ 
%defines a permutation tile $p\diamond_E q$ precisely when 
%the two paths $p,q$ define such a permutation tile $p\diamond_G q$
%in the asynchronous graph $(G,\diamond_G)$.
%%
%We establish that defined in this way,
%$(E,\diamond_E)$ satisfies the properties
%required of an asynchronous graph, 
%see the Appendix, \S\ref{section/finite-limits-of-Asynch}.
%%

\section{The Gray tensor product}\label{section/gray-tensor-product}
For the sake of completeness, we provide 
in the Appendix (see \S\ref{section/gray-tensor})
a purely algebraic description by generators and relations of the Gray tensor product 
$\Atwocategory\graytensor\Btwocategory$ of a pair of 2-categories $\Atwocategory$ and $\Btwocategory$.
The definition is somewhat involved however,
and we thus find more convenient to describe below
a characterization of the Gray tensor product $\Acategory\graytensor\Bcategory$
of two small 2-categories $\Acategory$, $\Bcategory$ 
adapted from the work by Bourke and Gurski~\cite{bourke-gurski-acs-2017}.
Using this specific formulation, we establish the main result of the section,
which states that the Gray tensor product of 2-categories
$\Atwocategory,\Btwocategory\mapsto\Atwocategory\graytensor\Btwocategory$
preserves coreflexive equalizers componentwise.

\subsection{A concise characterization of the Gray tensor product}\label{section/gray-tensor/alternative-definition}
The Gray tensor product
$\Atwocategory,\Btwocategory\mapsto\Atwocategory\graytensor\Btwocategory$
%and we thus find more convenient to give here an alternative
%description of the 2-category $\Acategory\graytensor\Bcategory$,
of 2-categories has the remarkable property that its unit
coincides with the terminal 2-category $\mathbf{1}$.
%along an idea carefully developed by Bourke and Gurski \cite{bourke-gurski-acs-2017}.
%
From this follows that the Gray tensor product 
$\Atwocategory\graytensor\Btwocategory$
of a pair of 2-categories $\Atwocategory$ and $\Btwocategory$
comes equipped with a pair of ``projection'' 2-functors
$$
\begin{tikzcd}[column sep = 2em, row sep=2em]
{\Acategory\cong \Acategory\graytensor\mathbf{1}}
&&
\Acategory\graytensor\Bcategory
\arrow[ll,"{\pi_1}"{swap}]
\arrow[rr,"{\pi_2}"]
&&
{\mathbf{1}\graytensor\Bcategory\cong \Bcategory}
\end{tikzcd}
$$
which induce in turn a 2-functor
%every pair of 2-categories $\Acategory$, $\Bcategory$ induces a canonical functor 
\begin{equation}\label{equation/the-two-functor-q}
\begin{tikzcd}[column sep = 2em, row sep=2em]
q_{\Acategory,\Bcategory}
\quad : \quad
\Acategory\graytensor\Bcategory
\arrow[rr,"{}"]
&&
\Acategory\times\Bcategory
\end{tikzcd}
\end{equation}
The key observation at this stage is that the definition 
of Gray tensor product is carefully carved to ensure
that this 2-functor is \emph{locally fully faithful}.
This means that for every pair of 1-cells
$$
\begin{tikzcd}[column sep = 2em, row sep=2em]
(A,B)
\arrow[rr,yshift=.2em,"{f}"]
\arrow[rr,yshift=-.2em,"{g}"{swap}]
&&
(A',B')
\end{tikzcd}
$$
in the Gray tensor product $\Acategory\graytensor\Bcategory$,
there is a one-to-one correspondence between the hom-sets
$$
\begin{array}{cc}
& \Acategory\graytensor\Bcategory((A,B),(A',B'))(f,g)
\\
\cong & \Acategory(A,A')(f_1,g_1)
\times
\Acategory(B,B')(f_2,g_2)
\end{array}
$$
where the morphisms $f_1$, $g_1$ of $\Acategory$
and $f_2$, $g_2$ of $\Bcategory$ are defined by projection:
\begin{equation}\label{equation/notations-f1-f2-g1-g2}
\begin{tikzcd}[column sep = 2em, row sep=2em]
A
\arrow[rr,yshift=.2em,"{f_1=\pi_1(f)}"]
\arrow[rr,yshift=-.2em,"{g_1=\pi_1(g)}"{swap}]
&&
A'
\end{tikzcd}
\quad\quad
\begin{tikzcd}[column sep = 2em, row sep=2em]
B
\arrow[rr,yshift=.2em,"{f_2=\pi_2(f)}"]
\arrow[rr,yshift=-.2em,"{g_2=\pi_2(g)}"{swap}]
&&
B'
\end{tikzcd}
\end{equation}
In other words, a 2-cell 
in the 2-category $\Atwocategory\graytensor\Btwocategory$ 
of the form
$$
\begin{tikzcd}[column sep = .8em, row sep=2em]
\theta \quad :\quad
f\arrow[rr,double,-implies]
&&
g
\quad
:
\quad
(A,B)
\arrow[rr]
&&
(A',B')
\end{tikzcd}
$$
may be equivalently defined as a pair 
$$
\begin{array}{c}
\begin{tikzcd}[column sep = .8em, row sep=2em]
\theta_1 
\,\,
:
\,\,
f_1\arrow[rr,double,-implies]
&&
g_1
\,\,
:
\,\,
A
\arrow[rr]
&&
A'
\end{tikzcd}
\\
\begin{tikzcd}[column sep = .8em, row sep=2em]
\theta_2 
\,\,
:
\,\,
f_2\arrow[rr,double,-implies]
&&
g_2
\,\,
:
\,\,
B
\arrow[rr]
&&
B'
\end{tikzcd}
\end{array}
$$
using the notations $f_1, f_2,g_1,g_2$ given in (\ref{equation/notations-f1-f2-g1-g2}).
%
%This observation gives a more direct way to describe the 2-cells of $\Acategory\graytensor\Bcategory$
%as pairs of 2-cells of $\Acategory$ and $\Bcategory$.
%
It is possible to define the Gray tensor product $\Acategory\graytensor\Bcategory$ 
directly from there, along an idea developed by Bourke and Gurski in \cite{bourke-gurski-acs-2017}.
Every small 2-category $\Atwocategory$ comes equipped 
with an underlying category of objects and morphisms noted $\underlyingcat{\Atwocategory}$
and with an underlying set of objects noted $\underlyingset{\Atwocategory}$.
The category $\underlyingcat{\Atwocategory}$ may be seen 
as a \emph{locally discrete} 2-category,
while the set $\underlyingset{\Atwocategory}$ may be seen 
as a discrete 2-category.
As such, they come equipped with a pair of canonical 2-functors
$$
\begin{tikzcd}[column sep = 1.8em]
{\underlyingset{\Acategory}} \arrow[rr,"inj"]
&&
{\underlyingcat{\Acategory}} \arrow[rr,"inj"]
&&
\Acategory
\end{tikzcd}
$$
From this follows that there exists a 2-functor
\begin{equation}\label{equation/composite-two-functor}
\begin{tikzcd}[column sep = 1.6em]
{\underlyingcat{\Atwocategory}\funnytensor\underlyingcat{\Btwocategory}}
\arrow[r]
&
{\underlyingcat{\Atwocategory}\times\underlyingcat{\Btwocategory}}
\arrow[rrr,"{inj\times inj}"]
&&&
{\Atwocategory\times\Btwocategory}
\end{tikzcd}
\end{equation}
where the funny tensor product $\underlyingcat{\Atwocategory}\funnytensor\underlyingcat{\Btwocategory}$
of the categories $\underlyingcat{\Acategory}$ and $\underlyingcat{\Bcategory}$
is defined as the pushout of the diagram below:
$$
\begin{small}
\begin{tikzcd}[column sep = .8em, row sep=.8em]
{\underlyingset{\Atwocategory}\times \underlyingset{\Btwocategory}}
\arrow[rr,"{\underlyingcat{\Atwocategory}\times inj}"]
\arrow[dd,"{inj\times\underlyingcat{\Btwocategory}}"{swap}]
&&
{\underlyingset{\Atwocategory}\times \underlyingcat{\Btwocategory}}
\arrow[dd]
\\
& pushout
\\
{\underlyingcat{\Atwocategory}\times \underlyingset{\Btwocategory}}
\arrow[rr]
&&
{\underlyingcat{\Atwocategory}\funnytensor \underlyingcat{\Btwocategory}}
\end{tikzcd}
\end{small}
$$
computed in the category $\Cat$, see \cite{bourke-gurski-acs-2017} for details.
The Gray tensor product ${\Atwocategory\graytensor\Btwocategory}$
is then characterized (or defined) as the unique 2-category such that the pair of 2-functors 
$$
\begin{tikzcd}[column sep = 1.8em]
{\underlyingcat{\Atwocategory}\funnytensor\underlyingcat{\Btwocategory}}
\arrow[rr,"{(a)}"]
&&
{\Atwocategory\graytensor\Btwocategory}
\arrow[rr,"{(b)}"]
&&
{\Atwocategory\times\Btwocategory}
\end{tikzcd}
$$
factors the composite 2-functor~(\ref{equation/composite-two-functor})
in such a way that the 2-functor $(a)$ 
is an identity-on-objects and identity-on-morphisms 2-functor
from ${\underlyingcat{\Atwocategory}\funnytensor\underlyingcat{\Btwocategory}}$
to ${\Atwocategory\graytensor\Btwocategory}$
and the 2-functor $(b)$ 
is a locally fully faithful 2-functor from ${\Atwocategory\graytensor\Btwocategory}$ to 
${\Atwocategory\times\Btwocategory}$.
Note that the 2-functor $(b)$ coincides with the 2-functor $q_{\Acategory,\Bcategory}$
mentioned in~(\ref{equation/the-two-functor-q}).

%removing the generators (\ref{equation/gray-cells})
%as well as [b], [c] and [d].
\subsection{Gray tensor product preserves coreflexive equalizers}\label{section/gray-preserves-equalizers}
A pair of 2-functors 
$$
\begin{tikzcd}[column sep = 2em]
\Atwocategory
\arrow[rr,"{F}",yshift=.2em]
\arrow[rr,"{G}"{swap},yshift=-.2em]
&&
\Btwocategory
\end{tikzcd}
$$
is called coreflexive when there exists a 2-functor $S:\Btwocategory\to\Acategory$
such that $S\circ F=S\circ G=\id{\Acategory}$.
An equalizer~$\Ecategory$ of a coreflexive pair of 2-functors $F,G:\Acategory\to\Bcategory$
in the category $\TwoCat$ is called a coreflexive equalizer.
Suppose that $\Ecategory$ is an equalizer of a coreflexive pair,
as shown in the diagram below:
$$
\begin{tikzcd}[column sep = 2em]
\Etwocategory
\arrow[rr,dashed,"{E}"]
&&
\Atwocategory
\arrow[rr,"{F}",yshift=.2em]
\arrow[rr,"{G}"{swap},yshift=-.2em]
&&
\Btwocategory
\arrow[ll,"{S}"{swap},yshift=-.2em,bend right=40]
\end{tikzcd}
$$
We establish that 
\medbreak
\begin{proposition}\label{proposition/preservation-of-equalizers}
The Gray tensor product preserves coreflexive equalizers
in the sense that for every 2-category $\Ctwocategory$, the diagram
\vspace{-1em}
\begin{equation}\label{equation/equalizer-between-gray-tensor-products}
\begin{tikzcd}[column sep = 2em]
\Ctwocategory\graytensor\Etwocategory
\arrow[rr,dashed,"{\Ccategory\graytensor E}"]
&&
\Ctwocategory\graytensor\Atwocategory
\arrow[rr,"{\Ctwocategory\graytensor F}",yshift=.2em]
\arrow[rr,"{\Ctwocategory\graytensor G}"{swap},yshift=-.2em]
&&
\Ctwocategory\graytensor\Btwocategory
\arrow[ll,"{\Ctwocategory\graytensor H}"{swap},yshift=-.2em,bend right=40]
\end{tikzcd}
\end{equation}
exhibits the 2-category $\Ctwocategory\graytensor\Ecategory$ 
as a coreflexive equalizer 
of the 2-functors $\Ctwocategory\graytensor F, \Ccategory\graytensor G:\Ccategory\graytensor\Acategory\to\Ccategory\graytensor\Bcategory$.
\end{proposition}
\medbreak
The fact is proved in three combined steps 
carefully described in the Appendix~\ref{appendix/proof-of-preservation-of-equalizers}.

%%%%%%%%%%%%%%%%%%%%%%%%%%%%%%%%%%%%%%%

\section{How every asynchronous graph can be seen as a 2-category, functorially}\label{section/the-translation}
%\section{Asynchronous graphs translated into 2-categories, functorially}
% of positions, trajectories and reshufflings}
%A functorial translation from asynchronous graphs to 2-categories}
We explain in this section how to define the functor
%We explained in the introduction, 
%and more specifically in~(\ref{equation/asynchtwocat}), 
%how to construct a functor
$$
%\asynchtwocat{-} \,\,\, : \,\,\, (\Asynch,\shuffletensor,\textbf{I}) \,\,\, \longrightarrow \,\,\, (\TwoCat,\graytensor,\textbf{I})
\asynchtwocat{-} \,\,\, : \,\,\, \Asynch \,\,\, \longrightarrow \,\,\, \TwoCat
$$
which associates to every asynchronous graph $(G,\diamond)$
the 2-category $\asynchtwocat{G,\diamond}$ whose objects
and morphisms are the vertices and paths of the graph~$G$, 
and whose 2-cells are \emph{reshufflings} as defined in~\S\ref{section/reshuffling}.
%
%$\varphi:f\Rightarrow g:x\to y$ between paths of the graph~$G$,
%
%We then give an alternative and purely algebraic description
%of the 2-category~$\asynchtwocat{G,\diamond}$
%by generators and relations in \S\ref{section/alternative-definition}.
%
We finally observe in \S\ref{section/shuffle-vs-gray}
that there exists a coherent family 
of isomorphisms between 2-categories
$$
\asynchtwocat{G\shuffletensor H,\diamond_{G\shuffletensor H}}
\cong
\asynchtwocat{G,\diamond_G}\graytensor\asynchtwocat{H,\diamond_H}
\quad\quad
\asynchtwocat{\mathbf{I}}
\cong
\mathbf{1}
$$
where $\mathbf{I}$ denotes the neutral element of the shuffle tensor product
of asynchronous graphs, and $\mathbf{1}$ denotes the terminal 2-category.
This provides a firm conceptual foundation to the intuition
that the Gray tensor product should be understood as
an shuffle (or asynchronous) tensor product of 2-categories.

%As we will see in \S\ref{section/asynchronous-to-gray}, 
%the functor $\asynchtwocat{-}$ has the remarkable property
%that it transports the asynchronous tensor product $G,H\mapsto G\shuffletensor H$
%of asynchronous graphs to the Gray tensor product $\Atwocategory,\Btwocategory\mapsto \Atwocategory\graytensor \Btwocategory$ of 2-categories,
%in the sense that 
%
%whose objects 
%are the vertices of $G$, whose morphisms are the paths $f,g$ of $G$,
%and whose 2-cells $\varphi:f\Rightarrow g$ are \emph{reshufflings}
%%$$\begin{tikzcd}[column sep = 1em]
%%\varphi \quad : \quad f\arrow[rr,-implies,double] && g \quad : \quad x\arrow[rr] && y
%%\end{tikzcd}
%%$$
%between paths $f$ and $g$ of same length in $G$,
%as defined in \S\ref{section/reshuffling}.
%
\subsection{Reshufflings between paths of an asynchronous graph}\label{section/reshuffling}
Given an asynchronous graph $(G,\diamond)$,
a reshuffling 
$$\begin{tikzcd}[column sep = 1em]
\varphi \quad : \quad f\arrow[rr,-implies,double] && g \quad : \quad x\arrow[rr] && y
\end{tikzcd}
$$
between two paths $f$, $g$ of the graph~$G$
with same source~$x$, same target~$y$, and same length
$$
\begin{array}{c}
\begin{tikzcd}[column sep = 1.2em]
f \,\,\,\,\, = \,\,\,\,\,  x \arrow[rr,"{u_1}"] && z_1\arrow[rr,"{u_2}"] && \cdots \arrow[rr,"{u_{k-1}}"]
&& z_{k-1} \arrow[rr,"{u_k}"] && y 
\end{tikzcd}
\\
\begin{tikzcd}[column sep = 1.2em]
g \,\,\,\,\, = \,\,\,\,\, x \arrow[rr,"{v_1}"] && z'_1\arrow[rr,"{v_2}"] && \cdots \arrow[rr,"{v_{k-1}}"]
&& z'_{k-1} \arrow[rr,"{v_k}"] && y 
\end{tikzcd}
\end{array}
$$
is defined as a bijection 
\begin{equation}\label{equation/reshuffling}
\varphi \,\,\, : \,\,\, \{1,\dots,k\} \,\,\, \longrightarrow \,\,\, \{1,\dots,k\}
\end{equation}
%$$\varphi:\{1,\dots,k\} \longrightarrow \{1,\dots,k\}$$
tracking a sequence of permutation tiles in $\diamond$
transforming the path $f$ into the path $g$.
Typically, the permutation tile~(\ref{equation/permutation-tile-bicore})
%{equation/from-permutation-tiles-to-reshufflings}) 
is tracked
by the reshuffling $\varphi:\{1,2\}\mapsto \{1,2\}$ defined 
as $1\mapsto 2$ and $2\mapsto 1$
and represented below using the two arrows:
\begin{equation}\label{equation/from-permutation-tiles-to-reshufflings}
\raisebox{-3.8em}{\includegraphics[height=7.6em]{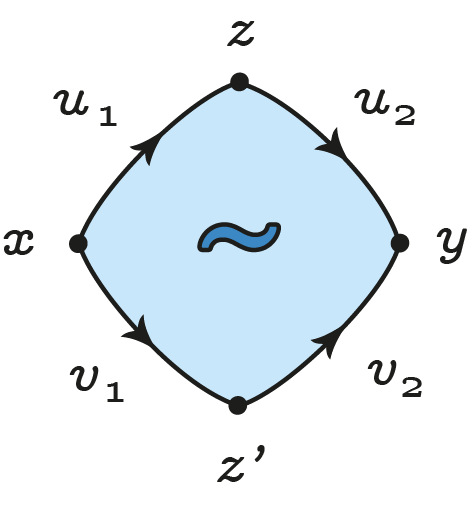}}
\quad \mapsto \quad
\raisebox{-3.8em}{\includegraphics[height=7.6em]{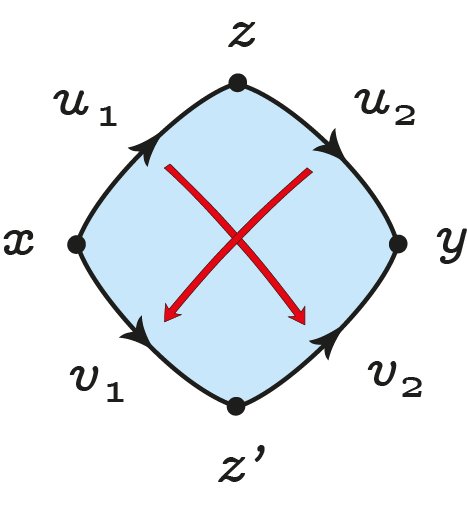}}
% \vspace{-0em}
\end{equation}
Similarly, the sequence of three permutations below
is tracked by the reshuffling $\varphi:\{1,2,3\}\mapsto \{1,2,3\}$
defined as $1\mapsto 3$, $2\mapsto 3$, $3\mapsto 1$,
and represented using the three arrows:
\begin{equation}\label{equation/from-permutation-tiles-to-reshufflings-cube}
\raisebox{-3.6em}{\includegraphics[height=7.4em]{figure9a.png}}
\hspace{.5em} \mapsto \hspace{.5em} 
\raisebox{-3.6em}{\includegraphics[height=7.4em]{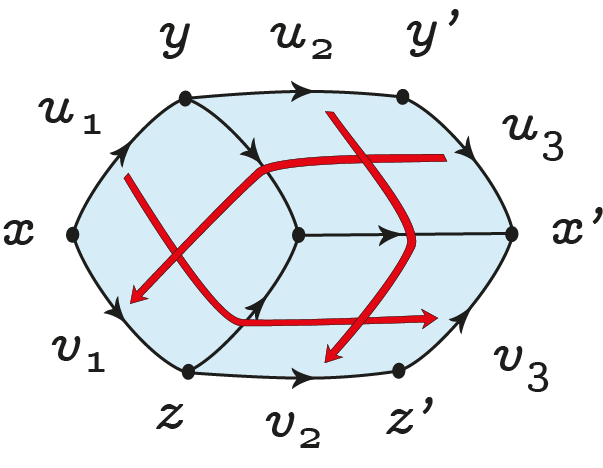}}
% \vspace{-0em}
\end{equation}
One main benefit of using the notion of reshuffling is that 
it enables  to identify the two sequences of permutation tiles below 
%(and which may be seen as half-surfaces of the cube) 
as the very same reshuffling $\varphi:\{1,2,3\}\mapsto \{1,2,3\}$:
\begin{equation}\label{equation/cubical-tiles}
\begin{small}
\begin{array}{ccc}
\raisebox{-3.8em}{\includegraphics[height=8em]{figure9b.png}}
& \cong & 
\raisebox{-3.8em}{\includegraphics[height=8em]{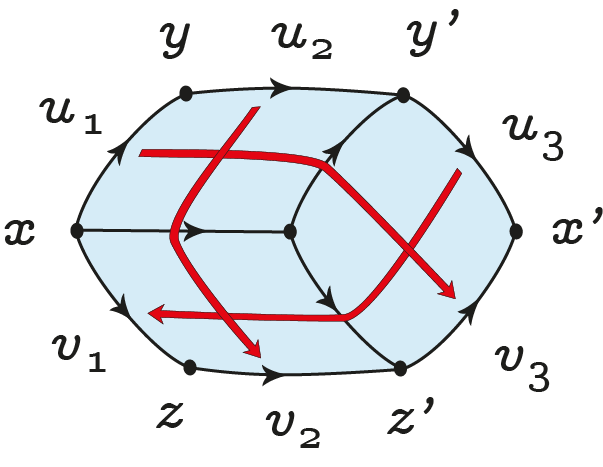}}
\end{array}
\end{small}
\end{equation}
Similarly, the sequence of two inverse permutation tiles
on the left-hand side is tracked by the identity reshuffling $\varphi:\{1,2\}\to\{1,2\}$
and thus identified to the empty sequence of permutation tiles
on the right-hand side:
\begin{equation}\label{equation/reverse-tiles}
\begin{small}
\begin{array}{ccc}
\raisebox{-3.8em}{\includegraphics[height=8em]{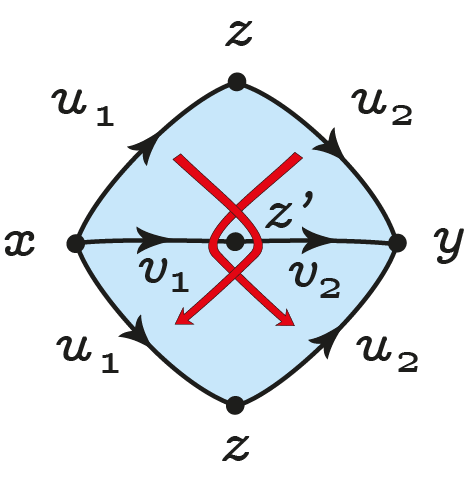}}
& \cong & 
\raisebox{-3.8em}{\includegraphics[height=8em]{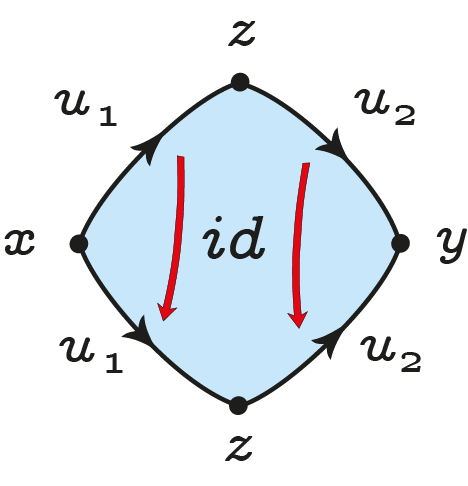}}
\end{array}
\end{small}
\end{equation}
Although the notion of a reshuffling~(\ref{equation/reshuffling}) tracking a sequence
of permutation tiles in the asynchronous graph~$(G,\diamond)$ should be intuitive
at this stage, we give a detailed and explicit definition in the Appendix,
see \S\ref{section/construction-of-the-two-category}.
%
%As a short note to the more categorically-minded readers,
\medbreak

\noindent
At a more conceptual level, the 2-category~$\asynchtwocat{G,\diamond}$
may be characterized as the 2-category freely generated 
by the vertices and edges of the graph $G$ at the dimensions~$0$ and~$1$,
by the family of 2-cells
$$\begin{tikzcd}[column sep = 1em]
\sigma_{f,g} \quad : \quad {f}\arrow[rr,-implies,double] && {g} \quad : \quad x\arrow[rr] && y
\end{tikzcd}
$$
indexed by set of permutation tiles $f \diamond g$ in~$(G,\diamond)$
at dimension~$2$, modulo the two families of equations between 2-cells depicted in the diagrammatic notations
of~(\ref{equation/cubical-tiles}) and~(\ref{equation/reverse-tiles}).
%
%The categorically-minded reader will also observe that
%\begin{proposition}
%The 2-category $\asynchtwocat{G,\diamond}$ is the 2-category freely generated
%by the graph $G$ and the tiles $\diamond$ modulo the equations (...) and (...).
%\end{proposition}

%\subsection{An alternative definition of the 2-category $\asynchtwocat{G,\diamond}$}\label{section/alternative-definition}
%It is worth mentioning that 
%A nice and conceptual description of the 2-category $\asynchtwocat{G,\diamond}$ is possible, by generators and relations.

%
%\subsection{The symmetric monoidal structure}
%\begin{equation}\label{equation/functor-from-asynch-to-twocat}
%\begin{tikzcd}[column sep = 1em]
%\asynchtwocat{-} 
%\quad : \quad 
%\Asynch
%\arrow[rr]
%&&
%\TwoCat
%\end{tikzcd}
%\end{equation}
%from the category $\Asynch$ of asynchronous graphs 
%to the category $\TwoCat$ of small 2-categories.
%%

%

\subsection{The symmetric monoidal functor}\label{section/shuffle-vs-gray}
Once the functor $\asynchtwocat{-}$ explicated,
it is not difficult to establish that it defines in fact a symmetric monoidal functor
$$
\asynchtwocat{-} \,\,\,\,\, : \,\,\,\,\, {(\Asynch,\shuffletensor,\mathbf{I})}
\ \,\,\, \longrightarrow \,\,\, {(\TwoCat,\graytensor,\grayunit)}
$$
which preserves the shuffle tensor product up to isomorphisms
$$
\asynchtwocat{G\shuffletensor H,\diamond_{G\shuffletensor H}}
\cong
\asynchtwocat{G,\diamond_G}\graytensor\asynchtwocat{H,\diamond_H}
\quad\quad
\asynchtwocat{\mathbf{I}}
\cong
\mathbf{1}
$$
in the category $\TwoCat$.
The existence and coherence laws of the isomorphisms
may be deduced from the explicit definition in~\S\ref{section/reshuffling}
of the functor $\asynchtwocat{-}$ combined with the description
of the Gray tensor product in \S\ref{section/gray-tensor/alternative-definition}.
We also note that although the functor $\asynchtwocat{-}$
does not preserve cartesian products, it does preserve equalizers.

\section{The weak double category of bicomodules}\label{section/weak-double-category-of-bicomodules}
Given a monoidal category $(\Scategorysurround,\graytensor,\grayunit)$ 
with coreflexive equalizers,
we construct the weak double category $\Comod{\Scategorysurround}$ of comonoids
and bicomodules between them.
In order to perform the construction, we make
the assumption that the tensor product preserves coreflexive equalizers componentwise,
in the sense that for every object~$A$ of the category~$\Scategorysurround$,
the functors 
$$A\graytensor -:\Scategorysurround\to\Scategorysurround
\quad\quad
-\graytensor A:\Scategorysurround\to\Scategorysurround$$
preserve coreflexive equalizers.
%
%This means that given any equalizer diagram
%$$
%\begin{tikzcd}
%E 
%\arrow[rrr,"h"]
%&&&
%X
%\arrow[rrr,yshift=.3em,"{f}"]
%\arrow[rrr,yshift=-.3em,"{g}"{swap}]
%&&&
%Y
%\end{tikzcd}
%$$
%%
%$$
%\begin{tikzcd}
%A\graytensor E 
%\arrow[rrr,"A\graytensor equalizer"]
%&&&
%A\graytensor X
%\arrow[rrr,yshift=.3em,"{A\graytensor f}"]
%\arrow[rrr,yshift=-.3em,"{A\graytensor g}"{swap}]
%&&&
%A\graytensor Y
%\end{tikzcd}
%$$
%is an equalizer diagram.
We have in mind the example of the 2-category $\Scategorysurround=\TwoCat$ 
of small 2-categories, equipped with the Gray tensor product, 
see \S\ref{section/gray-preserves-equalizers}.
We define in \S\ref{section/vertical-maps} and \S\ref{section/horizontal-maps}
the vertical and horizontal maps of the weak double category $\Comod{\Scategorysurround}$
and then describe its double cells in \S\ref{section/double-cells}.

%, noted $\graytensor$.
%
\subsection{Vertical maps in $\Comod{\Scategorysurround}$}\label{section/vertical-maps}
The vertical category of $\Comod{\Scategorysurround}$ is defined
as the category $\Comonoid{\Scategorysurround,\graytensor,\grayunit}$
of comonoids in the monoidal category $(\Scategorysurround,\graytensor,\grayunit)$.
Recall that a comonoid $(A,d,e)$ is a triple consisting of an object~$A$ of $\Scategorysurround$
and of a pair of morphisms
$$
\begin{tikzcd}[column sep = 1em]
d\,\, : \,\, A \arrow[rr] && {A\graytensor A}
\end{tikzcd}
\quad\quad\quad
\begin{tikzcd}[column sep = 1em]
e\,\, : \,\, A \arrow[rr] && {\grayunit}
\end{tikzcd}
$$
making the diagrams commute:
$$
\begin{small}
\begin{tikzcd}[column sep = 1.5em, row sep = 1.5em]
A\arrow[rr,"d"] \arrow[dd,"d"{swap}] 
&& {A\graytensor A}\arrow[dd,"{A\graytensor d}"]
\\
\\
{A\graytensor A}\arrow[rr,"{d\graytensor A}"]
&&
{A\graytensor A\graytensor A}
\end{tikzcd}
\quad\quad
\begin{tikzcd}[column sep = 1.5em, row sep = 1.5em]
&& A\arrow[dd,"d"] \arrow[rrdd,"{id_A}"] \arrow[lldd,"{id_A}"{swap}] 
&&
\\
\\
A
&&
{A\graytensor A}\arrow[rr,"{e\graytensor A}"]\arrow[ll,"{A\graytensor e}"{swap}]
&&
A
\end{tikzcd}
\end{small}
$$
A comonoid homomorphism
% between comonoids $(A,d_A,e_A)$ and $(B,d_B,e_B)$,
% $h:(A,d_A,e_A)\longrightarrow (B,d_B,e_B)
$$\begin{tikzcd}[column sep = 1em]
h\,\, : \,\, (A,d_A,e_A)\arrow[rr] && (B,d_B,e_B)
\end{tikzcd}
$$
is defined a morphism $h:A\to B$ making the two diagrams below commute:
$$\begin{small}
\begin{tikzcd}[column sep = 3em, row sep = 1em]
A\arrow[rr,"{h}"]\arrow[dd,"{d_A}"{swap}]
&&
B
\arrow[dd,"{d_B}"]
\\
\\
A\graytensor A\arrow[rr,"{h\graytensor h}"]
&&
{B\graytensor B}
\end{tikzcd}
\quad\quad
\begin{tikzcd}[column sep = 3em, row sep = 1em]
A\arrow[rr,"{h}"]\arrow[dd,"{e_A}"{swap}]
&&
B
\arrow[dd,"{e_B}"]
\\
\\
{\grayunit}\arrow[rr,"{id}"]
&&
{\grayunit}
\end{tikzcd}
\end{small}
$$
The category $\Comonoid{\Scategorysurround,\graytensor,\grayunit}$
is defined as the category whose objects are comonoids
and whose morphisms are comonoid homomorphisms.

\subsection{Horizontal maps in $\Comod{\Scategorysurround}$}\label{section/horizontal-maps}
We recall the following definition.
\begin{definition}
An $A,B$-comodule 
$$
\begin{tikzcd}[column sep = 1.5em, row sep=1em]
S \quad : \quad A\arrow[rr,spanmap]
&&
B
\end{tikzcd}
$$
between two comonoids $(A,\comult{A},\counit{A})$ and $(B,\comult{B},\counit{B})$
is defined as an object $S$ equipped with a morphism
$$
\begin{tikzcd}[column sep = 1.2em, row sep=.8em]
\coaction{S} \quad : \quad S\arrow[rr]
&&
A\graytensor S\graytensor B
\end{tikzcd}
$$
such that the two diagrams commute
\begin{equation}\label{equation/coactions}
\begin{small}
\begin{tikzcd}[column sep = -2em, row sep=.9em]
&&
S\arrow[rrdd,"{\coaction{S}}", bend left=20]\arrow[lldd,"{\coaction{S}}"{swap}, bend right=20]
\\
\\
A\graytensor S\graytensor B
\arrow[rrdd,"{\coaction{S}}"{swap,near start}, bend right=10]
&&
&&
A\graytensor S\graytensor B
\arrow[lldd,"{\comult{A}\graytensor S\graytensor \comult{B}}"{near start}, bend left=10]
\\
\\
&&
A\graytensor A\graytensor S\graytensor B\graytensor B
\end{tikzcd}
\quad
\begin{tikzcd}[column sep = 0em, row sep=.9em]
&&
S\arrow[dddd,"{id_S}"{swap}]\arrow[rrdd,"{\coaction{S}}", bend left=20]
\\
\\
&&
&&
A\graytensor S\graytensor B
\arrow[lldd,"{\counit{A}\graytensor S\graytensor \counit{B}}", bend left=10]
\\
\\
&&
S
\end{tikzcd}
\end{small}
\end{equation}
\end{definition}
Given an $A,B$-comodule $S$, we use $\coactionleft{S}$ and $\coactionright{S}$
as notations for the composite morphisms:
$$
\begin{array}{c}
\begin{tikzcd}[column sep = 2em]
\coactionleft{S} \,\, : \,\,
S
\arrow[rr,"{\coaction{S}}"]
&&
A\graytensor S\graytensor B
\arrow[rr,"{A\graytensor S\graytensor\counit{B}}"]
&&
A\graytensor S
\end{tikzcd}
\\
\begin{tikzcd}[column sep = 2em]
\coactionright{S} \,\, : \,\,
S
\arrow[rr,"{\coaction{S}}"]
&&
A\graytensor S\graytensor B
\arrow[rr,"{\counit{A}\graytensor S\graytensor B}"]
&&
S\graytensor B
\end{tikzcd}
\end{array}$$
which define a left comodule structure on the comonoid $A$
and a right comodule structure on the comonoid $B$, respectively.
Given three comonoids $A$, $B$ and $C$,
the composite of an $A,B$-comodule $S$ and of a $B,C$-comodule $T$
$$
\begin{tikzcd}[column sep = 1.5em, row sep=1em]
S \,\, : \,\, A\arrow[rr,spanmap]
&&
B
\end{tikzcd}
\quad\quad
\begin{tikzcd}[column sep = 1.5em, row sep=1em]
T \,\, : \,\, B\arrow[rr,spanmap]
&&
C
\end{tikzcd}
$$
is the $A,C$-comodule 
$$
\begin{tikzcd}[column sep = 1.5em, row sep=1em]
S\graytensoreq{B}T \quad : \quad A\arrow[rr,spanmap]
&&
C
\end{tikzcd}
$$
defined as the equalizer of the coreflexive pair of morphisms
$\coactionright{S}\graytensor T$ and $S\graytensor\coactionleft{T}$
$$
\begin{tikzcd}[column sep=1.6em]
S\graytensoreq{B} T
\arrow[rrr,"equalizer"]
&&&
S\graytensor T 
\arrow[rrr,yshift=.3em,"{\coactionright{S}\graytensor T}"]
\arrow[rrr,yshift=-.3em,"{S\graytensor\coactionleft{T}}"{swap}]
&&&
S\graytensor B\graytensor T
\arrow[lll,"{S\graytensor e\graytensor T}"{swap},yshift=-.2em,bend right=50]
\end{tikzcd}
$$
%from $S\graytensor T$ to $S\graytensor B\graytensor T$.
%computed in the 2-category $\Scategorysurround$.
%
The $A,C$-comodule structure of $S\graytensoreq{B} T$
is given by the morphism $\coaction{S\graytensoreq{B} T}$
defined as the unique solution of the universal problem
represented below:
$$
\hspace{-.5em}
\begin{footnotesize}
\begin{tikzcd}[column sep=1em]
S\graytensoreq{B} T
\arrow[rr,"{equ}"]
\arrow[dd,dashed,"{\coaction{S\graytensoreq{B} T}}"{description}]
&&
S\graytensor T 
\arrow[rrr,yshift=.3em,"{\coactionright{S}\graytensor T}"]
\arrow[rrr,yshift=-.3em,"{S\graytensor\coactionleft{T}}"{swap}]
\arrow[dd,"{\coactionleft{S}\graytensor\coactionright{T}}"{swap}]
&&&
S\graytensor B\graytensor T
\arrow[dd,"{\coactionleft{S}\graytensor B\graytensor\coactionright{T}}"{description}]
\\
\\
A\graytensor(S\graytensoreq{B} T)\graytensor C
\arrow[rr,"{equ}"]
&&
A\graytensor S\graytensor T \graytensor C
\arrow[rrr,yshift=.3em,"{A\graytensor\coactionright{S}\graytensor T\graytensor C}"]
\arrow[rrr,yshift=-.3em,"{A\graytensor S\graytensor\coactionleft{T}\graytensor C}"{swap}]
&&&
A\graytensor S\graytensor B\graytensor T\graytensor C
\end{tikzcd}
\end{footnotesize}
$$
\subsection{Double cells in $\Comod{\Scategorysurround}$}\label{section/double-cells}
A double cell in the double category $\Comod{\Scategorysurround}$
of the form
\begin{equation}
%\begin{small}
\begin{tikzcd}[column sep = 2.4em, row sep=.8em]
A 
\arrow[dd,"{\fsubA}"{swap}]
\arrow[rr,spanmap,"{S}"{yshift=.4em}, ""{yshift=-.4em,swap,name=source}] 
&&
B
\arrow[dd,"{\fsubB}"]
\\
\\
A'
\arrow[rr,spanmap,"{S'}"{yshift=-.4em,swap},""{yshift=.4em,name=target}]
&&
B'
\arrow[from=source, to=target, "{\theta}", double,-implies]
\end{tikzcd}
%\end{small}
\end{equation}
with horizontal edges an $A,B$-comodule $S$ 
and an $A',B'$-comonoid $S'$, and with vertical edges a pair of comonoid
homomorphisms $\fsubA:A\to A'$ and $\fsubB:B\to B'$ 
is defined as a morphism
$$\begin{tikzcd}[column sep = 1.5em, row sep=1em]
\theta \quad : \quad S \arrow[rr] && S'
\end{tikzcd}$$
of the category $\Scategorysurround$ making the diagram below commute:
\begin{equation}\label{equation/double-cell-as-morphism}
\begin{tikzcd}[column sep = 3em, row sep=1em]
S\arrow[rr,"{\theta}"] \arrow[dd,"{\coaction{S}}"{swap}]
&& 
{S'}
\arrow[dd,"{\coaction{S'}}"]
\\
\\
A\graytensor S\graytensor B \arrow[rr,"{\fsubA\graytensor\theta\graytensor\fsubB}"]
&&
{A'\graytensor S'\graytensor B'}
\end{tikzcd}
\end{equation}
Given a pair of double cells
\begin{equation}
\begin{small}
\begin{tikzcd}[column sep = 1.6em, row sep=1em]
A 
\arrow[dd,"{\fsubA}"{swap}]
\arrow[rr,spanmap,"{S}"{yshift=.4em}, ""{yshift=-.4em,swap,name=source}] 
&&
B
\arrow[dd,"{\fsubB}"]
\\
\\
A'
\arrow[rr,spanmap,"{S'}"{yshift=-.4em,swap},""{yshift=.4em,name=target}]
&&
B'
\arrow[from=source, to=target, "{\varphi}", double,-implies]
\end{tikzcd}
\quad\quad\quad
\begin{tikzcd}[column sep = 1.6em, row sep=1em]
B
\arrow[dd,"{\fsubB}"{swap}]
\arrow[rr,spanmap,"{T}"{yshift=.4em}, ""{yshift=-.4em,swap,name=source}] 
&&
C
\arrow[dd,"{\fsubC}"]
\\
\\
B'
\arrow[rr,spanmap,"{T'}"{yshift=-.4em,swap},""{yshift=.4em,name=target}]
&&
C'
\arrow[from=source, to=target, "{\psi}", double,-implies]
\end{tikzcd}
\end{small}
\end{equation}
the composite double cell 
\begin{equation}
\begin{small}
\begin{tikzcd}[column sep = 2.8em, row sep=1em]
A 
\arrow[dd,"{\fsubA}"{swap}]
\arrow[rr,spanmap,"{S\graytensoreq{B} T}"{yshift=.4em}, ""{yshift=-.4em,swap,name=source}] 
&&
C
\arrow[dd,"{\fsubC}"]
\\
\\
A'
\arrow[rr,spanmap,"{S'\graytensoreq{B'}T'}"{yshift=-.4em,swap},""{yshift=.4em,name=target}]
&&
C'
\arrow[from=source, to=target, "{\varphi\graytensoreq{\fsubB}\psi}", double,-implies]
\end{tikzcd}
\end{small}
\end{equation}
is defined as the unique morphism $\varphi\graytensoreq{\fsubB}\psi$
making the diagram below commute:
\begin{center}
$
\begin{tikzcd}[column sep=1.5em, row sep=1em]
S\graytensoreq{B} T\arrow[rr,"equ"]
\arrow[dd,"{\varphi\graytensoreq{\fsubB}\psi}"{swap},dashed]
&&
S\graytensor T \arrow[rrr,yshift=.3em,"{\coactionright{S}\graytensor T}"]
\arrow[rrr,yshift=-.3em,"{S\graytensor\coactionleft{T}}"{swap}]
\arrow[dd,"{\varphi\graytensor\psi}"{swap}]
&&&
S\graytensor B\graytensor T
\arrow[dd,"{\varphi\graytensor\fsubB\graytensor\psi}"]
\\
\\
{S'\graytensoreq{B'} T'}
\arrow[rr,"equ"]
&&
{S'\graytensor T'} \arrow[rrr,yshift=.3em,"{\coactionright{S'}\graytensor T'}"]\arrow[rrr,yshift=-.3em,"{S'\graytensor\coactionleft{T'}}"{swap}]
&&&
{S'\graytensor B'\graytensor T'}
\end{tikzcd}
$
\end{center}
The following result is essentially folklore, especially
when one sees the construction of $\Comod{\Scategorysurround}$
as the dual of the construction of a weak double category of bimodules
between monoids, see for instance~\cite{Street07}.
\medbreak
\begin{theorem}
$\Comod{\Scategorysurround}$ defines a weak double category.
\end{theorem}
\medbreak
Note that the assumption that the tensor product 
preserves coreflexive equalizers componentwise in $\Scategorysurround$
plays a critical role in the proof that horizontal composition is associative
in $\Comod{\Scategorysurround}$.

\section{The template $\anchorofasynch$ of asynchronous games}\label{section/asynchronous-template}
In this section, we recall in \S\ref{section/internal-category}
the notion of an internal category in a monoidal category $(\Scategorysurround,\graytensor,\grayunit)$
%whose tensor product preserves equalizers componentwise,
as defined by Aguiar~\cite{aguiar-phd-1997}.
We then show that the template $\anchorofasynch$
of asynchronous games formulated in the introduction
defines an internal category in $(\TwoCat,\graytensor,\grayunit)$.

%\subsection{Signed sets and Gray comonoids}
%A signed set is a set $\Label$ equipped with a function $\Label\to\{-1,≠1\}$.
%%
%The category $SignedSet$.
%%
%We observe that every $\twocatanchorof{\Label}$ where $\Label$
%is a signed set defines a comonoid.
%%
%A functor
%$$SignedSet\to\Comonoid{\TwoCat,\graytensor,\grayunit}$$

\subsection{Definition of internal category in a monoidal category~$\Scategorysurround$}\label{section/internal-category}
Suppose given a monoidal category $(\Scategorysurround,\graytensor,\grayunit)$
whose tensor product preserves coreflexive equalizers componentwise,
We start with the following definition:
% coming from the work by Aguiar~\cite{aguiar-phd-1997}.
%
\medbreak
\begin{definition}[internal category]
An internal category $\moo$ in the monoidal category $(\Scategorysurround,\graytensor,\grayunit)$
is defined as a monad in the weak double category $\Comod{\Scategorysurround}$.
\end{definition}
\medbreak
The definition may be expounded as follows: 
an internal category $\moo$ in $(\Scategorysurround,\graytensor,\grayunit)$ is a quadruple
$$
\moo \quad = \quad (\moo[0],\moo[1],\mathsf{mult},\mathsf{unit})
$$
consisting of a comonoid $\moo[0]$ with comultiplication and counit noted
$$
\begin{tikzcd}[column sep = 1.1em, row sep=1em]
d \,\, : \,\, {\moo[0]}\arrow[rr]
&&
{\moo[0]}\graytensor{\moo[0]}
\end{tikzcd}
\quad\quad
\begin{tikzcd}[column sep = 1.1em, row sep=1em]
e \,\, : \,\, {\moo[0]}\arrow[rr]
&&
{\grayunit}
\end{tikzcd}
$$
together with an $\moo[0],\moo[0]$-bicomodule 
\begin{equation}\label{equation/bicomodule-anchor}
\begin{tikzcd}[column sep = 1.5em, row sep=1em]
{\moo[1]} \quad : \quad {\moo[0]}\arrow[rr,spanmap]
&&
{\moo[0]}
\end{tikzcd}
\end{equation}
defined by a morphism in~$\Scategorysurround$
$$
\begin{tikzcd}[column sep = 1.2em]
\mathsf{coact}_{\moo} \,\, : \,\,
{\moo[1]}\arrow[rr] && 
{\moo[0]\graytensor\moo[1]\graytensor\moo[1]}
&& 
\end{tikzcd}
$$
making the diagrams~(\ref{equation/coactions}) commute,
together with double cells
$$
\begin{tikzcd}[column sep = 1.2em, row sep=.8em]
{\moo[0]}
\arrow[dd,"{id}"{swap}]
\arrow[rr,spanmap,"{{\moo[1]}}"{yshift=.4em}] 
&&
{\moo[0]}
\arrow[rr,spanmap,"{{\moo[1]}}"{yshift=.4em}] 
\arrow[yshift=-.1em,to=target, "{\mathsf{mult}}", double,-implies]
&&
{\moo[0]}
\arrow[dd,"{id}"]
\\
\\
{\moo[0]}
\arrow[rrrr,spanmap,"{{\moo[1]}}"{yshift=-.4em,swap},""{yshift=.4em,name=target}]
&&&&
{\moo[0]}
\end{tikzcd}
\quad\quad
\begin{tikzcd}[column sep = 1.4em, row sep=.8em]
{\moo[0]}
\arrow[dd,"{id}"{swap}]
\arrow[rr,spanmap,"{id}"{yshift=.4em},""{yshift=-.4em,swap,name=source}] 
&&
{\moo[0]}
\arrow[dd,"{id}"]
\\
\\
{\moo[0]}
\arrow[rr,spanmap,"{{\moo[1]}}"{yshift=-.4em,swap},""{yshift=.4em,name=target}]
&&
{\moo[0]}
\arrow[from=source, to=target, "{\mathsf{unit}}", double,-implies]
\end{tikzcd}
$$
defined by a pair of morphisms in $\Scategorysurround$
$$
\begin{array}{c}
\begin{tikzcd}[column sep = 1.2em, row sep=.8em]
\mathsf{mult} \,\, : \,\,
{\moo[1]\graytensoreq{\moo[0]}\moo[1]}\arrow[rr] 
&&{\moo[1]}
\end{tikzcd}
\\
\begin{tikzcd}[column sep = 1.2em, row sep=.8em]
\mathsf{unit} \,\, : \,\,
{\moo[0]}\arrow[rr] 
&&{\moo[1]}
\end{tikzcd}
\end{array}
$$
making the expected diagrams~(\ref{equation/double-cell-as-morphism}) commute.
One requires moreover that $\mathsf{mult}$ and $\mathsf{unit}$
make the associativity and neutrality diagrams~(\ref{equation/associativity-of-moo})
and~(\ref{equation/associativity-of-moo}) commute as required
of a monad $\moo$ in the weak double category $\Comod{\Scategorysurround}$.
%
%$$
%\begin{tikzcd}[column sep=3em,row sep=1em]
%{\moo[3]}\arrow[rr,"{\mathsf{mult}\graytensoreq{\moo[0]}{\moo[1]}}"]
%\arrow[dd,"{{\moo[1]}\graytensoreq{\moo[0]}\mathsf{mult}}"{swap}"]
%&&
%{\moo[2]}
%\arrow[dd,"{\mathsf{mult}}"]
%\\
%\\
%{\moo[2]}\arrow[rr,"{\mathsf{mult}}"]
%&&
%{\moo[1]}
%\end{tikzcd}
%$$
\begin{equation}\label{equation/associativity-of-moo}
\begin{tikzcd}[column sep=.5em,row sep=.5em]
&&
{\moo[3]}\arrow[rrdd,"{\mathsf{mult}\,\,\graytensoreq{\moo[0]}\,\,{\moo[1]}}"]
\arrow[lldd,"{{\moo[1]}\,\,\graytensoreq{\moo[0]}\,\,\mathsf{mult}}"{swap}"]
\\
\\
{\moo[2]}
\arrow[rrdd,"{\mathsf{mult}}"{swap}]
&&&&
{\moo[2]}\arrow[lldd,"{\mathsf{mult}}"]
\\
\\
&&
{\moo[1]}
\end{tikzcd}
\end{equation}
\begin{equation}\label{equation/neutrality-of-moo}
\begin{tikzcd}[column sep=.5em,row sep=.5em]
&&
{\moo[1]}
\arrow[rrdd,"{\mathsf{unit}\,\,\graytensoreq{\moo[0]}\,\,{\moo[1]}}"]
\arrow[lldd,"{{\moo[1]}\,\,\graytensoreq{\moo[0]}\,\,\mathsf{unit}}"{swap}"]
\arrow[dddd,"{id}"]
\\
\\
{\moo[2]}
\arrow[rrdd,"{\mathsf{mult}}"{swap}]
&&&&
{\moo[2]}\arrow[lldd,"{\mathsf{mult}}"]
\\
\\
&&
{\moo[1]}
\end{tikzcd}
\end{equation}
%
%$$
%\begin{tikzcd}[column sep=3em,row sep=1em]
%{\moo[1]}\arrow[rr,"{\mathsf{unit}\graytensoreq{\moo[0]}{\moo[1]}}"]
%\arrow[dd,"{{\moo[1]}\graytensoreq{\moo[0]}\mathsf{mult}}"{swap}"]
%&&
%{\moo[2]}
%\arrow[dd,"{\mathsf{mult}}"]
%\\
%\\
%{\moo[2]}\arrow[rr,"{\mathsf{mult}}"]
%&&
%{\moo[1]}
%\end{tikzcd}
%$$
%$$
%\begin{tikzcd}[column sep=3em,row sep=1em]
%{\moo[3]}\arrow[rr,"{\mathsf{mult}\graytensoreq{\moo[0]}{\moo[1]}}"]
%\arrow[dd,"{{\moo[1]}\graytensoreq{\moo[0]}\mathsf{mult}}"{swap}"]
%&&
%{\moo[2]}
%\arrow[dd,"{\mathsf{mult}}"]
%\\
%\\
%{\moo[2]}\arrow[rr,"{\mathsf{mult}}"]
%&&
%{\moo[1]}
%\end{tikzcd}
%$$
where we write
\begin{equation}\label{equation/horizontal-composites}
\begin{array}{lcc}
{\moo[2]} &  = & {\moo[1]\graytensoreq{\moo[0]}\moo[1]}
\\
\vspace{-.8em}
\\
{\moo[3]} & = & {\moo[1]\graytensoreq{\moo[0]}\moo[1]\graytensoreq{\moo[0]}\moo[1]}
\end{array}
\end{equation}
for the horizontal composites 
of the $\moo[0],\moo[0]$-bicomodule~(\ref{equation/bicomodule-anchor})
with itself.

\subsection{The template $\anchorofasynch$ of asynchronous games}
We establish now that the template $\anchorofasynch$ of asynchronous games
defined in the introduction as the Gray comonoid
$$\anchorofasynch[0] = \twocatanchorof{\polarityminus,\polarityplus}$$
with comultiplication and counit~(\ref{equation/Gray-comonoid-of-objects})
and the bicomodule
$$\anchorofasynch[1] = \twocatanchorof{\polarityminussource,\polarityplussource,\polarityminustarget,\polarityplustarget}$$
with coaction~(\ref{equation/bicomodule-of-anchor}).
The proof relies on the key observation that the horizontal composites
$\anchorofasynch[2]$ and $\anchorofasynch[3]$ defined in~(\ref{equation/horizontal-composites})
are themselves freely generated in the following way:
$$
\begin{array}{c}
\anchorofasynch[2] = \twocatanchorof{\polarityminussource,\polarityplussource,
\polarityminusmid,\polarityplusmid,\polarityminustarget,\polarityplustarget}
\\
\vspace{-.8em}
\\
\anchorofasynch[3] = \twocatanchorof{\polarityminussource,\polarityplussource,
\polarityminusmidone,\polarityplusmidone,\polarityminusmidtwo,\polarityplusmidtwo,
\polarityminustarget,\polarityplustarget}
\end{array}
$$
where the generators $\polarityminusmid$, $\polarityplusmid$,
$\polarityminusmidone$, $\polarityplusmidone$,
$\polarityminusmidtwo$, $\polarityplusmidtwo$
are the polarities of internal moves exchanged during the interaction
between strategies.
Note that we recover in this way a critical phenomenon 
and synchronization mechanism of traditional template games
observed by Melli{\`e}s~\cite{mellies-popl-2019}.

\medbreak

Once this key observation has been made, the proof becomes
a nice and enlightening combinatorial exercise, which
essentially amounts to checking very carefully
all the coherence diagrams required of an internal category
in \S\ref{section/internal-category} are indeed satisfied by~$\anchorofasynch$.
At this stage, we observe that:

\medbreak
\noindent
1. an asynchronous template game $(A,\lambda_A)$ 
as defined in Def.~\ref{definition/asynchronous-template-games}
is the same thing as an object of $\Comod{\TwoCat,\graytensor,\grayunit}$
equipped with a vertical map $\lambda_A:A\longrightarrow \anchorofasynch[0]$.

\medbreak
\noindent
2. an asynchronous strategy $\sigma$
as defined in Def.~\ref{definition/asynchronous-strategies}
is the same thing as a double cell of the form
$$
\begin{tikzcd}[column sep = 3em, row sep=1em]
{A}
\arrow[dd,"{\lambda_A}"{swap}]
\arrow[rr,spanmap,"{S}"{yshift=.4em},""{yshift=-.4em,swap,name=source}] 
&&
{B}
\arrow[dd,"{\lambda_B}"]
\\
\\
{\anchorofasynch[0]}
\arrow[rr,spanmap,"{{\anchorofasynch[1]}}"{yshift=-.4em,swap},""{yshift=.4em,name=target}]
&&
{\anchorofasynch[0]}
\arrow[from=source, to=target, "{\lambda_{\sigma}}", double,-implies]
\end{tikzcd}
$$

\medbreak
\noindent
3. a notion of simulation
$$
\begin{tikzcd}[column sep=.8em, row sep=1.2em]
\theta \,\, : \,\, \sigma \arrow[-implies,double,rr] && \tau\, \,\, : \,\, (A,\lambda_A) \arrow[spanmap]{rrrr} &&&& (B,\lambda_B)
\end{tikzcd}
$$
between asynchronous strategies $\sigma=(S,\lambda_{\sigma})$
and $\tau=(T,\lambda_{\tau})$ may be defined as a double cell
in $\Comod{\TwoCat,\graytensor,\grayunit}$
$$
\begin{tikzcd}[column sep = 5em, row sep=.8em]
{A}
\arrow[dd,"{id}"{swap}]
\arrow[rr,spanmap,"{S}"{yshift=.4em},""{yshift=-.4em,swap,name=source}] 
&&
{B}
\arrow[dd,"{id}"]
\\
\\
{A}
\arrow[rr,spanmap,"{T}"{yshift=-.4em,swap},""{yshift=.4em,name=target}]
&&
{B}
\arrow[from=source, to=target, "{\theta}", double,-implies]
\end{tikzcd}
$$
satisfying the expected property that $\lambda_{\sigma}$
concides with $\lambda_{\tau}$ vertically precomposed
with the double cell~$\theta$.
\medbreak

From these basic observations together with the property that $\anchorofasynch$
defines an internal category in $(\TwoCat,\graytensor,\grayunit)$,
one easily deduces the main result of our paper, either directly
or using the nice and general recipe in~\cite{eberhard-hirschowitz-laouar-fscd-2019}.
\medbreak
\begin{theorem}
$\Games{(\anchorofasynch)}$ defines a bicategory
of asynchronous templates games, strategies and simulations.
\end{theorem}
\medbreak

%The bicategory 
%Every internal category $\moo$ in a monoidal category 
%$(\Scategorysurround,\graytensor,\grayunit)$
%whose tensor product preserves equalizers componentwise
%defines a bicategory $\Games{\anchor}$ whose objects are games,
%whose morphisms are strategies, and whose 2-cells are simulations.
%\end{theorem}

\section{The star-autonomous bicategory\\ of asynchronous template games}\label{section/star-autonomous}
%\cite{eberhard-hirschowitz-laouar-fscd-2019}
%Besides the precious information we obtain on the structure
%of asynchronous games,
%In the case of traditional template games, 
One main benefit of the structural approach to game semantics
based on template games~\cite{mellies-popl-2019,mellies-lics-2019}
is that it becomes easier to establish in full precision and rigour
that a given bicategory of games, strategies and simulations
is symmetric monoidal closed, or even $\ast$-autonomous.
This important aspect of template game semantics
remains true in the asynchronous and 2-categorical framework
developed in the present paper.
Indeed, we establish that the bicategory~$\Games{(\anchorofasynch)}$
of asynchonous template games is $\ast$-autonomous by observing

\medbreak
\noindent
1. that the Gray tensor product $\graytensor$ is symmetric and induces
for that reason a tensor product on the double category $\Comod{\TwoCat,\graytensor,\grayunit}$
which coincides with the tensor product of Gray comonoids on the vertical category.

\medbreak
\noindent
2. that the resulting Gray tensor product $\graytensor$ 
induces in turn a tensor product $\graytensor$ 
on the category of internal categories and internal functors
in~$(\TwoCat,\graytensor,\grayunit)$ where the notion
of internal functor is defined as a monad morphism 
in $\Comod{\TwoCat,\graytensor,\grayunit}$.
%as defined by Aguiar in~\cite{aguiar-phd-1997}.

\medbreak
\noindent
3. that the template $\anchorofasynch$ comes with an internal functor 
$$\tensor:\anchorofasynch\graytensor\anchorofasynch\longrightarrow\anchorofasynch$$
defined as the double cell in $\Comod{\TwoCat,\graytensor,\grayunit}$
$$
\begin{small}
\begin{tikzcd}[column sep = 3em, row sep=1em]
{\anchorofasynch[0]\graytensor\anchorofasynch[0]}
\arrow[dd,"{\tensor[0]}"{swap}]
\arrow[rr,spanmap,"{{\anchorofasynch[1]}\graytensor{\anchorofasynch[1]}}"{yshift=.4em},""{yshift=-.4em,swap,name=source}] 
&&
{\anchorofasynch[0]\graytensor\anchorofasynch[0]}
\arrow[dd,"{\tensor[0]}"]
\\
\\
{\anchorofasynch[0]}
\arrow[rr,spanmap,"{{\anchorofasynch[1]}}"{yshift=-.4em,swap},""{yshift=.4em,name=target}]
&&
{\anchorofasynch[0]}
\arrow[from=source, to=target, "{\tensor[1]}", double,-implies]
\end{tikzcd}
\end{small}
$$
defined by the canonical Gray monoid structures $\tensor[0]$
and $\tensor[1]$ on $\anchorofasynch[0]$ and $\anchorofasynch[1]$ 
obtained by transporting the monoid structure $S\mapsto S+S$
of any set of generators along the symmetric monoidal functor $\twocatanchorof{-}$
described in~(\ref{equation/anchorof-two-categories}).
%$$
%\begin{tikzcd}[column sep = 1em]
%\tensor \quad : \quad 
%\anchorofasynch\graytensor\anchorofasynch\arrow[rr] && \anchorofasynch
%\end{tikzcd}
%$$
%defining the tensor product of $\Games{(\anchorofasynch)}$.

\medbreak
Similarly, we observe that the template $\anchorofasynch$ is equipped with an internal functor
$$
\mathsf{neg} \,\, : \,\, \anchorofasynch^{op}\longrightarrow\anchorofasynch
$$
defined as the double cell in $\Comod{\TwoCat,\graytensor,\grayunit}$
$$
\begin{small}
\begin{tikzcd}[column sep = 4em, row sep=1em]
{\anchorofasynch[0]^{op}}
\arrow[dd,"{\mathsf{neg}[0]}"{swap}]
\arrow[rr,spanmap,"{{\anchorofasynch[1]^{op}}}"{yshift=.4em},""{yshift=-.4em,swap,name=source}] 
&&
{\anchorofasynch[0]^{op}}
\arrow[dd,"{\mathsf{neg}[0]}"]
\\
\\
{\anchorofasynch[0]}
\arrow[rr,spanmap,"{{\anchorofasynch[1]}}"{yshift=-.4em,swap},""{yshift=.4em,name=target}]
&&
{\anchorofasynch[0]}
\arrow[from=source, to=target, "{\mathsf{neg}[1]}", double,-implies]
\end{tikzcd}
\end{small}
$$
where the opposite $\anchor^{op}$ of an internal category 
is defined by permuting the outputs of the underlying comonoids
and bicomodules using the symmetry $\mathsf{sym}$ of the Gray tensor product.
Note that this purely algebraic operation in $(\TwoCat,\graytensor,\grayunit)$
reverses the left-to-right and right-to-left orientations of the template $\anchorofasynch$, 
and thus the role of each polarity, as seen for instance
in the comultiplication~(\ref{equation/Gray-comonoid-of-objects}) 
discussed in the introduction:
$$
\begin{small}
\begin{tikzcd}[column sep = .8em]
\twocatanchorof{\polarityminus,\polarityplus}
\arrow[rr,"d"]
&&
\twocatanchorof{\polarityminus,\polarityplus}\graytensor
\twocatanchorof{\polarityminus,\polarityplus}
\arrow[rr,"{\mathsf{sym}}"]
&&
\twocatanchorof{\polarityminus,\polarityplus}\graytensor
\twocatanchorof{\polarityminus,\polarityplus}
\end{tikzcd}
\end{small}
$$
We obtain in this way the main result of the paper:
\medbreak
\begin{theorem}
The bicategory $\Games{(\anchorofasynch)}$ is symmetric monoidal closed,
and in fact $\ast$-autonomous.
\end{theorem}

\section{Conclusion and future works}\label{section/conclusion}
In this paper, we have shown how to upgrade to a properly asynchronous 
and 2-categorical framework the template game semantics
designed by Melli{\`e}s~\cite{mellies-popl-2019} for concurrent games.
As we explain, the resulting concurrent game model based on the ``asynchronous''
Gray tensor product of 2-categories
resolves a defect in the interpretation of deadlocks
in the original semantics based on functorial spans and pullbacks of categories.
One main challenge for future work will be to extend our current interpretation
of multiplicative linear logic (MALL) to an asynchronous game semantics of differential linear logic,
integrating to our 2-categorical setting the homotopy approach of \cite{mellies-lics-2019}.
%
%The exponential modality $A\mapsto {!}A$ would be interpreted
%as the construction of the free sylleptic monoidal 2-category $A\mapsto \textbf{Syll}(A)$.
%
%\cite{Shulman2008}
%bisimulation
%\cite{mellies-stefanesco-lics-2018}
%
We also believe that the formalism of asynchronous template games 
is sufficiently simple, general and conceptually clean to provide
a unifying framework for various forms of sequential and concurrent
game semantics.
% in harmony with categorical algebra.
% and bring harmony to the field.
%work on higher dimensional categories.
%
%and to shed light to the extensions of innocent and positional strategies
%to nondeterministic and concurrent scenarios~\cite{harmer-hyland-mellies-lics-2007,mellies-jacq-fossacs-2018,TsukadaOng15,EberhartHirschowitz18}.
%and to shed light on the different structuralist approaches
%to game semantics

%to nondeterministic and concurrent scenarios.

%shed light on the fine-grained structure of innocence,
%both in a deterministic and nondeterministic
%scenario, and shed light on the existing formulations
%~\cite{harmer-hyland-mellies-lics-2007,EberhartHirschowitz18,mellies-jacq-fossacs-2018}
%sequential and concurrent framework,
%and offer a tentative synthesis between the existing structural approaches to innocence
%based on the existence of a distributivity law ${!}\circ{?}\to{?}\circ{!}$ between exponential modalities
%as initiated by Harmer, Hyland and Mellies

%or alternatively, a category of permutations on plays with pointers,
%as studied by Tsukada and Ong~\cite{TsukadaOng15}
%and then Eberhart and Hirschowitz~\cite{EberhartHirschowitz18}.
%

\section*{Acknowledgment}
The author would like to thank Pierre Clairambault for his early remarks
on the treatment of deadlocks in the concurrent template game semantics,
as well as the LICS reviewers for their helpful comments on the final version of the paper.

%\bibliography{../../template-games.bib}

\appendices

\section{Proof of the existence of finite limits\\
in the category $\Asynch$
of asynchronous graphs (\S\ref{section/finite-limits-of-Asynch}, 
Prop.~\ref{proposition/finite-limits-in-Asynch})}\label{appendix/finite-limits}
We have established in \S\ref{section/finite-limits-of-Asynch}
that the category $\Asynch$ has finite products.
%\medbreak
%
A second and less immediate property
%less immediate and made possible
%by the careful way we carved our definition of asynchronous games,
is that the category $\Asynch$ has all equalizers,
and thus all finite limits.
Indeed, given a pair of homomorphisms between asynchronous graphs
$$
\begin{tikzcd}[column sep = 1.4em]
(G,\diamond_G)
\arrow[rr,"{f}",yshift=.2em]
\arrow[rr,"{g}"{swap},yshift=-.2em]
&&
(H,\diamond_H)
\end{tikzcd}
$$
the equalizer $(E,\diamond_E)$ may be defined as follows:
its vertices are the vertices $x$ of the graph $G$ such that $f(x)=g(x)$
and its edges are the edges $u:x\to y$ such that $f(u)=g(u)$.
Note that $E$ may be seen as a subgraph of $G$.
Then, a pair of paths $p,q:x\transitionpath y$ of length 2 in the graph $E$ 
defines a permutation tile $p\diamond_E q$ precisely when 
the two paths $p,q$ define such a permutation tile $p\diamond_G q$
in the asynchronous graph $(G,\diamond_G)$.
We establish that defined in this way,
$(E,\diamond_E)$ satisfies the properties
required of an asynchronous graph, 
%As we explain in \S\ref{section/finite-limits-of-Asynch},
%given a pair of homomorphisms between asynchronous graphs
%$$
%\begin{tikzcd}[column sep = 1.4em]
%(G,\diamond_G)
%\arrow[rr,"{f}",yshift=.2em]
%\arrow[rr,"{g}"{swap},yshift=-.2em]
%&&
%(H,\diamond_H)
%\end{tikzcd}
%$$
%the equalizer $(E,\diamond_E)$ is defined as follows:
%its vertices are the vertices $M$ of the graph $G$ such that $f(M)=g(M)$
%and its edges are the edges $u:M\to N$ such that $f(u)=g(u)$.
%
$E$ may be seen as a subgraph of $G$.
Then, by definition, a pair of paths $p,q:M\transitionpath N$ of length 2 in the graph $E$ 
defines a permutation tile $p\diamond_E q$ precisely when 
the two paths $p,q$ define such a permutation tile $p\diamond_G q$
in the asynchronous graph $(G,\diamond_G)$.
At this stage, the main difficulty is to show that the pair $(E,\diamond_E)$ 
just constructed satisfies the three properties required of an asynchronous graph.
Symmetry and determinism of permutations are immediately deduced from
the fact the properties are satisfied by $(G,\diamond_G)$,
while the cube property of $(E,\diamond_E)$ follows easily
from the following observation:
% combined with the determinism
%of permutation in the asynchronous graph $(G,\diamond_G)$:
%
\medbreak
\begin{proposition}
Suppose given a pair of paths $p,q:M\transitionpath N$ of length 2 
defining a permutation tile $p\diamond_G q$ in the asynchronous graph $(G,\diamond_G)$.
Suppose moreover that the path $p=u_1\cdot u_2:M\transitionpath N$ 
is a path of the subgraph $E$ of the graph $G$.
In that case, the path $q=v_1\cdot v_2:M\transitionpath N$
is also a path of the subgraph $E$ and $p\diamond_E q$.
\end{proposition}
\begin{proof}
In order to prove the property, suppose that $f(u_1)=g(u_1)$ and $f(u_2)=g(u_2)$.
By the property of a homomorphism, we know that $f(u_1)\cdot f(u_2)\diamond_H f(v_1)\cdot f(v_2)$
and that $g(u_1)\cdot g(u_2)\diamond_H g(v_1)\cdot g(v_2)$.
Since $f(u_1)\cdot f(u_2)=g(u_1)\cdot g(u_2)$,
it follows by determinism of permutation in $(H,\diamond_H)$
that $f(v_1)\cdot f(v_2)=g(v_1)\cdot g(v_2)$ and thus that $f(v_1)=g(v_1)$ and $f(v_2)=g(v_2)$.
From this, we conclude that $v_1$ and $v_2$ are edges in the subgraph $E$ of the graph $G$
and we conclude that the path $q=v_1\cdot v_2$ is a path of $E$
which defines a permutation tile $p\diamond_E q$.
\end{proof}

%
%
% Appendix
%
%

\section{An algebraic presentation\\
of the Gray tensor product}\label{section/gray-tensor}
%
%The main purpose of the section is to give a formal definition
%of the Gray tensor product $\Acategory\graytensor\Bcategory$
%of two small 2-categories $\Acategory$, $\Bcategory$.
%We give here the definition by generators and relations
%of the 2-category $\Acategory\graytensor\Bcategory$,
%as found for instance in \cite{gurski-book-2013}.
%We recall in \S\ref{section/gray-tensor/alternative-definition} the alternative definition 
%by factorisation developed in \cite{bourke-gurski-acs-2017}.
%\subsubsection{Definition of the Gray tensor product}\label{section/gray-tensor/definition}
We suppose given a pair of small 2-categories $\Acategory$, $\Bcategory$
and define their Gray tensor product as the 2-category $\Acategory\graytensor\Bcategory$ described by generators and relations, in the following way.
The construction is somewhat part of the folklore and may be found 
for instance in~\cite{gurski-book-2013}.

\medbreak
\noindent
\textbf{Definition of the 0-cells.\,\,}
The objects of the 2-category $\Acategory\graytensor\Bcategory$ 
are the pairs $(A,B)$ consisting of an object $A$ of $\Acategory$
and an object $B$ of $\Bcategory$.
\medbreak
\noindent
\textbf{Definition of the 1-cells by generators and relations.\,\,}
Its 1-cells
% of the 2-category $\Acategory\graytensor\Bcategory$
are generated by two families of 1-cells
\begin{equation}\label{equation/gray-cells}
\begin{small}
\begin{array}{cc}
\begin{tikzcd}[column sep = 1.5em, row sep=2em]
(A,B) 
\arrow[rr,"{(\fsubA,B)}"]
&&
(A',B)
\end{tikzcd}
&
\begin{tikzcd}[column sep = 1.5em, row sep=2em]
(A,B) 
\arrow[rr,"{(A,\fsubB)}"]
&&
(A,B')
\end{tikzcd}
\end{array}
\end{small}
\end{equation}
indexed by the 1-cells $\fsubA:A\to A'$ of the 2-category $\Acategory$
and by the 1-cells $\fsubB:B\to B'$ of the 2-category $\Bcategory$, respectively.
The 1-cells of the freely generated category are then quotiented by the relations:
\begin{equation}
\fbox{
$\begin{array}{ccc}
&
(\fsubA,B)\horcomp(\fsubAprime,B)=(\fsubA\horcomp\fsubAprime,B)
&
\\
&
(\horid{A},B)=\horid{(A,B)}
&
\\
&
(A,\fsubB)\horcomp(A,\fsubBprime)=(A,\fsubB\horcomp\fsubBprime)
&
\\
&
(A,\horid{B})=\horid{(A,B)}
&
\end{array}$}
\end{equation}
which can be also described as the four equations:
%\begin{equation}
$$
\begin{array}{c}
\begin{tikzcd}[column sep = 2em]
(A,B)\arrow[rr,"{(\fsubA,B)}"]
&&
(A',B)\arrow[rr,"{(\fsubAprime,B)}"]
&&
(A'',B)
\end{tikzcd}
\\
=
\\
\begin{tikzcd}[column sep = 3em]
(A,B)\arrow[rr,"{(\fsubAprime\horcomp\fsubA,B)}"]
&&
(A'',B)
\end{tikzcd}
\end{array}
$$
%\end{equation}
%\begin{equation}
$$
\begin{array}{ccc}
& \begin{tikzcd}[column sep = 2.5em]
(A,B)\arrow[rr,"{(\horid{A},B)}"]
&&
(A,B)
\end{tikzcd} &
\\
& = &
\\
&
\begin{tikzcd}[column sep = 2.5em]
(A,B)\arrow[rr,"{\horid{(A,B)}}"]
&&
(A,B)
\end{tikzcd}
&
\end{array}
%\end{equation}
$$
$$
%\begin{equation}
\begin{array}{c}
\begin{tikzcd}[column sep = 1.5em]
(A,B)\arrow[rr,"{(A,\fsubB)}"]
&&
(A,B')\arrow[rr,"{(A,\fsubBprime)}"]
&&
(A,B'')
\end{tikzcd}
\\
\quad = \quad
\\
\begin{tikzcd}[column sep = 2.5em]
(A,B)\arrow[rr,"{(A,\fsubBprime\horcomp\fsubB)}"]
&&
(A,B'')
\end{tikzcd}
\end{array}
$$
%\end{equation}
%
%\begin{equation}
$$
\begin{array}{c}
\begin{tikzcd}[column sep = 2.5em]
(A,B)\arrow[rr,"{(A,\horid{B})}"]
&&
(A,B)
\end{tikzcd}
\\
\quad = \quad
\\
\begin{tikzcd}[column sep = 2.5em]
(A,B)\arrow[rr,"{\horid{(A,B)}}"]
&&
(A,B)
\end{tikzcd}
\end{array}
$$
%\end{equation}
\medbreak
\noindent
\textbf{Definition of the 2-cells by generators and relations.\,\,}
Its 2-cells are generated by four families of generators: the two families of 2-cells
\begin{equation}
\begin{tikzcd}[column sep = 2.8em, row sep=2em]
(A,B) 
\arrow[rr,"{(\fsubA,B)}", bend left, ""{yshift=-.2em,swap,name=source}] 
\arrow[rr,"{(\fsubAprime,B)}"{swap}, bend right, ""{yshift=.2em,name=target}] 
&&
(A',B)
\arrow[from=source, to=target, "{(\alpha,B)}", double,-implies]
\end{tikzcd}
\end{equation}
\begin{equation}
\begin{tikzcd}[column sep = 2.8em, row sep=2em]
(A,B) 
\arrow[rr,"{(A,\fsubB)}", bend left, ""{yshift=-.2em,swap,name=source}] 
\arrow[rr,"{(A,\fsubBprime)}"{swap}, bend right, ""{yshift=.2em,name=target}] 
&&
(A,B')
\arrow[from=source, to=target, "{(A,\beta)}", double,-implies]
\end{tikzcd}
\end{equation}
indexed by the 2-cells $\alpha:\fsubA\Rightarrow\fsubAprime:A\to A'$ of the 2-category $\Acategory$
and by the 2-cells $\beta:\fsubB\Rightarrow\fsubBprime:B\to B'$ of the 2-category $\Bcategory$, respectively ;
and the two families of 2-cells called Gray commutation:
%\begin{equation}
%\begin{small}
%\begin{tikzcd}[column sep = 2.8em, row sep=2em]
%(A,B) 
%\arrow[dd,"{(A,\fsubB)}"{swap}, bend right, ""{xshift=-.2em,name=source}] 
%\arrow[dd,"{(A,\fsubBprime)}", bend left, ""{xshift=.2em,swap,name=target}] 
%\\
%\\
%(A,B')
%\arrow[from=source, to=target, "{(A,\beta)}", double,-implies]
%\end{tikzcd}
%\end{small}
%\end{equation}
\begin{equation}
\begin{tikzcd}[column sep = 2.8em, row sep=2em]
(A,B) 
\arrow[dd,"{(A,\fsubB)}"{swap}]
\arrow[rr,"{(\fsubA,B)}", ""{yshift=-.4em,swap,name=source}] 
&&
(A',B)
\arrow[dd,"{(A',\fsubB)}"]
\\
\\
(A,B')
\arrow[rr,"{(\fsubA,B')}"{swap},""{yshift=.4em,name=target}]
&&
(A',B')
\arrow[from=source, to=target, "{\graytile{\fsubA}{\fsubB}}", double,-implies]
\end{tikzcd}
\end{equation}
\begin{equation}
\begin{tikzcd}[column sep = 2.8em, row sep=2em]
(A,B) 
\arrow[dd,"{(A,\fsubB)}"{swap}, ""{xshift=.2em,name=source}]
\arrow[rr,"{(\fsubA,B)}"] 
&&
(A',B)
\arrow[dd,"{(A',\fsubB)}",""{xshift=-.2em,swap,name=target}]
\\
\\
(A,B')
\arrow[rr,"{(\fsubA,B')}"{swap}]
&&
(A',B')
\arrow[from=source, to=target, "{\graytiletilde{\fsubA}{\fsubB}}", double,-implies]
\end{tikzcd}
\end{equation}
indexed by the 1-cells $\fsubA:A\to A'$ of the 2-category $\Acategory$
and the 1-cells $\fsubB:B\to B'$ of the 2-category $\Bcategory$.
The 2-cells of the 2-category freely generated by these generators 
are then quotiented by four families of relations.
The first family of relations [a] regulates how the pairing interacts with vertical composition,
while the three remaining families of relations [b-c-d] supervise the Gray commutation.

%%%%%%%%%%%%%%%%%%%%%%%%%%%%%%%%%%%%%%%%
%
%
%   Vertical composition and identity on the left side
%
%
%%%%%%%%%%%%%%%%%%%%%%%%%%%%%%%%%%%%%%%%
\medbreak
\noindent
\textbf{[a] Functoriality with respect to vertical composition:}
the first group of four relations enforces that the operations $\fsubB\mapsto(A,\fsubB)$ and $\fsubA\mapsto(\fsubA,B)$
are functorial with respect to vertical composition, for every object $A$ of the 2-category $\Acategory$
and every object $B$ of the 2-category $\Bcategory$:
\begin{equation}
\fbox{
$\begin{array}{ccc}
&
(\alpha,B)\vertcomp(\alpha',B) = (\alpha\vertcomp\alpha',B)
&
\\
&
(\vertid{\fsubA},B) = \vertid{(\fsubA,B)}
&
\\
&
(A,\beta)\vertcomp(A,\beta') = (\beta\vertcomp\beta',B)
&
\\
&
(A,\vertid{\fsubB}) = \vertid{(A,\fsubB)}
&
\end{array}$
}
\end{equation}
The four equations between 2-cells are nicely depicted using pasting diagrams:
$$
%\begin{equation}
\begin{footnotesize}
\begin{array}{ccc}
\begin{tikzcd}[column sep = 3em, row sep=2.4em]
(A,B) 
\arrow[rr,"{(\fsubA,B)}", bend left = 60 , ""{yshift=-.2em,swap,name=uppersource}] 
\arrow[rr,"{(\fsubAprime,B)}"{description}, ""{yshift=.4em,name=uppertarget}, ""{swap,yshift=-.4em,name=lowersource}]
\arrow[rr,"{(\fsubAsecond,B)}"{swap}, bend right = 60, ""{yshift=.2em,name=lowertarget}] 
&&
(A',B)
\arrow[from=uppersource, to=uppertarget, "{(\alpha,B)}", double,-implies]
\arrow[from=lowersource, to=lowertarget, "{(\alpha',B)}", double,-implies]
\end{tikzcd}
&
=
&
\begin{tikzcd}[column sep = 3em, row sep=2.4em]
(A,B) 
\arrow[rr,"{(\fsubA,B)}", bend left = 60 , ""{yshift=-.2em,swap,name=source}] 
\arrow[rr,"{(\fsubAsecond,B)}"{swap}, bend right = 60, ""{yshift=.2em,name=target}] 
&&
(A',B)
\arrow[from=source, to=target, "{(\alpha\vertcomp\alpha',B)}"{description}, double,-implies]
\end{tikzcd}
\end{array}
\end{footnotesize}
%\end{equation}
$$
$$
%\begin{equation}
\begin{footnotesize}
\begin{array}{ccc}
\begin{tikzcd}[column sep = 3em, row sep=2.4em]
(A,B) 
\arrow[rr,"{(\fsubA,B)}", bend left = 60 , ""{yshift=-.2em,swap,name=source}] 
\arrow[rr,"{(\fsubA,B)}"{swap}, bend right = 60, ""{yshift=.2em,name=target}] 
&&
(A',B)
\arrow[from=source, to=target, "{(\vertid{\fsubA},B)}"{description}, double,-implies]
\end{tikzcd}
&
=
&
\begin{tikzcd}[column sep = 3em, row sep=2.4em]
(A,B) 
\arrow[rr,"{(\fsubA,B)}", bend left = 60 , ""{yshift=-.2em,swap,name=source}] 
\arrow[rr,"{(\fsubA,B)}"{swap}, bend right = 60, ""{yshift=.2em,name=target}] 
&&
(A',B)
\arrow[from=source, to=target, "{\vertid{(\fsubA,B)}}", double,-implies]
\end{tikzcd}\end{array}
\end{footnotesize}
%\end{equation}
$$
%%%%%%%%%%%%%%%%%%%%%%%%%%%%%%%%%%%%%%%%
%
%
%   Vertical composition and identity on the right side
%
%
%%%%%%%%%%%%%%%%%%%%%%%%%%%%%%%%%%%%%%%%
%\noindent
%\textbf{Vertical composition and identity on the right side:}
%$$
%(A,\beta')\vertcomp(A,\beta) = (\beta'\vertcomp\beta,B)
%\quad\quad
%(A,\vertid{\fsubB}) = \vertid{(A,\fsubB)}
%$$
$$
%\begin{equation}
\begin{footnotesize}
\begin{array}{ccc}
\begin{tikzcd}[column sep = 3em, row sep=2.4em]
(A,B) 
\arrow[rr,"{(A,\fsubB)}", bend left = 60 , ""{yshift=-.2em,swap,name=uppersource}] 
\arrow[rr,"{(A,\fsubBprime)}"{description}, ""{yshift=.4em,name=uppertarget}, ""{swap,yshift=-.4em,name=lowersource}]
\arrow[rr,"{(A,\fsubBsecond)}"{swap}, bend right = 60, ""{yshift=.2em,name=lowertarget}] 
&&
(A,B')
\arrow[from=uppersource, to=uppertarget, "{(A,\beta)}", double,-implies]
\arrow[from=lowersource, to=lowertarget, "{(A,\beta')}", double,-implies]
\end{tikzcd}
&
=
&
\begin{tikzcd}[column sep = 3em, row sep=2.4em]
(A,B) 
\arrow[rr,"{(A,\fsubB)}", bend left = 60 , ""{yshift=-.2em,swap,name=source}] 
\arrow[rr,"{(A,\fsubBsecond)}"{swap}, bend right = 60, ""{yshift=.2em,name=target}] 
&&
(A,B')
\arrow[from=source, to=target, "{(A,\beta\vertcomp\beta')}"{description}, double,-implies]
\end{tikzcd}
\end{array}
\end{footnotesize}
%\end{equation}
$$
$$
%\begin{equation}
\begin{footnotesize}
\begin{array}{ccc}
\begin{tikzcd}[column sep = 3em, row sep=2.4em]
(A,B) 
\arrow[rr,"{(\fsubA,B)}", bend left = 60 , ""{yshift=-.2em,swap,name=source}] 
\arrow[rr,"{(\fsubA,B)}"{swap}, bend right = 60, ""{yshift=.2em,name=target}] 
&&
(A',B)
\arrow[from=source, to=target, "{(\vertid{\fsubA},B)}"{description}, double,-implies]
\end{tikzcd}
&
=
&
\begin{tikzcd}[column sep = 3em, row sep=2.4em]
(A,B) 
\arrow[rr,"{(\fsubA,B)}", bend left = 60 , ""{yshift=-.2em,swap,name=source}] 
\arrow[rr,"{(\fsubA,B)}"{swap}, bend right = 60, ""{yshift=.2em,name=target}] 
&&
(A',B)
\arrow[from=source, to=target, "{\vertid{(\fsubA,B)}}", double,-implies]
\end{tikzcd}\end{array}
\end{footnotesize}
%\end{equation}
$$
%%%%%%%%%%%%%%%%%%%%%%%%%%%%%%%%%%%%%%%%
%
%
%   Invertibility of the graytile
%
%
%%%%%%%%%%%%%%%%%%%%%%%%%%%%%%%%%%%%%%%%
\medbreak
\noindent
\textbf{[b] Invertibility of the Gray commutation:}
the combination of two relations below indicates that the 2-cell $\graytile{\fsubA}{\fsubB}$ is invertible
with the 2-cell $\graytiletilde{\fsubA}{\fsubB}$ as vertical inverse:
\begin{equation}
\fbox{
$\begin{array}{ccc}
&
\graytile{\fsubA}{\fsubB}
\vertcomp
\graytiletilde{\fsubA}{\fsubB} = \vertid{(\fsubA,B)\horcomp(A',\fsubB)}
&
\\
&
\graytiletilde{\fsubA}{\fsubB}
\vertcomp
\graytile{\fsubA}{\fsubB} = \vertid{(A,\fsubB)\horcomp(\fsubA,B')}
&
\\
\end{array}$
}
\end{equation}
%%%%%%%%%%%%%%%%%%%%%%%%%%%%%%%%%%%%%%%%
%
%
%   Gray permutation tile and 2-cell on the left and/or right sides
%
%
%%%%%%%%%%%%%%%%%%%%%%%%%%%%%%%%%%%%%%%%
\medbreak
\noindent
\textbf{[c] Naturality of the Gray commutation:}
% on the left and on the right:}
the group of four relations below enforces that the Gray commutation 
is natural with respect to vertical composition,
both on the left and on the right side of the Gray tensor product:
\begin{equation}
\fbox{
$\begin{array}{ccc}
&
(\alpha,B)\vertcomp\graytile{\fsubAprime}{\fsubB}
=
\graytile{\fsubA}{\fsubB}\vertcomp(\alpha,B')
&
\\
&
(A',\beta)\vertcomp\graytile{\fsubA}{\fsubBprime}
=
\graytile{\fsubA}{\fsubB}\vertcomp(A,\beta)
&
\end{array}$
}
\end{equation}
The two equations between 2-cells are nicely expressed
in the diagrammatic language of pasting diagrams.
The first equation says that the pasting diagram below
%\begin{equation}
$$
\begin{tikzcd}[column sep = 2.5em, row sep=2em]
(A,B) 
\arrow[rr,"{(\fsubA,B)}", bend left = 60 , ""{yshift=-.2em,swap,name=uppersource}] 
\arrow[rr,"{(\fsubAprime,B)}"{description}, bend right = 60, ""{yshift=.4em,name=uppertarget}, ""{swap,yshift=-.4em,name=lowersource}] 
\arrow[dd,"{(A,\fsubB)}"{swap}] 
&&
(A',B)
\arrow[dd,"{(A',\fsubB)}"] 
\\
\\
(A,B') 
\arrow[rr,"{(\fsubAprime,B')}"{swap}, bend right = 60, ""{yshift=0em,name=lowertarget}] 
&&
(A',B')
\arrow[from=uppersource, to=uppertarget, "{(\alpha,B)}"{description}, double,-implies]
\arrow[from=lowersource, to=lowertarget, "{\graytile{\fsubAprime}{\fsubB}}"{description}, double,-implies]
\end{tikzcd}
$$
should be identified with the pasting diagram
$$
\begin{tikzcd}[column sep = 2.5em, row sep=2em]
(A,B) 
\arrow[rr,"{(\fsubA,B)}", bend left = 60 , ""{yshift=-.2em,swap,name=uppersource}] 
\arrow[dd,"{(A,\fsubB)}"{swap}] 
&&
(A',B)
\arrow[dd,"{(A',\fsubB)}"] 
\\
\\
(A,B') 
\arrow[rr,"{(\fsubA,B')}"{description}, bend left = 60 , ""{yshift=-.4em,swap,name=lowersource}, ""{yshift=.4em,name=uppertarget}] 
\arrow[rr,"{(\fsubAprime,B')}"{swap}, bend right = 60, ""{yshift=0em,name=lowertarget}] 
&&
(A',B')
\arrow[from=uppersource, to=uppertarget, "{\graytile{\fsubA}{\fsubB}}"{description}, double,-implies]
\arrow[from=lowersource, to=lowertarget, "{(\alpha,B')}"{description}, double,-implies]
\end{tikzcd}
%\end{array}
%\end{equation}
$$
The second equation says that the pasting diagram below
%\begin{equation}
%\begin{small}
$$
\begin{tikzcd}[column sep = 1.5em, row sep=3em]
(A,B) 
\arrow[dd,"{(A,\fsubBprime)}"{swap},bend right = 60,""{name=lefttarget}]
\arrow[rr,"{(\fsubA,B)}"] 
&&
(A',B) 
\arrow[dd,"{(A',\fsubB)}", bend left = 60, ""{xshift=0em,swap,name=rightsource}] 
\arrow[dd,"{(A',\fsubBprime)}"{description}, bend right = 60, ""{xshift=-1.2em,swap,name=leftsource}, ""{xshift=1em,name=righttarget}] 
\\
\\
(A,B')
\arrow[rr,"{(\fsubA,B')}"{swap},""{yshift=.2em,name=target}]
&&
(A',B')
\arrow[from=leftsource, to=lefttarget, "{\graytile{\fsubA}{\fsubBprime}}"{swap}, double,-implies]
\arrow[from=rightsource, to=righttarget, "{(A',\beta)}"{swap}, double,-implies]
\end{tikzcd}
$$
should be identified with the pasting diagram
$$
\begin{tikzcd}[column sep = 1.5em, row sep=3em]
(A,B) 
\arrow[dd,"{(A,\fsubBprime)}"{swap},bend right = 60,""{name=lefttarget}]
\arrow[dd,"{(A,\fsubB)}"{description}, bend left = 60, ""{xshift=-.9em,swap,name=leftsource}, ""{xshift=.9em,name=righttarget}] 
\arrow[rr,"{(\fsubA,B)}"] 
&&
(A',B) 
\arrow[dd,"{(A',\fsubB)}", bend left = 60, ""{xshift=0em,swap,name=rightsource}] 
%\arrow[dd,"{(A,\fsubBprime)}"{description}, bend right = 60, ""{xshift=-.9em,swap,name=leftsource}] 
\\
\\
(A,B')
\arrow[rr,"{(\fsubA,B')}"{swap},""{yshift=.2em,name=target}]
&&
(A',B')
\arrow[from=rightsource, to=righttarget, "{\graytile{\fsubA}{\fsubB}}"{swap}, double,-implies]
\arrow[from=leftsource, to=lefttarget, "{(A,\beta)}"{swap}, double,-implies]
\end{tikzcd}
$$
%\end{equation}
%%%%%%%%%%%%%%%%%%%%%%%%%%%%%%%%%%%%%%%%
%
%
%   Horizontal composition and identity on the left side
%
%
%%%%%%%%%%%%%%%%%%%%%%%%%%%%%%%%%%%%%%%%
\medbreak
\noindent
\textbf{[d] Coherence of the Gray commutation:}
% on the left and the right:}
the group of four relations below should be understood 
as coherence properties of the Gray commutation,
both on the left and on the right side of the Gray tensor product:
%\noindent
%\textbf{Horizontal composition and identity on the left side:}
\begin{equation}
%\begin{small}
\fbox{
$\begin{array}{c}
((\fsubA,B)\horcomp\graytile{\fsubAprime}{\fsubB})
\vertcomp
(\graytile{\fsubA}{\fsubB}\horcomp(\fsubAprime,B'))
=
\graytile{\fsubA\horcomp\fsubAprime}{\fsubB}
\\
\graytile{\horid{A}}{\fsubB}
=
\vertid{(A,\fsubB)}
\\
(\graytile{\fsubA}{\fsubB}\horcomp(A',\fsubBprime))
\vertcomp
((A,\fsubB)\horcomp\graytile{\fsubA}{\fsubB})
=
\graytile{\fsubA}{\fsubB\horcomp\fsubBprime}
\\
\graytile{\fsubA}{\horid{B}}
=
\vertid{(\fsubA,B)}
\end{array}$
}
%\end{small}
\end{equation}
The four equations between 2-cells are nicely expressed
in the diagrammatic language of pasting diagrams.
%The relations between 2-cells are nicely depicted as pasting diagrams:
The first equation says that the pasting diagram
%\begin{equation}
%\begin{small}
%\begin{array}{ccc}
$$
\begin{tikzcd}[column sep = 2em, row sep=2.2em]
(A,B) 
\arrow[dd,"{(A,\fsubB)}"{swap}]
\arrow[rr,"{(\fsubA,B)}", ""{yshift=-.2em,swap,name=leftsource}] 
&&
(A',B)
\arrow[rr,"{(\fsubAprime,B)}", ""{yshift=-.2em,swap,name=rightsource}] 
\arrow[dd,"{(A',\fsubB)}"{description}]
&&
(A'',B)
\arrow[dd,"{(A'',\fsubB)}"]
\\
\\
(A,B')
\arrow[rr,"{(\fsubA,B')}"{swap},""{yshift=.2em,name=lefttarget}]
&&
(A',B')
\arrow[rr,"{(\fsubAprime,B')}"{swap},""{yshift=.2em,name=righttarget}]
&&
(A'',B')
\arrow[from=leftsource, to=lefttarget, "{\graytile{\fsubA}{\fsubB}}", double,-implies]
\arrow[from=rightsource, to=righttarget, "{\graytile{\fsubAprime}{\fsubB}}", double,-implies]
\end{tikzcd}
$$
coincides with the 2-cell below:
$$
\begin{tikzcd}[column sep = 2em, row sep=2.2em]
(A,B) 
\arrow[dd,"{(A,\fsubB)}"{swap}]
\arrow[rr,"{(\fsubA\horcomp\fsubAprime,B)}", ""{yshift=-.2em,swap,name=source}] 
&&
(A',B)
\arrow[dd,"{(A',\fsubB)}"]
\\
\\
(A'',B)
\arrow[rr,"{(\fsubA\horcomp\fsubAprime,B)}"{swap},""{yshift=.2em,name=target}]
&&
(A'',B')
\arrow[from=source, to=target, "{\graytile{\fsubA\horcomp\fsubAprime}{\fsubB}}", double,-implies]
\end{tikzcd}
$$
%
%
%cas de l'identite
%
%
The second equation indicates that the 2-cell below
$$
%\begin{equation}
%\begin{footnotesize}
%\begin{array}{ccc}
\begin{tikzcd}[column sep = 2em, row sep=2em]
(A,B) 
\arrow[dd,"{(A,\fsubB)}"{swap}]
\arrow[rr,"{(\horid{A},B)}", ""{yshift=-.2em,swap,name=source}] 
&&
(A,B)
\arrow[dd,"{(A,\fsubB)}"]
\\
\\
(A,B')
\arrow[rr,"{(\horid{A},B')}"{swap},""{yshift=.2em,name=target}]
&&
(A,B')
\arrow[from=source, to=target, "{\graytile{\horid{A}}{\fsubB}}", double,-implies]
\end{tikzcd}
$$
coincides with the identity 2-cell
$$
\begin{tikzcd}[column sep = 2em, row sep=2em]
(A,B) 
\arrow[dd,"{(A,\fsubB)}"{swap}]
\arrow[rr,"{\horid{(A,B)}}", ""{yshift=-.2em,swap,name=source}] 
&&
(A,B)
\arrow[dd,"{(A,\fsubB)}"]
\\
\\
(A,B')
\arrow[rr,"{\horid{(A,B')}}"{swap},""{yshift=.2em,name=target}]
&&
(A,B')
\arrow[from=source, to=target, "{\vertid{(A,b)}}", double,-implies]
\end{tikzcd}
$$
%\end{array}
%\end{footnotesize}
%\end{equation}

%%%%%%%%%%%%%%%%%%%%%%%%%%%%%%%%%%%%%%%%
%
%
%   Horizontal composition and identity on the right side
%
%
%%%%%%%%%%%%%%%%%%%%%%%%%%%%%%%%%%%%%%%%
%\noindent
%\textbf{Horizontal composition and identity on the right side:}
%$$
%(\graytile{\fsubA}{\fsubB}\horcomp(A,\fsubB))
%\vertcomp
%((A',\fsubBprime)\horcomp\graytile{\fsubA}{\fsubB})
%=
%\graytile{\fsubA}{\fsubBprime\horcomp\fsubB}
%%
%\quad\quad\quad
%%
%\graytile{\fsubA}{\horid{B}}
%=
%\vertid{(\fsubA,B)}
%$$
%\begin{equation}
%\begin{footnotesize}
The third equation indicates that the pasting diagram
%\begin{array}{ccc}
$$
\begin{tikzcd}[column sep = 2em, row sep=2em]
(A,B) 
\arrow[dd,"{(A,\fsubB)}"{swap}]
\arrow[rr,"{(\fsubA,B)}", ""{yshift=-.2em,swap,name=uppersource}] 
&&
(A',B)
\arrow[dd,"{(A',\fsubB)}"]
\\
\\
(A,B')
\arrow[dd,"{(A,\fsubBprime)}"{swap}]
\arrow[rr,"{(\fsubA,B')}"{description},""{yshift=.4em,name=uppertarget},""{yshift=-.4em,swap,name=lowersource}]
&&
(A',B')
\arrow[dd,"{(A',\fsubBprime)}"]
\\
\\
(A,B'')
\arrow[rr,"{(\fsubA,B'')}"{swap},""{yshift=.2em,name=lowertarget}]
&&
(A',B'')
\arrow[from=uppersource, to=uppertarget, "{\graytile{\fsubA}{\fsubB}}", double,-implies]
\arrow[from=lowersource, to=lowertarget, "{\graytile{\fsubA}{\fsubBprime}}", double,-implies]
\end{tikzcd}
$$
coincides with the 2-cell below:
$$
\begin{tikzcd}[column sep = 2em, row sep=5em]
(A,B) 
\arrow[dd,"{(A,\fsubB\horcomp\fsubBprime)}"{swap}]
\arrow[rr,"{(\fsubA,B)}", ""{yshift=-.2em,swap,name=source}] 
&&
(A',B)
\arrow[dd,"{(A',\fsubB\horcomp\fsubBprime)}"]
\\
\\
(A,B'')
\arrow[rr,"{(\fsubA,B')}"{swap},""{yshift=.2em,name=target}]
&&
(A',B'')
\arrow[from=source, to=target, "{\graytile{\fsubA}{\fsubB\horcomp\fsubBprime}}", double,-implies]
\end{tikzcd}
$$
%cas de l'identite verticale
The fourth equation indicates that the 2-cell
$$
\begin{tikzcd}[column sep = 2em, row sep=2em]
(A,B) 
\arrow[dd,"{(A,\horid{B})}"{swap}]
\arrow[rr,"{(\fsubA,B)}", ""{yshift=-.2em,swap,name=source}] 
&&
(A',B)
\arrow[dd,"{(A',\horid{B})}"]
\\
\\
(A,B)
\arrow[rr,"{(\fsubA,B)}"{swap},""{yshift=.2em,name=target}]
&&
(A',B)
\arrow[from=source, to=target, "{\graytile{\fsubA}{\horid{B}}}", double,-implies]
\end{tikzcd}
$$
coincides with the identity 2-cell:
$$
\begin{tikzcd}[column sep = 2em, row sep=2em]
(A,B) 
\arrow[dd,"{\horid{(A,B)}}"{swap}]
\arrow[rr,"{(\fsubA,B)}", ""{yshift=-.2em,swap,name=source}] 
&&
(A',B)
\arrow[dd,"{\horid{(A',B)}}"]
\\
\\
(A,B)
\arrow[rr,"{(\fsubA,B)}"{swap},""{yshift=.2em,name=target}]
&&
(A',B)
\arrow[from=source, to=target, "{\vertid{(\fsubA,B)}}", double,-implies]
\end{tikzcd}
$$

\section{Proof of the preservation of coreflexive equalizers
(Prop.~\ref{proposition/preservation-of-equalizers}
in \S\ref{section/gray-preserves-equalizers})}
\label{appendix/proof-of-preservation-of-equalizers}
The preservation of coreflexive equalizers componentwise
is proved in three combined steps.
%carefully described in the Appendix \ref{appendix/proof-of-preservation-of-equalizers}.
%
First, by universality of the cartesian product in $\TwoCat$,
we know that the diagram
$$
\begin{tikzcd}[column sep = 2em]
\Ctwocategory\times\Etwocategory
\arrow[rr,dashed,"{\Ccategory\times E}"]
&&
\Ctwocategory\times\Atwocategory
\arrow[rr,"{\Ccategory\times F}",yshift=.2em]
\arrow[rr,"{\Ccategory\times G}"{swap},yshift=-.2em]
&&
\Ctwocategory\times\Btwocategory
\end{tikzcd}
$$
exhibits $\Ctwocategory\times\Etwocategory$ as a coreflexive equalizer
of the pair of 2-functors $\Ccategory\times F$ and $\Ccategory\times G$.
Then, we observe that the diagram
$$
\begin{tikzcd}[column sep = 2em]
{\underlyingcat{\Etwocategory}}
\arrow[rr,dashed,"{\underlyingcat{E}}"]
&&
{\underlyingcat{\Atwocategory}}
\arrow[rr,"{\underlyingcat{F}}",yshift=.2em]
\arrow[rr,"{\underlyingcat{G}}"{swap},yshift=-.2em]
&&
{\underlyingcat{\Btwocategory}}
\end{tikzcd}
$$
is a coreflexive equalizer in $\Cat$.
A simple combinatorial exercise establishes that the functor
$$
\begin{tikzcd}[column sep = 1em]
{{\Dcategory\funnytensor -}
\quad = \quad
\Xcategory\mapsto\Dcategory\funnytensor\Xcategory}
\quad : \quad
\Cat \arrow[rr] && \Cat
\end{tikzcd}
$$
obtained by tensoring using the "funny tensor product"
a category $\Xcategory$ with a fixed category $\Dcategory$ 
preserves coreflexive equalizers.
From this follows that the diagram below 
$$
\begin{tikzcd}[column sep = 2em]
{\underlyingcat{\Ctwocategory}\funnytensor\underlyingcat{\Etwocategory}}
\arrow[rr,dashed,"{\underlyingcat{\Ccategory}\funnytensor \underlyingcat{E}}"]
&&
{\underlyingcat{\Ctwocategory}\funnytensor\underlyingcat{\Atwocategory}}
\arrow[rr,"{\underlyingcat{\Ctwocategory}\funnytensor \underlyingcat{F}}",yshift=.2em]
\arrow[rr,"{\underlyingcat{\Ctwocategory}\funnytensor \underlyingcat{G}}"{swap},yshift=-.2em]
&&
{\underlyingcat{\Ctwocategory}\funnytensor\underlyingcat{\Btwocategory}}
\end{tikzcd}
$$
exhibits the category ${\underlyingcat{\Ctwocategory}\funnytensor\underlyingcat{\Etwocategory}}$
as the coreflexive equalizer in $\Cat$ of the pair of functors ${\underlyingcat{\Ctwocategory}\funnytensor \underlyingcat{F}}$
and ${\underlyingcat{\Ctwocategory}\funnytensor \underlyingcat{G}}$.

In the third and last step, we establish
that (\ref{equation/equalizer-between-gray-tensor-products}) 
is an equalizer in $\TwoCat$ by observing that given a 2-category $\Xcategory$,
a 2-functor $$H:\Xcategory\to\Ctwocategory\graytensor\Etwocategory$$
is the same thing as a pair $(\varphi,\psi)$ consisting of a functor 
$${\varphi}:\underlyingcat{\Xcategory}\to\underlyingcat{\Ctwocategory}\funnytensor\underlyingcat{\Etwocategory}$$
describing the 2-functor $H$ on objects and morphisms,
and of a 2-functor 
$${\psi}:\Xcategory\to\Ctwocategory\times\Etwocategory$$
describing the 2-functor~$H$ on cells, such that 
the diagram below
$$
\begin{tikzcd}[column sep = 2.8em, row sep=.8em]
{\underlyingcat{\Xcategory}}
\arrow[rr,"{\underlyingcat{\psi}}"]\arrow[dd,"{\varphi}"{swap}]
&&
{\underlyingcat{\Ctwocategory\times\Etwocategory}}
\arrow[dd,"{iso}"]
\\
\\
{\underlyingcat{\Ctwocategory}\funnytensor\underlyingcat{\Etwocategory}}
\arrow[rr,"{canonical}"]
&&
{\underlyingcat{\Ctwocategory}\times\underlyingcat{\Etwocategory}}
\end{tikzcd}
$$
commutes in $\Cat$.
Using the fact that $\underlyingcat{\Ctwocategory}\funnytensor\underlyingcat{\Etwocategory}$ 
and $\Ctwocategory\times\Etwocategory$ are equalizers in $\Cat$ and $\TwoCat$, respectively,
we conclude that it is the same thing as pair $(\varphi',\psi')$
consisting of a functor 
$${\varphi'}:\underlyingcat{\Xcategory}\to\underlyingcat{\Ctwocategory}\funnytensor\underlyingcat{\Atwocategory}$$
and of a 2-functor
$${\psi}:\Xcategory\to\Ctwocategory\times\Atwocategory$$
such that
$$
\begin{tikzcd}[column sep = 2.8em, row sep=.8em]
{\underlyingcat{\Xcategory}}
\arrow[rr,"{\underlyingcat{\psi'}}"]\arrow[dd,"{\varphi'}"{swap}]
&&
{\underlyingcat{\Ctwocategory\times\Atwocategory}}
\arrow[dd,"{iso}"]
\\
\\
{\underlyingcat{\Ctwocategory}\funnytensor\underlyingcat{\Atwocategory}}
\arrow[rr,"{canonical}"]
&&
{\underlyingcat{\Ctwocategory}\times\underlyingcat{\Atwocategory}}
\end{tikzcd}
$$
commutes, and making the expected diagrams commute between $F$ and $G$.
%:\Ccategory\times\Acategory\to\Ccategory\times\Bcategory.$
%
%which tranports every asynchronous graph $G$ to its associated 2-category.
%
%Interestingly, the category $\Asynch$ has finite limits just as the category $\TwoCat$.
%
%We observe that equalizers 
%are preserved by (\ref{equation/functor-from-asynch-to-twocat}) while this is not the case for cartesian products.
We conclude that the Gray tensor product preserves coreflexive equalizers componentwise.

\section{A detailed description of the functor\\
from asynchronous graphs to 2-categories}
\label{section/construction-of-the-two-category}
We give a detailed description of the functor
$$\asynchtwocat{-} \,\,\, : \,\,\, (\Asynch,\shuffletensor,\textbf{I}) \,\,\, \longrightarrow \,\,\, (\TwoCat,\graytensor,\textbf{I})$$
which associates to every asynchronous graph $(G,\diamond)$
a 2-category $\asynchtwocat{G,\diamond}$.
As it is well-known, every graph $G$ induces a free category $\asynchtwocat{G}$
whose objects are the vertices of $G$ and whose morphisms are the paths of $G$.
In all this section, we suppose given an asynchronous graph $(G,\diamond)$
and describe how the construction $G\mapsto\asynchtwocat{G}$ may be extended
in order to associate a 2-category $\asynchtwocat{G,\diamond}$ to the asynchronous graph $(G,\diamond)$.
The starting point of the construction is as expected:
the 2-category $\asynchtwocat{G,\diamond}$ has the category $\asynchtwocat{G}$
generated by $G$ as underlying category.
Our main purpose in this section is thus to define 
the 2-cells of the 2-category $\asynchtwocat{G,\diamond}$.
This is done in two steps, as follows.

\medbreak
\subsection{From asynchronous graphs to sesquicategories}\label{section/asynchronous-graph/from-ag-to-sc}
A \emph{permutation step} $\gamma=(h_1,p,q,h_2)$
is defined as a quadruple consisting
of a pair of paths $p,q:P\transitionpath Q$ involved in a permutation tile $p\diamond q$
%in the asynchronous graph $(G,\diamond)$, 
together with two paths $h_1:M\transitionpath P$ and $h_2:Q\transitionpath N$.
%
%To every permutation tile
%are declared \emph{equivalent modulo the permutation tile} $p\diamond q$
%when there are two paths $h_1:M\transitionpath P$ and $h_2:Q\transitionpath N$
%such that the paths $h$, $h'$ factor as $h=h_1\cdot p\cdot h_2$
%and $h'=h_1\cdot q\cdot h_2$.
%%
%We write $h\sim h' (\textrm{mod} \,\, p\diamond q)$.
We use the notation 
\begin{equation}\label{equation/permutation-step}
\begin{tikzcd}[column sep =1.2em]
\gamma = h_1\cdot(p,q)\cdot h_2
\quad : \quad
f
\arrow[rr,-implies,double,"{}"]
&&
g
\quad : \quad
x
\arrow[r]
&
y
\end{tikzcd}
\end{equation}
for such a permutation step $\gamma=(h_1,p,q,h_2)$,
where $f,g: x\transitionpath y$
are the two paths obtained by concatenation
$f=h_1\cdot p\cdot h_2$ and $g=h_1\cdot q\cdot h_2$
in the graph $G$.
A \emph{permutation sequence} $\varphi:f\Rightarrow g$ 
is defined as a sequence of permutation steps,
of the following form:
\begin{equation}\label{equation/permutation-sequence}
\begin{tikzcd}[column sep=1em]
f=f_1
\arrow[rr,-implies,double,"{\gamma_1}"]
&&
f_2
\arrow[rr,-implies,double,"{\gamma_2}"]
&&
\cdots
\arrow[rr,-implies,double,"{\gamma_n}"]
&&
f_{n+1}=g
\quad
:
\quad
x
\arrow[r]
&
y
\end{tikzcd}
\end{equation}
Every pair of vertices $x$, $y$
of the asynchronous graph $(G,\diamond)$ induces a category 
noted $\permutationcategory{G,\diamond}{x}{y}$
whose objects are the paths $f:x\transitionpath y$
and whose morphisms are the permutation sequences
$\varphi:f\Rightarrow g$.
Composition in $\permutationcategory{G,\diamond}{x}{y}$
is called \emph{vertical composition} 
and defined as concatenation of permutation sequences.
The vertical composition of the permutation sequences
\begin{equation}\label{equation/permutation-sequence-bis}
\begin{array}{c}
\begin{tikzcd}[column sep =1.6em]
\varphi
\quad : \quad
f
\arrow[rr,-implies,double,"{}"]
&&
g
\quad : \quad
x
\arrow[r]
&
y
\end{tikzcd}
\\
\begin{tikzcd}[column sep =1.6em]
\psi
\quad : \quad
g
\arrow[rr,-implies,double,"{}"]
&&
h
\quad : \quad
x
\arrow[r]
&
y
\end{tikzcd}
\end{array}
\end{equation}
induces the permutation sequence noted
\begin{equation}\label{equation/permutation-sequence-quatro}
\begin{tikzcd}[column sep =1.6em]
\varphi\vertcomp\psi
\quad : \quad
f
\arrow[rr,-implies,double,"{}"]
&&
h
\quad : \quad
x
\arrow[r]
&
y
\end{tikzcd}
\end{equation}
and defined by concatenation of $\varphi$ and $\psi$.
It is worth observing that the category $\permutationcategory{G,\diamond}{x}{y}$ just defined
is simply the free category generated by the permutation steps $\gamma:f\Rightarrow g$ 
between paths $f,g:x\transitionpath y$.
The notation and terminology used for vertical composition is guided by the intuition 
that every permutation sequence $\varphi:f\Rightarrow g$ defines a 2-cell represented
diagramatically as
$$
\begin{tikzcd}
x\arrow[rr, bend left,"f",""{swap,name=source}]
\arrow[rr, bend right,"g"{swap},""{name=target}] && y
\arrow[from=source, to=target, "{\varphi}", double,-implies]
\end{tikzcd}
$$
and that vertical composition corresponds to vertical diagram pasting:
$$
\begin{array}{ccc}
\begin{tikzcd}
x\arrow[rr, bend left = 60,"f",""{swap,name=source}]
\arrow[rr, bend right = 60,"h"{swap},""{name=target}] 
&& y
\arrow[from=source, to=target, "{\varphi\vertcomp\psi}", double,-implies]
\end{tikzcd}
& \quad = \quad &
\begin{tikzcd}
x\arrow[rr, bend left = 60,"f",""{swap,name=uppersource}]
\arrow[rr, "g"{description},""{name=uppertarget},""{swap,name=lowersource}] 
\arrow[rr, bend right = 60,"h"{swap},""{name=lowertarget}] 
&& y
\arrow[from=uppersource, to=uppertarget, "{\varphi}", double,-implies]
\arrow[from=lowersource, to=lowertarget, "{\psi}", double,-implies]
\end{tikzcd}
\end{array}
$$
In order to establish the statement, we turn our attention to the definition of horizontal composition
on paths and permutation sequences.
The first thing to observe is that every permutation step of the form (\ref{equation/permutation-step}) 
can be extended by a pair of paths $j_1:x'\transitionpath x$
and $j_2:y\transitionpath y'$
in order to obtain a permutation step
\begin{equation}\label{equation/permutation-step-bis}
\begin{tikzcd}[column sep =.5em]
j_1\horcomp\gamma\horcomp j_2 
\,\, : \,\,
j_1\cdot f \cdot j_2
\arrow[rr,-implies,double,"{}"]
&&
j_1\cdot g \cdot j_2
\,\, : \,\,
x'
\arrow[r]
&
y'
\end{tikzcd}
\end{equation}
defined as expected:
$$
j_1\horcomp\gamma\horcomp j_2 = j_1\cdot h_1\cdot (p,q) \cdot h_2\cdot j_2.
$$
As we will see, there are good reasons for using the symbol~$\horcomp$
as notation for this operation,
which we call \emph{horizontal composition.}
In order to keep our notations consistent, we also allow ourselves to notation $\horcomp$ for concatenation of path,
instead of the usual notation $p,q\mapsto p\cdot q$.
In the same way as for permutation steps, every permutation sequence
\begin{equation}\label{equation/permutation-sequence-bis}
\begin{tikzcd}[column sep=1.2em]
\varphi
\quad 
:
\quad
f
\arrow[rr,-implies,double]
&&
g
\quad
:
\quad
x
\arrow[rr]
&&
y
\end{tikzcd}
\end{equation}
and every pair of paths $j_1:x'\transitionpath x$
and $j_2:y\transitionpath y'$
induce a permutation sequence
\begin{equation}\label{equation/permutation-sequence-bis}
\begin{tikzcd}[column sep=1.4em, row sep=0em]
j_1\horcomp \varphi\horcomp j_2
\,\, 
:
\,\,
j_1\horcomp f \horcomp j_2
\arrow[rr,-implies,double]
&&
j_1\horcomp g \horcomp j_2
\\
\hspace{1.5em}
:
\,\,
x'
\arrow[r]
&
y'
\end{tikzcd}
\end{equation}
obtained by extending every permutation step $\gamma_i$
%f_i\Rightarrow f_{i+1}$
defining the permutation sequence in (\ref{equation/permutation-sequence})
by horizontal composition with $j_1:x'\transitionpath x$ and with $j_2:y\transitionpath y'$,
in the following way:
%
%\begin{equation}\label{equation/permutation-sequence}
$$\begin{tikzcd}[column sep=3.1em, row sep=0em]
j_1\horcomp f_1\horcomp j_2
\,\,\,\,
\arrow[rr,-implies,double,"{j_1\horcomp\gamma_1\horcomp j_2}"]
&&
\,\,\,\,
j_1\horcomp f_2\horcomp j_2
\\
\hspace{5.9em}
\arrow[rr,-implies,double,"{j_1\horcomp\gamma_2\horcomp j_2}"]
&&
\hspace{2.6em}
\cdots
\hspace{2.6em}
\\
\hspace{2.2em}
\cdots
\hspace{2.2em}
\arrow[rr,-implies,double,"{j_1\horcomp\gamma_n\horcomp j_2}"]
&&
\,\,\,\,
j_1\horcomp f_n\horcomp j_2.
\end{tikzcd}
$$
%\end{equation}
The notation and terminology we have just used for horizontal composition
is justified by the idea that every permutation sequence $\varphi:f\Rightarrow g$
defines a 2-cell
$$
\begin{tikzcd}
x\arrow[rr, bend left,"f",""{swap,name=source}]
\arrow[rr, bend right,"g"{swap},""{name=target}] && y
\arrow[from=source, to=target, "{\varphi}", double,-implies]
\end{tikzcd}
$$
and that horizontal composition with $j_1:M'\transitionpath M$
and $j_2:N\transitionpath N'$ defines a \emph{whiskering} operation
on the 2-cell, whose result may be depicted as follows:
$$
\begin{tikzcd}
x'\arrow[rr,"{j_1}"]
&&
x\arrow[rr, bend left,"f",""{swap,name=source}]
\arrow[rr, bend right,"g"{swap},""{name=target}] && y
\arrow[from=source, to=target, "{\varphi}", double,-implies]
\arrow[rr,"{j_2}"]
&&
y'
\end{tikzcd}
$$
Before establishing that $\asynchtwocat{G,\diamond}$
defines a 2-category (Prop. \ref{proposition/two-category-from-asynchronous-graph}),
we observe that the whiskering operation just defined
satisfies all the properties expected of a sesquicategory.
This means that
\medbreak
\begin{proposition}\label{proposition/two-category-from-asynchronous-graph}
Every asynchronous graph $(G,\diamond)$ defines a sesquicategory $\asynchsesquicat{G,\diamond}$
with the vertices of $G$ as objects, the paths of $G$ as morphisms, 
the permutation sequences as 2-cells, 
and vertical composition $\varphi\vertcomp\psi$ 
defined by concatenation of permutation sequences.
\end{proposition}

%
%Does not satisfy the Godement law: ...
%Note that $\cong$ is a congruence in the sense that it is preserved by whiskering.
%Hence, $\asynchtwocat{G,\diamond}$ defines a sesquicategory.
%Then, it satisfies the Godement law ...

%\begin{theorem}
%$\asynchtwocat{}$ defines a functor.
%\end{theorem}

\medbreak
\subsection{From asynchronous graphs to 2-categories}\label{section/asynchronous-graph/from-ag-to-tc}
Coming back to the definition of permutation sequence,
it is worth observing that every permutation step $\gamma:f\Rightarrow g:M\to N$
relates two paths $f,g$ of the same length.
From this follows that every path $f_i$ for $1\leq i\leq n+1$
involved in the permutation sequence (\ref{equation/permutation-sequence})
has the same length, noted $k$.
%Suppose that the length is $k$.
%
In particular, the two paths $f,g:M\transitionpath N$ have the same length,
and thus both sets of indices $\lengthindex{f}$ and $\lengthindex{g}$ are equal to $\{1,\dots,k\}$,
see \S\ref{section/asynchronous-graph/graphs} for the notation.
%
%Let $k\in\Nat$ denote this common length.
%
At this stage, we want to give a very simple recipe to associate a bijective function $[\varphi]:\lengthindex{f}\to\lengthindex{g}$
to a permutation sequence $\varphi:f \Rightarrow g$ of the form (\ref{equation/permutation-sequence}).
Every permutation step $\gamma=(h_1,p,q,h_2)$ of the form (\ref{equation/permutation-step})
between two paths $f=u_1\cdots u_k$ and $g=v_1\cdots v_k$ of length $k$ 
defines a bijective function $\lengthindex{\gamma}:\lengthindex{f}\to\lengthindex{g}$ in the following way:
$$
\lengthindex{\gamma} \quad = \quad \left\{\begin{array}{lcl} 
\ell+1\mapsto \ell +2 & & \\
% \mbox{when $i=\ell+1$} \\
\ell+2\mapsto \ell +1 & &\\
% \mbox{when $i=\ell+2$} \\
i\mapsto i & & \mbox{when $i\leq \ell$ or $i\geq \ell+3$}
\end{array}
\right.
$$
where $\ell$ denotes the length of the path $h_1=u_1\cdots u_{\ell}=v_1\cdots v_{\ell}$.
The intuition behind the definition is that the bijective function 
$$
\begin{tikzcd}[column sep = 1em]
\lengthindex{\gamma} \,\, : \,\, \lengthindex{f}\arrow[rr]
&&
\lengthindex{g}
\end{tikzcd}
$$
%$\lengthindex{\gamma}:\lengthindex{f}\to\lengthindex{g}$ 
reorders the edges of $f$ and $g$ in the way indicated by the permutation tile $p\diamond q$.
Pictorially, the permutation step 
$$
\begin{tikzcd}[column sep = 1em]
{\gamma} 
\,\, : \,\,
{f}\arrow[-implies,double,rr]
&&
{g}  \,\, : \,\, x \arrow[rr] && y.
\end{tikzcd}
$$
%$\gamma:f\Rightarrow g$ 
with permutation tile $p\diamond q$
depicted below between $p=u_{\ell+1}\cdot u_{\ell+2}$ and $q=v_{\ell+1}\cdot v_{\ell+2}$ 
$$
\raisebox{-2.8em}{\includegraphics[width=24em]{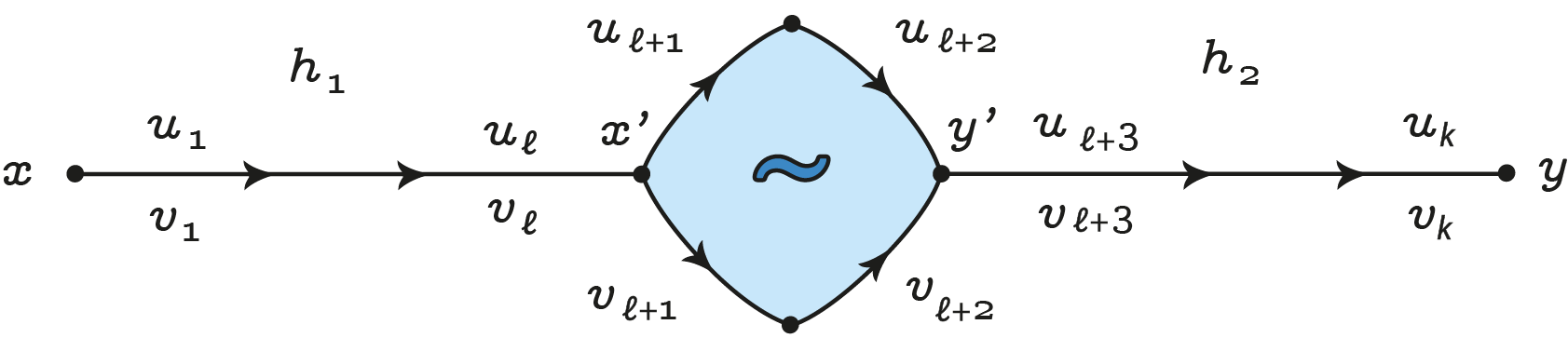}}
$$
is turned into the bijective function $\lengthindex{\gamma}:\lengthindex{f}\to\lengthindex{g}$ 
on the set of indices $[f]=[g]=\{0,\dots,k\}$ of edges appearing in the paths $f$ and $g$, 
defined as the transposition $\ell+1\mapsto \ell+2$ and $\ell+2\mapsto \ell+1$ 
indicated by the pair of purple (dark grey) arrows drawn on the permutation tile $p\diamond q$:
$$
\raisebox{-2.8em}{\includegraphics[width=24em]{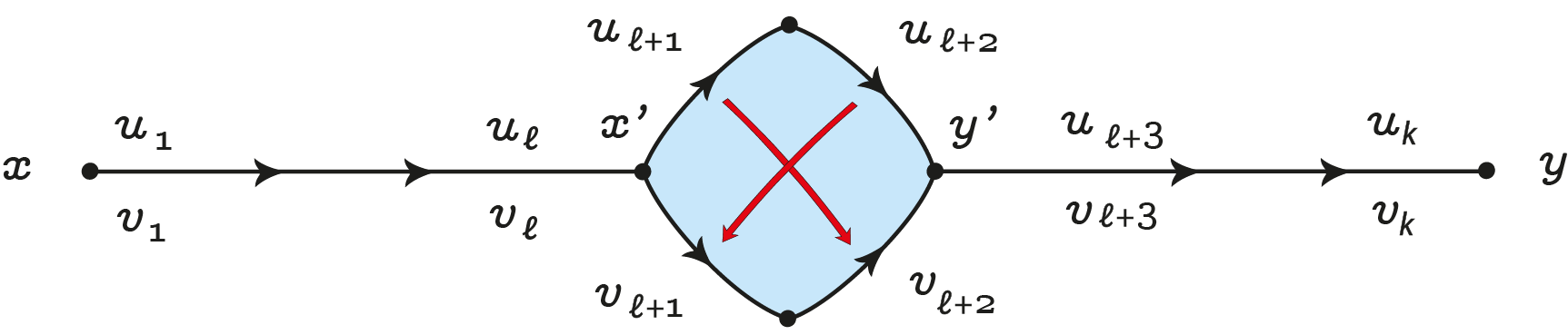}}
$$
The definition just given of the bijective function
$$
\begin{tikzcd}[column sep = 1em]
\lengthindex{\gamma} \,\, : \,\, \lengthindex{f}\arrow[rr]
&&
\lengthindex{g}
\end{tikzcd}
$$
associated to a permutation step
$$
\begin{tikzcd}[column sep = 1em]
{\gamma} 
\,\, : \,\,
{f}\arrow[-implies,double,rr]
&&
{g}  \,\, : \,\, x \arrow[rr] && y.
\end{tikzcd}
$$
extends to a permutation sequence
$$
\begin{tikzcd}[column sep = 1em]
{\varphi} 
\,\, : \,\,
{f}\arrow[-implies,double,rr]
&&
{g}  \,\, : \,\, x \arrow[rr] && y.
\end{tikzcd}
$$
in the expected way: the associated bijective function
$$
\begin{tikzcd}[column sep = 1em]
\lengthindex{\varphi} \,\, : \,\, \lengthindex{f}\arrow[rr]
&&
\lengthindex{g}
\end{tikzcd}
$$
%associated to a permutation sequence $\varphi:f\Rightarrow g$
%of the form (\ref{equation/permutation-sequence}) 
is defined as the (set-theoretic) composite of the bijective functions 
$\lengthindex{\gamma_i}:\lengthindex{f_i}\to\lengthindex{g_i}$ associated
to each of the permutation steps $\gamma_i:f_i\Rightarrow f_{i+1}$, for $1\leq i\leq k$.
This definition can be simply expressed as the equation
which says that the bijection
$$
\begin{tikzcd}[column sep =1.4em]
\lengthindex{f}
\arrow[rr,"{\lengthindex{\varphi}}"]
&&
\lengthindex{g}
\end{tikzcd}
$$
is equal to the composite
$$
\begin{tikzcd}[column sep =1.4em]
{\lengthindex{f}=\lengthindex{f_1}}
\arrow[rr,"{\lengthindex{\gamma_1}}"]
&&
{\lengthindex{f_2}}
\arrow[rr,"{\lengthindex{\gamma_2}}"]
&&
\cdots
\arrow[rr,"{\lengthindex{\gamma_n}}"]
&&
{\lengthindex{{f_{n+1}}}=\lengthindex{g}}
\end{tikzcd}
$$
This discussion leads us to introduce the equivalence relation $\cong$ 
which identifies two permutation sequences
$$
\begin{tikzcd}[column sep = 1em]
{\varphi}, {\psi}
\,\, : \,\,
{f}\arrow[-implies,double,rr]
&&
{g}  \,\, : \,\, x \arrow[rr] && y.
\end{tikzcd}
$$
precisely when they induce the same bijective function
$$
\begin{tikzcd}[column sep = 1em]
\lengthindex{\varphi} = \lengthindex{\psi}
\,\, : \,\, \lengthindex{f}\arrow[rr]
&&
\lengthindex{g}
\end{tikzcd}
$$
By way illustration, observe that the two permutation sequences
$$
\begin{tikzcd}[column sep =1.2em]
\varphi,\psi
\,\, : \,\,
u_1\cdot u_2\cdot u_3
\arrow[rr,-implies,double,"{}"]
&&
v_1\cdot v_2\cdot v_3
\,\, : \,\,
x
\arrow[r]
&
y
\end{tikzcd}
$$
involved in the cube property (\ref{equation/cube}) induce the very same bijective function
$$
\begin{tikzcd}[column sep = 1em]
\lengthindex{\varphi}, \lengthindex{\psi}
\,\, : \,\, \lengthindex{f}\arrow[rr]
&&
\lengthindex{g}
\end{tikzcd}
$$
which reverses the usual order on $\{1,2,3\}$
and thus turns every index $i\in\{1,2,3\}$ into the index $4-i\in\{1,2,3\}$.
This means that the two permutations sequences $\varphi$ and $\psi$
are equivalent modulo $\cong$ in a situation nicely depicted as:
\begin{center}
\begin{tabular}{ccc}
\raisebox{-3.4em}{\includegraphics[height=7.5em]{figure9b.png}}
& $\cong$ & 
\raisebox{-3.4em}{\includegraphics[height=7.5em]{figure10b.png}}
\end{tabular}
\end{center}
%
%We are now ready to define the notion of 2-cell for the 2-category $\asynchtwocat{G,\diamond}$.
%
Suppose given two paths $f,g:x\transitionpath y$.
We are now ready to define the notion of \emph{rescheduling} 
of the form
%for the 2-category $\asynchtwocat{G,\diamond}$.
%A reordering
\begin{equation}\label{equation/reordering}
\begin{tikzcd}[column sep =1.2em]
\varphi
\quad : \quad
f
\arrow[rr,-implies,double,"{}"]
&&
g
\quad : \quad
x
\arrow[rr]
&&
y
\end{tikzcd}
\end{equation}
%the 2-category $\asynchtwocat{G,\diamond}$ of the form
as a permutation sequence $\varphi:f\Rightarrow g$ considered modulo the equivalence relation $\cong$ just introduced.
In other words, a rescheduling
% (\ref{equation/reordering}) 
is an equivalence class modulo $\cong$ on the set of permutation sequences $\varphi:f\Rightarrow g$ 
transforming the path $f$ into the path $g$.
We establish that:
\medbreak
\begin{proposition}
Every asynchronous graph $(G,\diamond)$
defines a 2-category $\asynchtwocat{G,\diamond}$
with the vertices of $G$ as objects, the paths of $G$
as morphisms, and the reschedulings between paths
%of $(G,\diamond)$
as 2-cells.
\end{proposition}

\newpage
%\onecolumn 
\section{Asynchronous strategies as bicomodules}
We find clarifying to illustrate with a few pictures
the idea that an asynchronous strategy 
$$
\begin{tikzcd}[column sep=.8em, row sep=1.2em]
\sigma\,=\,(S,\mathsf{coact}_{\sigma},\lambda_{\sigma}) \,\, : \,\, (A,\lambda_A) \arrow[spanmap]{rrrr} &&&& (B,\lambda_B)
\end{tikzcd}
$$
%between template games $(A,\lambda_A)$ and $(B,\lambda_B)$
%is defined as a triple $\sigma=(S,\mathsf{coact}_{\sigma},\lambda_{\sigma})$
is described as a triple
$$(S,\mathsf{coact}_{\sigma},\lambda_{\sigma})$$
consisting of 2-category~$S$ used as support,
%called the support of the strategy $\sigma$
together with an $A,B$-bicomodule structure
$$
\begin{tikzcd}[column sep = 1em]
\mathsf{coact}_{\sigma}\quad : \quad
S
\arrow[rr] && 
{A} \graytensor S\graytensor {B}
\end{tikzcd}
$$
and a polarity 2-functor 
$$
\begin{tikzcd}[column sep = 1em]
\lambda_{\sigma}
\quad : \quad
S
\arrow[rr] && 
\anchorofstrat
%{\twocatanchorof{\polarityminussource,\polarityplussource,\polarityminustarget,\polarityplustarget}}
\end{tikzcd}
$$
The coaction $\mathsf{coact}_{\sigma}$ may be represented diagrammatically as follows:
\begin{center}
\raisebox{-3.4em}{\includegraphics[height=7.5em]{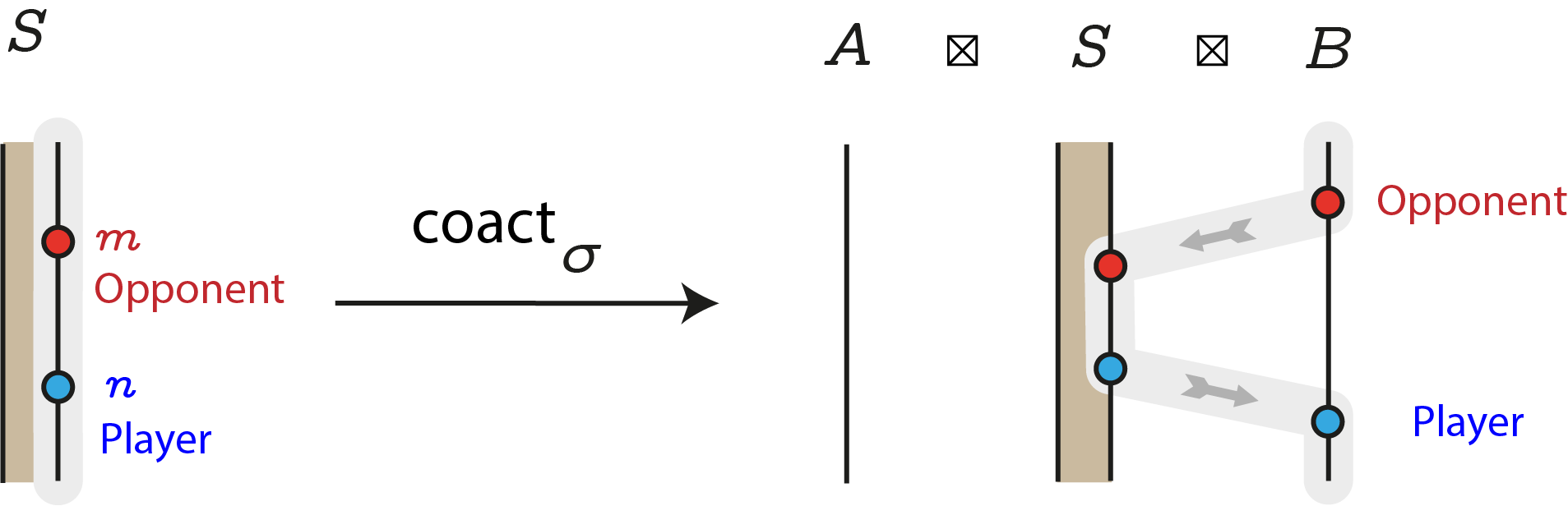}}
\end{center}
where the fact that the two moves $m$ and $n$ in the support $S$ 
are assigned the polarities
$$
\lambda_{\sigma}  \quad : \quad
m \mapsto \polarityminustarget
\quad\quad
n\mapsto \polarityplustarget
$$
is indicated by their position on the brown (or grey) ribbon
representing the strategy~$\sigma$ in the picture.
Similarly, the asynchronous strategy 
$$
\begin{tikzcd}[column sep=.8em, row sep=1.2em]
\tau\,=\,(T,\mathsf{coact}_{\tau},\lambda_{\tau}) \,\, : \,\, (B,\lambda_B) \arrow[spanmap]{rrrr} &&&& (C,\lambda_C)
\end{tikzcd}
$$
is represented as
\begin{center}
\raisebox{-4em}{\includegraphics[height=9em]{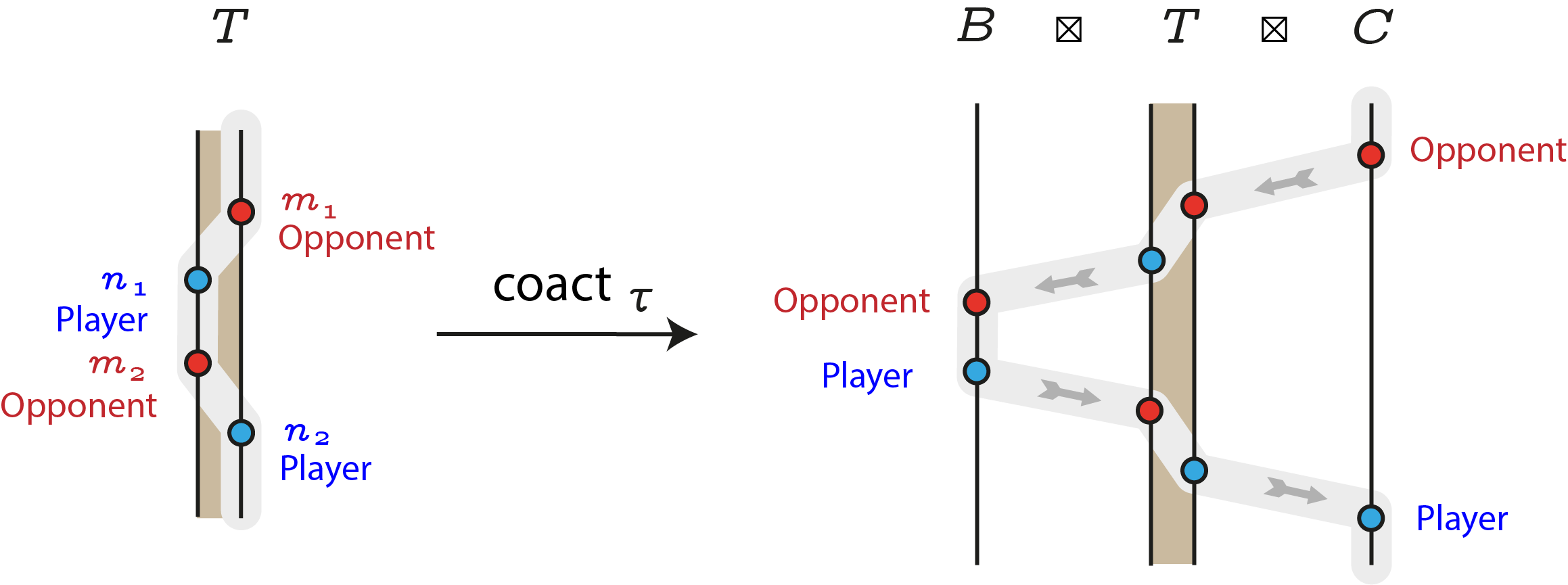}}
\end{center}
where the fact that the four moves $m_1$, $n_1$, $m_2$ and $n_2$
in the support $T$ are assigned the polarities
$$
\lambda_{\tau}  \quad : \quad
m_1 \mapsto \polarityminustarget
\quad\quad
n_1\mapsto \polarityplussource
\quad\quad
m_2 \mapsto \polarityminussource
\quad\quad
n_2\mapsto \polarityplustarget
$$
is indicated by their position on the brown (or grey) ribbon
representing the strategy~$\tau$ in the picture.
The composition of the two asynchronous strategies~$\sigma$
and~$\tau$ is then computed as the coreflexive equalizer
$$
\begin{tikzcd}[column sep = 2em]
{S\graytensoreq{B}T}
\arrow[rr,dashed,"{equ}"]
&&
{S\graytensor T}
\arrow[rr,"{\mathsf{coact}_{\sigma}^{\mathsf{right}}\graytensor T}",yshift=.2em]
\arrow[rr,"{S\graytensor\mathsf{coact}_{\tau}^{\mathsf{left}}}"{swap},yshift=-.2em]
&&
S\graytensor B\graytensor T
\end{tikzcd}
$$
represented pictorially as follows:
\begin{center}
\raisebox{-4em}{\includegraphics[height=8em]{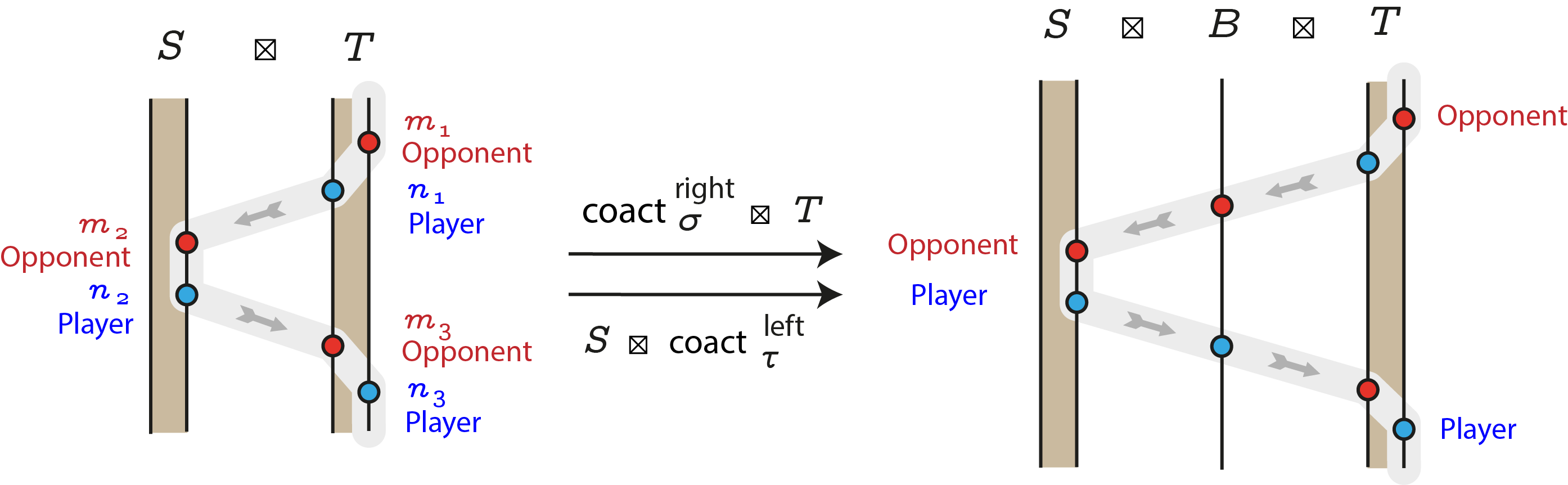}}
\end{center}
We indicate on the left a trajectory of six moves 
$$
m_1\cdot n_1 \cdot m_2\cdot n_2 \cdot m_3\cdot n_3
$$
in the 2-category and template game $S\graytensor T$
living at the same time in the coreflexive equalizer $S\graytensoreq{B} T$.
Note that the two moves $n_1$ played by~$\tau$ 
and $m_2$ played by~$\sigma$ are equal moves 
when seen from the point of view of the game~$B$.
Similarly, the two moves $n_2$ played by~$\sigma$ 
and $m_3$ played by~$\tau$ are equal moves
when seen from the point of view of the game~$B$.
This sequence of six moves may be thus described
in $S\graytensoreq{B} T$ as a sequence of four moves:
$$
m_1\cdot m\cdot n\cdot n_3
$$
as depicted below
\begin{center}
\raisebox{-4em}{\includegraphics[height=11em]{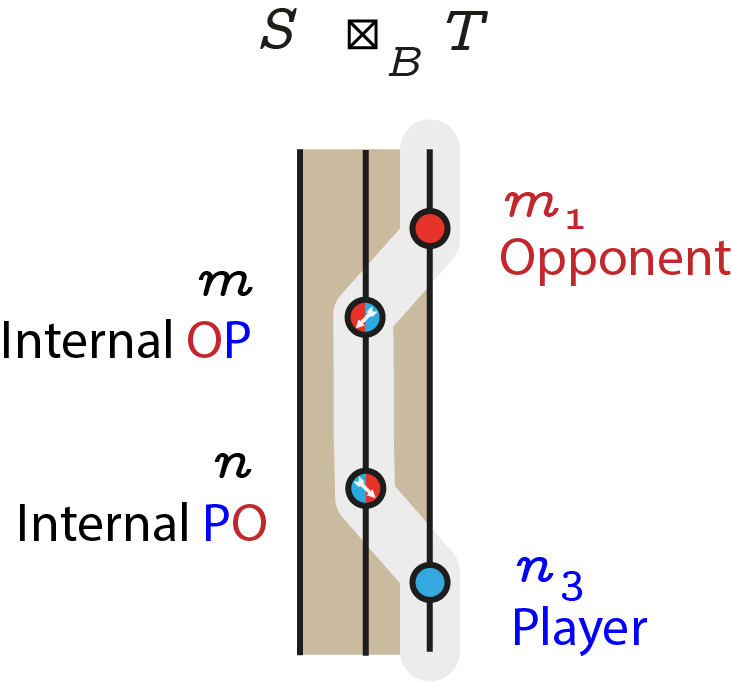}}
\end{center}
The position of the four moves on the brown (grey) area
describing $S\graytensoreq{B} T$ indicate the polarity
$$
\begin{array}{ccccccc}
m_1& \mapsto&\polarityminustarget
& \quad &
m &\mapsto&\polarityminusmid
\\
n &\mapsto&\polarityplusmid
& \quad &
n_3&\mapsto&\polarityplustarget
\end{array}
$$
assigned in the template of polarities
$$
\anchorofasynch[2] = \twocatanchorof{\polarityminussource,\polarityplussource,
\polarityminusmid,\polarityplusmid,\polarityminustarget,\polarityplustarget}
$$
itself defined as the coreflexive equalizer
$$
\anchorofasynch[2] = {\moo[1]\graytensoreq{\moo[0]}\moo[1]}.
$$
Note that postcomposing with multiplication
$$
\begin{tikzcd}[column sep = .8em]
\mathsf{mult} \,\, : \,\, {\anchorofasynch[2]} \arrow[rr] && {\anchorofasynch[1]}
\end{tikzcd}
$$
gives rise to the polarity function
$$
\begin{tikzcd}[column sep = 1em]
\lambda_{\tau\circ\sigma} \,\, : \,\, {S\graytensoreq{B} T} 
\arrow[rr] && {\anchorofasynch[2]} \arrow[rr] && {\anchorofasynch[1]}
 \end{tikzcd}
$$
on the four moves of the trajectory:
$$
\begin{array}{ccccccc}
m_1& \mapsto&\polarityminustarget
& \quad &
m &\mapsto&\id{\lrangle{\ast}}
\\
n &\mapsto&\id{\lrangle{\ast}}
& \quad &
n_3&\mapsto&\polarityplustarget
\end{array}
$$
in the template of polarities $\anchorofasynch[1]=\anchorofstrat$.
This polarity on the trajectory of the composite strategy $\tau\circ\sigma$
reflects the fact that the two moves $m$ and $n$
are considered as internal moves in the interaction between $\sigma$ and $\tau$;
while the moves $m_1$ and $n_3$ of respective polarities $\polarityminustarget$
and $\polarityplustarget$ are the output of the interaction, 
played by $\tau\circ\sigma$ on the target game~$C$.
%
%Note that the six moves has respective polarities
\end{document}